\newif\ifconf
\conftrue

\documentclass[11pt,letterpaper]{article}
\ifconf
\usepackage[notes=false,later=false]{dtrt}
\else
\usepackage[draft,notes=true,later=false]{dtrt}
\fi

\usepackage{fullpage}

\usepackage{amssymb,amsfonts,amsmath}

\usepackage{amsthm}

\usepackage{times}
\usepackage{microtype}
\usepackage{tikz}
  \usetikzlibrary{matrix,tikzmark}
\usepackage{graphicx}
\usepackage{verbatim}
\usepackage{array}
\usepackage{arydshln}
\usepackage{multirow}
\usepackage{paralist}
\usepackage{enumitem}
  \setlist[itemize]{leftmargin=*}
  \setlist[enumerate]{leftmargin=*}
\usepackage[capitalise]{cleveref}
  \crefname{step}{Step}{Steps}
  \crefname{item}{Item}{Items}
  \crefname{condition}{Condition}{Conditions}
  \crefname{enumi}{Step}{Steps}
  \crefname{case}{Case}{Cases}
  \creflabelformat{equation}{#2\textup{#1}#3}
\usepackage{dsfont}
\usepackage{url}
\usepackage{braket}
\usepackage{mathrsfs}
\usepackage{soul}
\usepackage{mdframed}
\usepackage{mathtools}
\usepackage{bbm}
\usepackage{setspace}
\usepackage{tabu}
\usepackage{multirow}
\usepackage{adjustbox}
\usepackage{booktabs}
\usepackage[binary-units=true]{siunitx}
\usepackage[margin=4mm,small,labelfont=bf]{caption}
\usepackage{centernot}
\usepackage{environ}
\usepackage{bbm}

\usepackage{algorithm}
\usepackage[noend]{algorithmic}
\usepackage{thmtools}
\usepackage{thm-restate}

\floatname{algorithm}{Reduction}
\crefname{algorithm}{reduction}{reductions}
\Crefname{algorithm}{Reduction}{Reductions}

\usepackage[giveninits,backend=bibtex,style=alphabetic,maxnames=6,maxalphanames=6]{biblatex}

\bibliography{title-long,references}

\theoremstyle{plain} 
\newtheorem{itheorem}{Theorem}
\newtheorem{ilemma}{Lemma}
\newtheorem{theorem}{Theorem}[section]
\newtheorem{lemma}[theorem]{Lemma}
\newtheorem{proposition}[theorem]{Proposition}
\newtheorem{claim}[theorem]{Claim}
\newtheorem{definition}[theorem]{Definition}
\newtheorem{corollary}[theorem]{Corollary}

\newtheorem{problem}{Problem}

\theoremstyle{definition} 
\newtheorem{construction}[theorem]{Construction}
\newtheorem{fact}[theorem]{Fact}

\newtheorem{remark}[theorem]{Remark}
\newtheorem{example}[theorem]{Example}

\theoremstyle{remark} %

\newcommand{\N}{\mathbb{N}}
\newcommand{\Field}{\mathbb{F}}
\newcommand{\Intersect}{\cap}
\newcommand{\Bits}{\{0,1\}}

\newcommand{\Code}{C}
\newcommand{\Dual}[1]{#1^\perp}
\newcommand{\DualCode}{\Dual\Code}

\newcommand{\Domain}{D}
\newcommand{\Subdomain}{I}
\newcommand{\SubdomainElt}{i}
\newcommand{\Codeword}{w}
\newcommand{\BlockLength}{\ell}
\newcommand{\Constraint}{z}
\newcommand{\ReedMuller}{\mathrm{RM}}
\newcommand{\SigReedMuller}{\Sigma\ReedMuller}

\newcommand{\MessageLength}{k}

\newcommand{\HammingDist}{\Delta}
\newcommand{\Message}{m}

\newcommand{\Enc}{\mathsf{ENC}}
\newcommand{\CodeSimulator}{\mathcal{S}}
\newcommand{\SigRMEnc}[2]{\Enc_{\SigmaRM(#1,#2)}}
\newcommand{\RMEnc}[2]{\Enc_{\ReedMuller(#1,#2)}}

\newcommand{\SumCode}[2]{\Sigma_{#1} #2}
\newcommand{\SumWord}[2]{\Sigma_{#1}[#2]}
\newcommand{\SigmaRM}{\SigReedMuller}

\newcommand{\Generator}{G}
\newcommand{\ParityCheck}{H}
\newcommand{\ZeroPoly}{\tilde{Z}}
\newcommand{\ProdPoly}{\vec{R}}

\newcommand{\Intset}{\Subdomain'}

\newcommand{\SumDomain}{\mathcal{A}}
\newcommand{\SumCompletion}{\bar{\SumDomain}}
\newcommand{\SumSet}{A}
\newcommand{\rev}{\mathsf{rev}}
\newcommand{\PrefixFree}{\mathsf{PrefixFree}}
\newcommand{\SymSets}{\mathsf{SymSets}}

\newcommand{\revvec}[1]{\vec{#1}_{\rev}}
\newcommand{\AntiSym}{\mathrm{AntiSym}}
\newcommand{\SigmaAntiSym}{\Sigma\AntiSym}

\newcommand{\Alphabet}{\Sigma}
\newcommand{\Prover}{\mathcal{P}}
\newcommand{\Verifier}{\mathcal{V}}
\newcommand{\MalVerifier}{\Verifier^*}

\newcommand{\Instance}{x}

\newcommand{\Proof}{\pi}

\newcommand{\MalProof}{\pi^*}

\newcommand{\Simulator}{\mathsf{Sim}}
\newcommand{\View}{\mathsf{View}}

\newcommand{\EncIn}{\Enc_{\text{in}}}
\newcommand{\EncOut}{\Enc_{\text{out}}}
\newcommand{\CLIn}{\CL_{\text{in}}}
\newcommand{\CLOut}{\CL_{\text{out}}}
\newcommand{\DomainIn}{\Domain_{\text{in}}}
\newcommand{\DomainOut}{\Domain_{\text{out}}}
\newcommand{\QueryIn}{\ell_{\text{in}}}
\newcommand{\QueryOut}{\ell_{\text{out}}}
\newcommand{\ConstraintsIn}{Z_{\text{in}}}
\newcommand{\ConstraintsOut}{Z_{\text{out}}}
\newcommand{\SubmessageIn}{R_{\text{in}}}
\newcommand{\SubmessageOut}{R_{\text{out}}}
\newcommand{\RandOut}{\mu_{\text{out}}}
\newcommand{\RandIn}{\mu_{\text{in}}}
\newcommand{\CD}{\mathsf{CD}}
\newcommand{\Uniform}{\mathcal{U}}
\newcommand{\Randomness}{\mu}

\newcommand{\TimeBound}{t}
\newcommand{\NumQuery}{\ell}
\newcommand{\QueryAnsSet}{T}
\newcommand{\Query}{\alpha}

\newcommand{\MsgSpace}{\mathcal{M}}
\newcommand{\CodeSpace}{\mathcal{C}}

\newcommand{\SearchSet}{S}
\newcommand{\SearchElt}{s}
\newcommand{\Accept}{\textsf{yes}}
\newcommand{\Reject}{\textsf{no}}
\newcommand{\SearchSetDim}{m}

\newcommand{\OutputSet}{T}

\newcommand{\SearchPath}{t}
\newcommand{\SearchSpace}{M}
\newcommand{\SearchAlg}{\mathrm{Search}}
\newcommand{\SubsetField}{S}
\newcommand{\SubsetTuple}{\mathcal{S}}
\newcommand{\HRoot}{z}
\newcommand{\CL}{\mathsf{CL}}

\newcommand{\SubConstraintMatrix}{Z}
\newcommand{\ConstraintMatrix}{Z'}
\newcommand{\CoefficientMatrix}{A_{\Subdomain}}
\newcommand{\BasisVecs}[1]{b_1, \dots, b_{#1}}
\newcommand{\Basis}{\mathcal{B}}
\newcommand{\BasisElt}{b}
\newcommand{\ZeroMatrix}{\mathbf{0}}
\newcommand{\VanishingFunc}{\mathcal{Z}}

\newcommand{\Degree}{d}
\newcommand{\NumVars}{m}
\newcommand{\Subcube}{H}
\newcommand{\SumVal}{\gamma}
\newcommand{\RandSumVal}{\kappa}
\newcommand{\SCPoly}{F}
\newcommand{\Polys}[3]{#1^{\leq #2}[X_{1},\ldots,X_{#3}]}
\newcommand{\RandPoly}{Q}

\newcommand{\RandLDPoly}{T}
\newcommand{\HPt}{w}
\newcommand{\HptB}{x}
\newcommand{\UndetPoly}{p}
\newcommand{\UnconPoly}{q}

\newcommand{\LagrangePoly}[1]{L_{#1}}

\newcommand{\Point}{w}

\newcommand{\ZCode}[2]{\mathbf{Z}_{#1}(#2)}
\newcommand{\ZVec}{\vec{0}}

\newcommand{\SmallElt}{u}
\newcommand{\SmallSubdomain}{U}
\newcommand{\BigElt}{v}
\newcommand{\BigSubdomain}{V}
\newcommand{\ExtElt}{x}

\newcommand{\VLen}[1]{\Length(\vec #1)}

\newcommand{\Image}{\mathrm{Im}}

\newcommand{\Length}{\ell}

\DeclareMathOperator{\colspan}{colspan}
\DeclareMathOperator{\Span}{span}

\newcommand{\Language}{\mathcal{L}}

\newcommand{\OracleTable}{\QueryAnsSet}

\newcommand{\MsgDomain}{H}

\newcommand{\Oracle}{\mathcal{O}}

\newcommand{\tdash}{\text{-}}
\newcommand{\NEXP}{\mathsf{NEXP}}
\newcommand{\NP}{\mathsf{NP}}
\newcommand{\BPP}{\mathsf{BPP}}
\newcommand{\SharpP}{\#\mathsf{P}}
\newcommand{\PZK}{\mathsf{PZK}}
\newcommand{\SZK}{\mathsf{SZK}}
\newcommand{\CZK}{\mathsf{CZK}}
\newcommand{\PCP}{\mathsf{PCP}}
\newcommand{\PZKPCP}{\PZK\tdash\PCP}

\newcommand{\SharpSAT}{\#\mathsf{SAT}}

\newcommand{\Union}{\cup}

\newcommand{\Poly}{\mathrm{poly}}
\newcommand{\supp}{\mathsf{supp}}
\newcommand{\eqdef}{\ {:=} \ }
\newcommand{\Alg}{\mathcal{A}}

\newcommand{\ConstrucSpacing}{\vspace{2mm}\newline}
\newcommand{\defemph}[1]{\textbf{\emph{#1}}}

\title{Perfect Zero-Knowledge PCPs for \texorpdfstring{$\SharpP$}{\#P}}
\author{Tom Gur \thanks{University of Cambridge. Email: \texttt{tom.gur@cl.cam.ac.uk}.}
\and Jack O'Connor\thanks{University of Cambridge. Email: \texttt{jo496@cam.ac.uk}.}
\and Nicholas Spooner\thanks{University of Warwick/NYU. Email: \texttt{nicholas.spooner@warwick.ac.uk}.}}
\date{\today}
\date{}

\begin{document}

\maketitle

\begin{abstract}
    \noindent We construct perfect zero-knowledge probabilistically checkable proofs (PZK-PCPs) for every language in $\SharpP$. This is the first construction of a PZK-PCP for any language outside $\BPP$. Furthermore, unlike previous constructions of (statistical) zero-knowledge PCPs, our construction simultaneously achieves non-adaptivity and zero knowledge against arbitrary (adaptive) polynomial-time malicious verifiers.

    Our construction consists of a novel masked sumcheck PCP, which uses the combinatorial nullstellensatz to obtain antisymmetric structure within the hypercube and randomness outside of it.  To prove zero knowledge, we introduce the notion of locally simulatable encodings: randomised encodings in which every local view of the encoding can be efficiently sampled given a local view of the message. We show that the code arising from the sumcheck protocol (the Reed--Muller code augmented with subcube sums) admits a locally simulatable encoding. This reduces the algebraic problem of simulating our masked sumcheck to a combinatorial property of antisymmetric functions.
\end{abstract}

\newpage
\tableofcontents
\newpage

\section{Introduction}
The notion of a \emph{zero-knowledge (ZK) proof}, a proof that conveys no information except the truth of a statement, is one of the most influential ideas in cryptography and complexity theory of the past four decades.

Zero knowledge was originally defined by Goldwasser, Micali and Rackoff \cite{GoldwasserMR89} in the context of interactive proofs (IPs). The deep and beautiful insight in this work is that it is possible to rigorously \emph{prove} that an interaction does not convey information, by exhibiting an efficient algorithm called the \emph{simulator} which can generate the distribution of protocol transcripts without interacting with the prover. In that work, the authors identify three different notions of zero knowledge, depending on the quality of the simulation. These are: \begin{inparaenum}[(1)]
 		\item \emph{perfect} zero knowledge (PZK), where the simulator's distribution is \emph{identical} to the real distribution of transcripts;
 		\item \emph{statistical} zero knowledge (SZK), where the distributions are inverse-exponentially close; and
 		\item \emph{computational} zero knowledge (CZK), where the distributions are computationally indistinguishable.
 	\end{inparaenum}

This hierarchy led naturally to the study of the complexity classes $\PZK,\SZK,\CZK$ of languages admitting ZK interactive proofs. Seminal results show that $\PZK$ contains interesting ``hard'' languages, including quadratic residuosity and nonresiduosity \cite{GoldwasserMR89} and graph isomorphism and nonisomorphism \cite{GoldreichMW91}. Despite this, the structure of the class $\PZK$, and its relation to other complexity classes, remains poorly understood. In light of this difficulty, the study of zero knowledge has followed two main routes: \begin{inparaenum}[(1)]
 		\item studying the ``relaxed'' notions of $\SZK$ and $\CZK$, and
 		\item studying zero knowledge in other models.
 	\end{inparaenum}

The former line of work has proved highly fruitful. A seminal result of \cite{GoldreichMW91} showed that $\CZK = \mathsf{IP} = \mathsf{PSPACE}$, which launched the cryptographic study of ZK proofs, yielding a plethora of theoretical and practical results \cite{thaler2022proofs}. The study of $\SZK$ has revealed that this class has a rich structure, and deep connections within complexity theory (cf.\ \cite{vadhan1999study}).

The present work lies along the second route, also hugely influential. The seminal work of Ben-Or, Goldwasser, Kilian, and Wigderson \cite{BenOrGKW88} introduced the model of multi-prover interactive proofs (MIP) in order to achieve perfect zero-knowledge without any computational assumptions \cite{LapidotS95, DworkFKNS92}. This in turn inspired some of the most important models and results in contemporary theoretical computer science, including entangled-prover interactive proofs (MIP*), interactive oracle proofs (IOPs), and most notably, probabilistically checkable proofs (PCPs). Indeed, the celebrated PCP theorem \cite{ALMSS92,AroraS03,Din07} is widely recognised as one of the most important achievements of modern complexity theory \cite{aharonov2013guest}.

\parhead{Zero-knowledge PCPs.}
Recall that a PCP is a proof which can be verified, with high probability, by a verifier that only makes a small number of queries (even $O(1)$) to the proof (cf. \cite{Arora:2009:CCM:1540612}). A zero-knowledge PCP is a \emph{randomised} proof which, in addition to being locally checkable, satisfies a zero knowledge condition: the view of any efficient (malicious) verifier can be efficiently simulated (see \cite{Weiss22} for a survey). Similarly to the IP case, we can distinguish PCPs with perfect, statistical and computational zero knowledge. Note that there appears to be no formal relationship between ZK-IPs and ZK-PCPs: standard transformations from IP to PCP spoil zero knowledge.

The first \emph{statistical} zero-knowledge PCPs appeared in the seminal work of Kilian, Petrank and Tardos \cite{KilianPT97}. Later works \cite{IshaiMS12,IshaiW14} simplified this construction, and extended it to PCPs of proximity. These constructions operate via a 2-step compilation of (non-zero-knowledge) PCP: first the PCP is endowed with honest-verifier statistical zero-knowledge (HVSZK), and then the latter property is boosted into full statistical zero knowledge by ``forcing'' the malicious verifier's query pattern to be similarly distributed to that of the honest verifier, using an object called a locking scheme. Kilian, Petrank and Tardos use this transformation to prove that $\NEXP$ has SZK-PCPs (with a polynomial-time verifier) that are zero knowledge against all polynomial-time malicious verifiers; i.e., $\SZK\text{-}\PCP[\Poly,\Poly] = \NEXP$.

Unfortunately, locking schemes have two inherent drawbacks: they require the honest PCP verifier to be adaptive and they cannot achieve perfect zero knowledge. Another line of work, motivated by cryptographic applications, focuses on obtaining SZK-PCPs for NP with a \emph{non-adaptive} honest verifier from leakage resilience. These results come with caveats, achieving either a weaker notion of zero knowledge known as \emph{witness indistinguishability} \cite{IshaiWY16}, or simulation against adversaries making only quadratically many more queries than the honest verifier \cite{HazayVW22}.

The complexity landscape of perfect zero knowledge is subtle. As discussed above, we know that $\mathsf{PZK\text{-}MIP} = \mathsf{MIP} = \NEXP$ \cite{BenOrGKW88}, where $\mathsf{MIP}$ is the class of languages with a multi-prover interactive proof. The quantum analogue of this result, $\mathsf{PZK\text{-}MIP}^* = \mathsf{MIP}^* = \mathsf{RE}$, is also known to hold \cite{ChiesaFGS22,GriloSY19}. We know a similar result for interactive PCPs (IPCPs) \cite{KalaiR08}, an interactive generalisation of PCPs (and special case of IOPs): $\mathsf{PZK\text{-}IPCP} = \mathsf{IPCP} = \mathsf{NEXP}$ \cite{ChiesaFS17}. For IPs, the picture is very different. We know that $\PZK \subseteq \SZK \subseteq \mathsf{AM} \cap \mathsf{coAM}$ \cite{Fortnow87,AielloH91}, and so it is unlikely that $\PZK = \mathsf{IP}$ ($= \mathsf{PSPACE}$), or even $\NP \subseteq \PZK$. It is unknown whether $\PZK = \SZK$; indeed, it was recently shown that there is an oracle relative to which this equality does not hold \cite{BoulandCHTV20}.

 In contrast, \emph{nothing at all} is known about the class $\PZKPCP$ except for the trivial inclusion $\mathsf{BPP} \subseteq \PZKPCP$. In particular, the following key question in the study of perfect zero knowledge remains open:
\begin{center}
    \emph{Do perfect zero-knowledge PCPs exist for \textbf{any} non-trivial language?}
\end{center}

\subsection{Our results}
In this work, we give strong positive answer to this question, showing that there exist perfect zero-knowledge PCPs (PZK-PCPs) for all of $\SharpP$.

\begin{itheorem}[Informally stated, see \cref{thm:main}]
\label{thm:main_informal}
    $\SharpP \subseteq \PZKPCP[\Poly,\Poly]$.
\end{itheorem}

We prove \cref{thm:main_informal} by constructing a PZK-PCP for the $\SharpP$-complete language
\[ \SharpSAT = \{ (\Phi,N) : \text{$\Phi$ is a CNF, } \sum_{x \in \Bits^n} \Phi(x) = N \} . \]
This is the first construction of a PZK-PCP for a language (believed to be) outside $\mathsf{BPP}$. Furthermore, unlike previous constructions of zero-knowledge PCPs, our construction simultaneously achieves \emph{non-adaptivity} for the honest verifier and zero knowledge against arbitrary (adaptive) polynomial-time malicious verifiers. We stress that \cref{thm:main_informal} is unconditional and does not rely on any cryptographic assumptions. 

\begin{remark}
    As with its IP counterpart $\PZK$, $\PZKPCP$ is a class of \emph{decision} problems, and so the inclusion $\SharpP \subseteq \PZKPCP$ refers to the decision version of $\SharpP$ (for which $\SharpSAT$ is complete). This class contains $\mathsf{UP}$ and $\mathsf{coNP}$ (in a natural way) but may not contain $\mathsf{NP}$. Toda's theorem, that $\mathsf{PH} \subseteq \mathsf{P}^{\SharpP}$, does not directly imply that $\mathsf{PH} \subseteq \PZKPCP$, as the $\SharpP$ oracle in this inclusion is a function oracle.
\end{remark}

On the way to proving \cref{thm:main_informal}, we solve the following general algebraic-algorithmic problem, which we consider to be of independent interest.

\begin{problem}[Local simulation of low-degree extensions]
	\label{problem:ldes}
	Given oracle access to a function $f \colon \Bits^\NumVars \to \Field$, where $\Field$ is a finite field, efficiently simulate oracle access to a uniformly random \defemph{degree-$d$ extension} of $f$; that is, a function $\tilde{f}$ drawn uniformly at random from the set of $\NumVars$-variate polynomials over $\Field$ of individual degree $d$ that agree with $f$ on $\Bits^\NumVars$.
\end{problem}

By ``efficient'' we mean polynomial in $\NumVars$, $\Degree$ and $\log |\Field|$. Chen et al. \cite{ChenCGOS23} give a \emph{query-efficient} simulator for $d \geq 2$, based on observations of Aaronson and Wigderson \cite{AaronsonW09}. (For $d = 1$, there is a query lower bound of $2^\NumVars$ \cite{JumaKRS09}.) In this work we give a \emph{computationally} efficient simulator for $d \geq 2$. 

\begin{itheorem}[informal, see \cref{thm:rm-constraint-detector}]
	There is an efficient algorithm solving \cref{problem:ldes} for $d \geq 2$.
\end{itheorem}

\subsection{Acknowledgements}
The authors thank Alessandro Chiesa, Hendrik Waldner and Arantxa Zapico for helpful discussions. TG is supported by UKRI Future Leaders Fellowship MR/S031545/1, EPSRC New Horizons Grant EP/X018180/1, and EPSRC RoaRQ Grant EP/W032635/1.

\section{Techniques}
We start by outlining how to construct (non-ZK) PCPs for $\SharpSAT$ via the sumcheck protocol \cite{LundFKN92}. Recall that the prover and verifier in the sumcheck protocol have oracle access to an $\NumVars$-variate polynomial $P$ (the arithmetisation of $\Phi$) of individual degree $\Degree$ over a finite field $\Field$, and the prover wishes to convince the verifier that $\sum_{\vec a \in \{0,1\}^m} P(\vec a) = \gamma$ for some $\gamma \in \Field$. This is achieved via a $2\NumVars$-message protocol, in which for every $i \in [\NumVars]$,
\begin{itemize}
	\item the $2i$-th message, sent by the verifier, is a uniformly random challenge $c_i \in \Field$; and
	\item the $(2i-1)$-th message, sent by the prover, is the univariate polynomial \[g_i(X) = \sum_{\vec a \in \{0,1\}^{\NumVars-i}} P(c_1,\ldots,c_{i-1},X,a_{i+1},\ldots,a_\NumVars).\]
\end{itemize}
The verifier accepts if
\begin{inparaenum}[(a)]
	\item $g_1(0) + g_1(1) = \gamma$; 
	\item for each $1 \leq i \leq \NumVars-1$, $g_i(c_i) = g_{i+1}(0) + g_{i+1}(1)$; and
	\item $g_\NumVars(c_\NumVars) = P(c_1,\ldots,c_\NumVars)$.
\end{inparaenum}
This protocol is complete and sound for sufficiently large $\Field$. Moreover, we can ``unroll'' the interactive sumcheck protocol into an (exponentially large) PCP by writing down the prover's answers to each possible (sub-)sequence of challenges.

Unfortunately, the sumcheck PCP is \emph{clearly not zero knowledge}: even the prover's first message $g_1$ reveals information about $P$ which is $\SharpP$-hard to compute. In this overview, we will explain how to modify the sumcheck PCP in order to achieve perfect zero knowledge.

\paragraph{Zero knowledge IOPs.}
Our approach to constructing a PZK-PCP for sumcheck is inspired by the \cite{BenSassonCFGRS17} construction of a perfect zero knowledge sumcheck in the \emph{interactive oracle proof} (IOP) model. The IOP model is an interactive generalisation of PCPs, in which the prover and verifier interact across multiple rounds, with the prover sending a PCP oracle in each round.

In their IOP, the prover first sends the evaluation table of a uniformly random $\NumVars$-variate polynomial $R$ such that $\sum_{\vec a \in \{0,1\}^\NumVars} R(\vec a) = 0$, which will be used as a mask. Next, to ensure soundness in case $R$ does not sum to zero, the verifier sends a challenge $\alpha$. Finally, the prover sends a sumcheck PCP for the statement $\sum_{\vec a \in \{0,1\}^\NumVars} \alpha P(\vec a) + R(\vec a) = \alpha \gamma$.

Intuitively, this protocol does not leak much information about $P$, because $\alpha P + R$ is a \emph{uniformly random} polynomial that sums to $\alpha \gamma$. \cite{BenSassonCFGRS17} shows, using techniques from algebraic complexity theory, that this IOP is indeed perfect zero knowledge; we will make use of these techniques later in our construction.

Unfortunately, the interaction in this protocol is crucial in order to balance soundness and zero knowledge. Indeed, one could imagine ``unrolling'' this IOP into a PCP in the same way as for sumcheck itself: writing down, for each $\alpha \in \Field$, a sumcheck PCP $\Pi_\alpha$ for $\alpha P + R$. This preserves soundness, but completely breaks zero knowledge: since the sumcheck PCP is a linear function of the underlying polynomial, given any \emph{two} $\Pi_\alpha,\Pi_{\alpha'}$ for distinct $\alpha,\alpha'$, we can recover the sumcheck PCP for $P$ as $(\Pi_\alpha - \Pi_{\alpha'})/(\alpha - \alpha')$.

On the other hand, if we have the prover send only \emph{one} $\Pi_\alpha$, then soundness is lost, as the prover can easily cheat by sending $R$ that does not sum to zero. We could attempt to fix the soundness issue by having the prover additionally \emph{prove}, via another sumcheck PCP $\Pi_R$, that $\sum_{\vec a} R(\vec a) = 0$. While this would restore soundness, we once again lose zero knowledge, as the sumcheck PCP for $P$ can be recovered as $(\Pi_\alpha - \Pi_R)/\alpha$.

\subsection{Structure versus randomness}
From the above discussion, we see that the central obstacle to obtaining a zero knowledge PCP is that the prover can break soundness by sending some $R$ with $\sum_{\vec a} R(\vec a) \neq 0$, and known methods for detecting this strategy (without interaction) break zero knowledge.

Our approach is to \emph{prevent} this malicious prover strategy by choosing a polynomial $R$ so that $\sum_{\vec a} R(\vec a) = 0$ \emph{by definition}. Of course, we must take care to ensure that $R$ still hides certain information about $P$. But what kind of information? We observe that, without loss of generality, we can view the sumcheck PCP for $P+R$ as an oracle $\Pi$ such that
\begin{equation*}
	\Pi(c_1,\ldots,c_i) = \sum_{\vec a \in \{0,1\}^{\NumVars-i}} (P+R)(c_1,\ldots,c_{i},a_{i+1},\ldots,a_\NumVars)~,
\end{equation*}
for all $i \in [\NumVars]$ and $c_1,\ldots,c_i \in \Field$. Computing the ``subcube sum'' $\sum_{\vec a \in \{0,1\}^{\NumVars-i}} P(c_1,\ldots,c_{i},a_{i+1},\ldots,a_\NumVars)$ requires $2^{\NumVars-i}$ queries to $P$. Hence, in order to hope to simulate, we must ensure that $R$ hides such subcube sums for any $1 \leq i < \NumVars - O(\log \NumVars)$. So, at a minimum, it should be the case that for all such $i$, and all $c_1,\ldots,c_i \in \Field$, the subcube sum $\sum_{\vec a \in \{0,1\}^{\NumVars-i}} R(c_1,\ldots,c_{i},a_{i+1},\ldots,a_\NumVars)$ is marginally uniform.
	
In short:
(1) for soundness, we need that the \emph{full} sum of the masking polynomial $R$ over the hypercube $\{0,1\}^n$ is fixed, and
(2) for zero knowledge, other than sum over the hypercube, we would like $R$ to be distributed as random as possible.

We start by ensuring that any \emph{partial} subcube sum of $R$ is distributed uniformly at random. Our key observation here is that while the full sum is \emph{invariant} under any reordering of the variables $X_1,\ldots,X_\NumVars$, any nontrivial subcube sum ($i < m$) is \emph{not}. This leads us to the following choice for the masking polynomial:
	\begin{equation*}
		R(X_1,\ldots,X_\NumVars) \eqdef Q(X_1,X_2,\ldots,X_\NumVars) - Q(X_\NumVars,X_{\NumVars-1},\ldots,X_1)~,
	\end{equation*}
	where $Q$ is a uniformly random $\NumVars$-variate polynomial of individual degree $\Degree$.

Clearly, by the permutation invariance of the complete sum, $\sum_{\vec a \in \{0,1\}^m} R(\vec a) = 0$. On the other hand, consider the special case of a subcube sum of $R$ on the hypercube, i.e., for $(c_1,\ldots,c_i) \in \{0,1\}^i$:
\begin{equation*}
	\sum_{\vec a \in \{0,1\}^{\NumVars-i}} Q(c_1,\ldots,c_{i},a_{i+1},\ldots,a_\NumVars) - \sum_{\vec a \in \{0,1\}^{\NumVars-i}} Q(a_{\NumVars},\ldots,a_{i+1},c_{i},\ldots,c_1)~.
\end{equation*}
This expression can only be identically zero for all $Q$ if $\{ c_1 \} \times \cdots \times \{c_i\} \times \{0,1\}^{m-i} = \{0,1\}^{m-i} \times \{ c_i \} \times \cdots \times \{ c_1 \}$, which is clearly only possible if $i \in \{0,m\}$. By linearity, if this subcube sum is not identically zero for all $Q$, then it is uniform when $Q$ is uniform.

Still, compared to the IOP case, the masking polynomial $R$ has extra structure: it satisfies an ``antisymmetry'' $R(\alpha_1,\ldots,\alpha_\NumVars) = -R(\alpha_\NumVars,\ldots,\alpha_1)$ for all $(\alpha_1,\ldots,\alpha_\NumVars) \in \Field^\NumVars$. This structure means that $P + R$ is no longer uniformly random and may leak nontrivial information about $P$.

To mitigate this, we modify $R$ to make it ``as random as possible'' while preserving antisymmetry on $\Bits^\NumVars$ and low-degree structure over $\Field^\NumVars$. Our final choice of $R$ is as follows:
\begin{equation}
\label{eq:R}
	R(\vec X) = Q(\vec X) - Q(\revvec X) + \sum_{i=1}^{\NumVars} X_i (1-X_i) T_i(\vec X)~,
\end{equation}
where $Q$ is as before, $\revvec X \eqdef (X_\NumVars,\ldots,X_1)$, and each $T_i$ is a uniformly random $\NumVars$-variate polynomial where $\deg_{X_i}(T_i) = \Degree - 2$ and $\deg_{X_j}(T_i) = \Degree$ for $j \neq i$. What is this extra term? It was observed by \cite{ChenCGOS23}, via the combinatorial nullstellensatz \cite{Alon99}, that $Z(\vec X) \eqdef \sum_{i=1}^{\NumVars} X_i (1-X_i) T_i(\vec X)$ is a uniformly random polynomial subject to the condition that, for all $\vec a \in \Bits^\NumVars$, $Z(\vec a) = 0$. Adding this to $R$ retains the antisymmetric structure on the hypercube, but has the effect of ``masking out'' any structure that appears outside of $\Bits^\NumVars$; in particular, $R$ is no longer antisymmetric over all of $\Field^\NumVars$.

\paragraph{Key challenge: local simulation of combinatorial and algebraic structure.}
Showing the completeness and soundness of the PCP obtained via the foregoing approach is straightforward. Hence, it remains to prove that the perfect zero-knowledge condition holds.

Starting with the trivial case, note that our observation about subcube sums in the discussion above provides a simple strategy for simulating any \emph{single} query $(c_1,\ldots,c_i)$ \emph{within the hypercube} $\{0,1\}^i$ to $\Pi$. Alas, this clearly does not suffice, as even the honest sumcheck verifier makes multiple queries to the proof, for $(c_1,\ldots,c_i) \notin \{0,1\}^i$, Hence, we must simulate multiple queries beyond the hypercube.

Indeed, the simulation of multiple queries to the low-degree extension of the hypercube is where the key challenge arises. More specifically, the difficulty is that our masking polynomial $R$ has both \emph{algebraic} structure arising from the polynomial degree bound and \emph{combinatorial} structure arising from the antisymmeric  reordering of the variables. In order to simulate, we must not only understand both types of structure, but crucially, also \emph{how they interact}.

In the remainder of this proof overview, we explain how we design an efficient simulator that can simulate responses to any number of queries to $\Pi$ (the sumcheck PCP for $P+R$), over the entire space $\Field^{\leq \NumVars}$. Towards this end, we introduce the notion of \emph{locally-simulatable encodings}.

Loosely speaking, locally-simulatable encodings are randomised encodings in which every local view of the encoding can be efficiently sampled from a local view of the message. More precisely, we say that a randomised encoding function $\Enc \colon (\MsgDomain \to \Field) \to (\Domain \to \Field)$ is locally simulatable if there is an algorithm $\CodeSimulator$ which, given oracle access to a message $\Message \colon \MsgDomain \to \Field$ and a set $S \subseteq \Domain$, samples from the distribution of $\Enc(\Message)|_S \subseteq \Field^S$ in time polynomial in $|S|$ (which may be much less than $|\MsgDomain|$). 

The notion of  locally-simulatable encodings is tightly connected to ZKPCPs: we can view the mapping from instances to corresponding PCPs as a randomised encoding. A local simulator for this encoding is a zero knowledge simulator for the PCP system.\footnote{To handle malicious verifiers that make adaptive queries, we will actually require a stronger notion of local simulation that allows for conditional sampling; see \cref{sec:lses} for details.} In the following sections, we will provide an overview of the proofs that the combinatorial antisymmetric structure of our masking, the algebraic structure of the Reed-Muller code, and their augmentations with subcube sums that arise in the sumcheck protocol admit locally simulatable encodings.

\subsection{Combinatorial structure of antisymmetric functions}

As discussed above, the response to any \emph{single} $(c_1,\ldots,c_i) \in \Bits^i$ to $\Pi$ can be perfectly simulated. In this section, we explain how to simulate any number of queries to $\Pi$, provided that those queries lie in the set $\Bits^{\leq \NumVars} \eqdef \cup_{i=1}^\NumVars \Bits^i$. Note that evaluated over the hypercube, the last term of \cref{eq:R} cancels, hence we can think of the masking polynomials as $R(\vec X) = Q(\vec X) - Q(\revvec X)$, where $\revvec X \eqdef (X_\NumVars,\ldots,X_1)$.

Observe that, when restricted to $\Bits^\NumVars$, the distribution of $R$ is well structured; namely, it is a uniformly random element of the vector space $\AntiSym$ of ``antisymmetric functions''\footnote{This is terminology that we introduce, which is distinct from the usual notion of antisymmetry for polynomials.}, i.e., functions $f \colon \Bits^\NumVars \to \Field$ such that for all $\vec c \in \Bits^\NumVars$, $f(\vec c) = -f(\revvec c)$, where $\revvec c = (c_\NumVars,\ldots,c_1)$ is the reverse of $\vec c$. We can view the restriction of $\Pi$ to $\Bits^{\leq \NumVars}$ as the output of a randomised encoding function $\Enc_{\SigmaAntiSym} \colon (\Bits^\NumVars \to \Field) \to (\Bits^\NumVars \to \Field)$, applied to the restriction of $P$ to $\Bits^\NumVars$. The function $\Enc_{\SigmaAntiSym}$ has the following description.
\begin{itemize}
\item[] \underline{$\Enc_{\SigmaAntiSym}(f)$:}
	\begin{enumerate}[nolistsep]
		\item Choose a uniformly random $r \in \AntiSym$.
		\item Output the function $\Sigma[f+r]$, given by
			\begin{equation}
                    \label{eq:tech-sigmaf}
				\Sigma[f+r](c_1,\ldots,c_i) \eqdef \sum_{\vec a \in \Bits^{\NumVars-i}} (f+r)(c_1,\ldots,c_i,a_{i+1},\ldots,a_\NumVars)~,
			\end{equation}
	\end{enumerate}
\end{itemize}
 
 for all $i \in [m]$; i.e., the function $f+r$, augmented with all of its subcube sums.

In this perspective, our task is now to provide a local simulator for $\Enc_{\SigmaAntiSym}$. That is, for a query set $S \subseteq \Bits^{\leq \NumVars}$, we want to efficiently simulate $\Sigma[f+r]|_S$ for a uniformly random $r \in \SigmaAntiSym$, given oracle access to $f$.

Let $\SigmaAntiSym|_S$ denote the vector space $\{ \Sigma[r]|_S : r \in \AntiSym \}$, and let $B$ be a basis for its dual code $\Dual{(\SigmaAntiSym|_S)}$; that is, $v \in \SigmaAntiSym|_S$ if and only if $Bv = 0$. Then, by linearity, $\Sigma[f+r]|_S$ is distributed as a uniformly random vector $w \colon S \to \Field$ such that $Bw = B(\Sigma[f]|_S)$.

With this in mind, our local simulator $\CodeSimulator$ could work as follows.

\begin{itemize}
	\item[] \underline{$\CodeSimulator^f(S)$:}
\begin{enumerate}[nolistsep]
	\item Compute a basis $B$ for $\Dual{(\SigmaAntiSym|_S)}$.
	\item Compute $y = B(\Sigma[f]|_S)$ by querying $f$ as necessary.
	\item Output a random $w$ such that $Bw = y$.
\end{enumerate}
\end{itemize}

Correctness is straightforward. The only potential issue is efficiency: evaluating $\Sigma[f]$ at a point in $S$ may require exponentially many queries to $f$. Note, however, that if the column of $B$ corresponding to $\vec c \eqdef (c_1,\ldots,c_i)$ is all zero, then we do not need to know $\Sigma[f](c_1,\ldots,c_i)$ to compute $B(\Sigma[f]|_S)$. Hence, for efficiency, it would suffice to (efficiently) find a basis $B$ for $\Dual{(\SigmaAntiSym|_S)}$ such that $\sum_{\vec c \in \supp(B)} 2^{\NumVars - \VLen{c}} = \Poly(|S|)$, where $\supp(B)$ is the set of nonzero columns of $B$ and $\VLen(c) = i$ for $c \in \Bits^i$.

Does such a basis exist? Not exactly: for example, if $S = \{ (0), (1) \}$, then $\supp(B) = \{ (0), (1) \}$, since $\Sigma[r](0) + \Sigma[r](1) = \sum_{\vec a \in \Bits^{\NumVars-1}} r(0,\vec a) + \sum_{\vec a \in \Bits^{\NumVars-1}} r(1,\vec a) = 0$ for all $r \in \AntiSym$. However, this counterexample does not actually cause a problem for zero knowledge, since for this basis $B$, it holds that $B(\Sigma[P]|_S) = \sum_{\vec a \in \Bits^\NumVars} P(\vec a) = \gamma$, which is part of the input to the problem. Our key technical result in this section is that all possible counterexamples are essentially of this form.
\begin{ilemma}
	\label{lem:informal-antisym}
	There is a polynomial $p$ such that, for any prefix-free\footnote{A set $S \subseteq \Bits^{\leq \NumVars}$ is \emph{prefix-free} if for any $\vec c_1, \vec c_2 \in S$, if $\vec c_1$ is a prefix of $\vec c_2$ then $\vec c_1 = \vec c_2$.} $S \subseteq \Bits^{\leq \NumVars}$, there exists a basis $B$ of $\Dual{(\SigmaAntiSym|_S)}$ where each row $b_i$ of $B$ is a $0$-$1$ vector, and letting $T(b_i) \eqdef \{ (\vec c, \vec a) : b_i(\vec c) = 1, \vec a \in \Bits^{\NumVars-|\vec c|} \}$, either
	\begin{equation*}
		|T(b_i)| \leq p(|S|) \qquad \text{ or } \qquad |T(b_i)| \geq 2^{\NumVars} - p(|S|)~.
	\end{equation*}
	Moreover, $B$ can be efficiently computed from $S$.
\end{ilemma}

\noindent For the purposes of this overview, we will assume that $S$ is indeed prefix-free; in the full proof we show that this holds without loss of generality.

To gain some intuition for this result, let us consider the 2-dimensional case of antisymmetric functions $A \colon [n]^2 \to \Field$; i.e., $n \times n$ antisymmetric matrices $A$ over $\Field$. Suppose that $X \subseteq [n] \times [n]$ is a set consisting of $r$ full rows and $t$ individual entries (so $|X| = rn + t$) such that for all antisymmetric $A$, it holds that $\sum_{(i,j) \in X} a_{ij} = 0$. The latter condition implies that the indicator matrix for $X$ must be symmetric; one possible choice of $X$ is highlighted in \cref{fig:antisym}. It is straightforward to show that this symmetry implies $t \geq r \cdot (n - r)$, from which we obtain the pair of inequalities
\begin{equation*}
	r \leq \frac12 \cdot n \left(1 - \sqrt{1 - \frac{4t}{n^2}}\right) \approx \frac{t}{n} + \frac{t^2}{n^3}  \quad \text{or} \quad
	r \geq \frac12 \cdot n \left(1 + \sqrt{1 - \frac{4t}{n^2}}\right) \approx n - \frac{t}{n} - \frac{t^2}{n^3}~.
\end{equation*}
This in turn yields appropriate bounds on $|X|$: $|X| \lesssim 2t + t^2/n^2$ or $|X| \gtrsim n^2 - t^2/n^2$. The proof of \cref{lem:informal-antisym} generalises this approach to higher dimensions.

\begin{figure}
	\newcommand{\hlcolor}{yellow!40}
	\newcommand{\mathcolorbox}[2]{\text{\colorbox{#1}{$#2$}}}
	\newcommand{\mathhl}[1]{\mathcolorbox{\hlcolor}{#1}}
	\newcommand{\matrixspacer}{\vphantom{\Big|}}
	\begin{equation*}
		\left(
			\begin{array}{ccccc}
			0 & \matrixspacer \mathhl{a_{12}} & \mathhl{a_{13}} & a_{14} & a_{15} \\
			\rowcolor{\hlcolor} -a_{12} & 0 & a_{23} & a_{24} & a_{25} \\
			\rowcolor{\hlcolor} -a_{13} & -a_{23} & 0 & a_{34} & a_{35} \\
			-a_{14} & \matrixspacer \mathhl{-a_{24}} & \mathhl{-a_{34}} & 0 & a_{45} \\
			-a_{15} & \mathhl{-a_{25}} & \mathhl{-a_{35}} & -a_{45} & \mathhl{0}
			\end{array}
		\right)
	\end{equation*}
	\caption{A general antisymmetric matrix ($n = 5$), with an element of the dual highlighted ($r = 2, t = 7$).}
	\label{fig:antisym}
\end{figure}

\cref{lem:informal-antisym} suggests the following modified local simulator, which we additionally provide with the total sum $\gamma$:
\begin{itemize}
	\item[] \underline{$\CodeSimulator_{\AntiSym}^{f}(S; \gamma)$:}
\begin{enumerate}[nolistsep]
	\item Compute a basis $B$ for $\Dual{(\SigmaAntiSym|_S)}$ as in \cref{lem:informal-antisym}.
	\item For each row $b_i$ of $B$:
	\begin{enumerate}[nolistsep]
		\item if $|T(b_i)| \leq 2^{\NumVars-1}$, compute $y_i = \sum_{\vec x \in T(b_i)} f(\vec x)$;
		\item otherwise, compute $y_i = \gamma - \sum_{\vec x \in \Bits^\NumVars \setminus T(b_i)} f(\vec x)$.
	\end{enumerate}
	\item Output a random $w$ such that $Bw = y$.
\end{enumerate}
\end{itemize}
Clearly, the number of queries to $f$ required is $\min(|T(b_i)|, 2^\NumVars - |T(b_i)|)$, which is $\Poly(|S|)$ by \cref{lem:informal-antisym}; hence the overall running time of the simulator is $\Poly(|S|)$. 

\subsection{Local simulation of random low-degree extensions}
\label{sec:tech-rm-sim}
In the previous section we addressed the \emph{combinatorial} problem of simulating queries to $\Pi$ that lie in the set $\Bits^{\leq \NumVars}$, which exhibits antisymmetric strucutre. In this section we consider the \emph{algebraic} problem of simulating queries to $\Pi$ that also lie outside of the hypercube (i.e., general point queries), which exhibits pseudorandom structure. We then bring these two parts together in \cref{sec:tech-sigmarm}.

Recall that our choice of $R$ is
\begin{equation*}
	R(\vec X) = Q(\vec X) - Q(\revvec X) + \sum_{i=1}^{\NumVars} X_i (1-X_i) T_i(\vec X)~,
\end{equation*}
where $Q$ is as before, $\revvec X \eqdef (X_\NumVars,\ldots,X_1)$, and each $T_i$ is a uniformly random $\NumVars$-variate polynomial where $\deg_{X_i}(T_i) = \Degree - 2$ and $\deg_{X_j}(T_i) = \Degree$ for $j \neq i$.

\newcommand{\LD}{\mathsf{LD}}
Denote by $\ReedMuller[\Field,\NumVars,\Degree] \subseteq \Field^\NumVars \to \Field$ the \emph{Reed--Muller} code of evaluations of $\NumVars$-variate polynomials of individual degree $\Degree$ over $\Field$. For a function $f \colon \Bits^\NumVars \to \Field$, denote by $\LD[f,\Degree] \eqdef \{ \hat{f} \in \ReedMuller[\Field,\NumVars,\Degree] : \hat{f}(x) = f(x)\, \forall \vec x \in \Bits^\NumVars \}$ the set of \emph{degree-$\Degree$ extensions} of $f$. With the above modification, we can describe the sumcheck PCP for $P + R$, restricted to $\Field^\NumVars$, as the output of the following randomised encoding:

\begin{itemize}
\item[] \underline{$\Enc_{\ReedMuller\AntiSym}(P)$:}
\begin{enumerate}[nolistsep]
	\item Sample a random antisymmetric function $f \colon \Bits^\NumVars \to \Field$.
	\item Sample $R$ uniformly at random from $\LD[f,\Degree]$.
	\item Output $P+R$.
\end{enumerate}
\end{itemize}
Our problem now becomes to design a local simulator for $\Enc_{\ReedMuller\AntiSym}$.

We will do this by solving a much more general problem: \emph{we prove that random low-degree extensions are locally simulatable.} More precisely, let $\Enc_{\ReedMuller}^\Degree(f)$ be the encoding function that outputs a uniformly random element of $\LD[f,\Degree]$. We prove that $\Enc_\ReedMuller^\Degree$ has a local simulator $\CodeSimulator_{\ReedMuller}^\Degree$ for any $\Degree \geq 2$.
We can then easily obtain a simulator for $\Enc_{\ReedMuller\AntiSym}$ by composing $\CodeSimulator_{\ReedMuller}$ with the local simulator for $\AntiSym$ described in the previous section.

\parhead{Our starting point}
We will build on the ``stateful emulator'' that was constructed in \cite{ChenCGOS23} (for a particular oracle model that they introduced), which can be conceptualised as an \emph{inefficient} local simulator for $\Enc_{\ReedMuller}$.
To describe it, we first introduce some notation.

For $w \in \Bits^{\NumVars}$, we denote by $\delta_w$ the unique multilinear polynomial satisfying $\delta_w(w) = 1$ and $\delta_w(x) = 0$ for all $x \in \Bits^{\NumVars} \setminus \{ w \}$. For a set $S \subseteq \Field^{\NumVars}$, we say that $w$ is \emph{$S$-good} if there exists an $\NumVars$-variate polynomial $q_w$ of individual degree at most $d$ such that \begin{inparaenum}[(i)] \item $q_w(x) = 0$ for every $x \in \Bits^{\NumVars} \setminus \{ w \}$; \item $q_w(z) = 0$ for every $z \in S$; and \item $q_w(w) = 1$\end{inparaenum}. We say that $w$ is \emph{$S$-bad} if it is not $S$-good. Intuitively, when $\hat{f}$ is a uniformly random low-degree extension of $f$, $\hat{f}|_S$ only conveys information about $f(w)$ for $S$-bad points $w$. For example, any point in $\Bits^\NumVars \cap S$ is trivially $S$-bad. Less trivially, if $S$ consists of sufficiently many points on a curve passing through $w \in \Bits^\NumVars$, then $w$ may be $S$-bad even if $\Bits^\NumVars \cap S$ is empty.

The simulator of Chen et al. works roughly as follows:\footnote{Their analysis requires sampling $P$ in a different (and more complicated) way.}

\begin{itemize}
\item[] \underline{$\CodeSimulator^{f}_{\ReedMuller}(S)$:}

\begin{enumerate}[nolistsep]
\item \label[step]{step:bad-points} Compute the set of $S$-bad points $W$.

\item \label[step]{step:sample-p} Query $f(w)$ at all $w \in W$, sample a uniformly random polynomial $P$ such that $P(w) = f(w)$ for all $w \in W$, and output $P|_S$.
\end{enumerate}
\end{itemize}
Chen et al. prove that $\CodeSimulator_{\ReedMuller}$ is correct, and moreover that it is \emph{query-efficient}, provided that $\Degree \geq 2$: the number of queries it makes to $f$ (equal to $|W|$) is at most $|S|$ (this follows from \cite[Lemma 4.3]{AaronsonW09}).

However, they leave open the question of whether $\CodeSimulator_{\ReedMuller}$ can be made \emph{computationally} efficient. We note that \cref{step:sample-p} can be achieved efficiently using an algorithm of \cite{BenSassonCFGRS17}. The issue is in \cref{step:bad-points}; namely, it is not clear how to compute the set $W$ from $S$. Indeed, at first glance this appears to be hopeless: naively, computing whether a single $w$ is $S$-bad requires solving an exponentially large linear system, and there are $2^\NumVars$ possible choices for $w$. Nonetheless, we will present an efficient algorithm which, given $S$, outputs a list of all $S$-bad points.\footnote{Chen et al. follow a different route: roughly, they show that because their $f$ is \emph{extremely sparse}, it is very unlikely that $f(w) \neq 0$ for any $w \in W \setminus S$, and so for simulation one can pretend that $f$ is zero on $\Bits^\NumVars \setminus S$. This does not work for us for two reasons: first, a random antisymmetric function is not sparse with overwhelming probability, and second, this strategy introduces a small statistical error in simulation, which would spoil perfect zero knowledge.}

\parhead{Step 1: solving the decision problem}
Our first step is to consider the \emph{decision} variant of this problem: given a set $S$, does there \emph{exist} an $S$-bad $w$? Our key insight here is that we can relate the existence of an $S$-bad $w$ to the dimension of the vector space $\LD[Z,\Degree]|_S$, where $Z$ is the constant zero function. We show that all $w$ are $S$-good if and only if $\LD[Z,\Degree]|_S = \ReedMuller[\Field,\NumVars,\Degree]|_S$. Indeed, if every $w$ is $S$-good, then for any degree-$\Degree$ polynomial $P$, the polynomial $P - \sum_{w \in \Bits^\NumVars} P(w) \cdot q_w$ is a degree-$\Degree$ extension of $Z$ that agrees with $P$ on $S$. On the other hand, if $\LD[Z,\Degree]|_S = \ReedMuller[\Field,\NumVars,\Degree]|_S$, then for every $w$, there is a $\hat{Z}_w \in \LD[Z,\Degree]$ that agrees with $\delta_w$ on $S$; hence we can take $q_w = \delta_w - \hat{Z}_w$, and so $w$ is $S$-good.

Hence to solve the decision problem it suffices to compute (the dimensions of) $\ReedMuller[\Field,\NumVars,\Degree]|_S$ and $\LD[Z,d]|_S$. The former can be achieved efficiently (in time $\Poly(|S|,\NumVars,\Degree,\log |\Field|)$) using the \emph{succinct constraint detector} for Reed--Muller discovered by \cite{BenSassonCFGRS17}.

To compute the latter, we build on ideas introduced in \cite{ChenCGOS23}. As discussed above, in that work the authors observe that to sample a vector $v \gets \LD[Z,d]|_S$, it suffices to (lazily) sample random polynomials $T_i$, $i \in [\NumVars]$, of the appropriate degrees, and then set $v(\vec \alpha) \eqdef \sum_{i=1}^{\NumVars} \alpha_i (1-\alpha_i) T_i(\vec \alpha)$ for each $\vec \alpha \in S$. Similarly, we show how to compute a basis for $\LD[Z,d]|_S$ by combining bases for the subspaces $\mathcal{P}_i \eqdef \{ T_i|_S \in \ReedMuller[\Field,\Degree,\NumVars]|_S : \deg_{X_i}(T_i) = d-2 \}$, for $i \in [m]$. These bases can also be computed using the succinct constraint detector for Reed--Muller; see \cref{sec:rm-cl}.

\parhead{Step 2: search-to-decision reduction}
Next, we will build upon the techniques we developed in Step 1 to solve the original search problem. A natural strategy is to employ a binary search: for each $b \in \{0,1\}$, test if $\{b\} \times \Bits^{\NumVars-1}$ contains any $S$-bad points, and recurse if it does. Since the number of $S$-bad points is bounded by $|S|$, the recursion will terminate quickly. It therefore suffices to give an efficient algorithm that, given a set $A$ of the form $\{ a_1 \} \times \cdots \times \{ a_i \} \times \Bits^{\NumVars-i}$, determines whether $A$ contains any $S$-bad point.

For $A \subseteq \Bits^\NumVars$, let $\LD[Z_A,\Degree]|_S \eqdef \{ P \in \ReedMuller[\Field,\NumVars,\Degree] : P(\vec x) = 0 \, \forall \vec x \in A \}$. Our algorithm for computing a basis for $\LD[Z,\Degree]|_S$ can be straightforwardly extended to compute a basis for $\LD[Z_A,\Degree]|_S$. We can also extend the reasoning from Step 1 to show that if every $w \in A$ is $S$-good, then $\LD[Z_A,\Degree]|_S = \ReedMuller[\Field,\NumVars,\Degree]|_S$. Unfortunately, the converse does not hold: it may be that $A$ contains $S$-bad points but $\LD[Z_A,\Degree]|_S = \ReedMuller[\Field,\NumVars,\Degree]|_S$. This can happen, for example, if $S$ contains points on a curve passing through both $w \in A$ and $w' \in \Bits^\NumVars \setminus A$. In particular, the argument from Step 1 fails: for $\hat{Z}_w \in \LD[Z_A,d]$, $\delta_w - \hat{Z}_w$ is not necessarily zero on $\Bits^\NumVars \setminus A$.

To obtain our final algorithm, we will instead exploit the algebraic relationship between $\ReedMuller[\Field,\NumVars,\Degree]$ and its subcode $\ReedMuller[\Field,\NumVars,\Degree-1]$ to derive a \emph{sufficient} condition for every $w \in A$ being $S$-good. Observe that if $\LD[Z_A,\Degree-1]|_S = \ReedMuller[\Field,\NumVars,\Degree-1]|_S$, then for every $w \in A$ there is a $\hat{Z}_w \in \LD[Z_A,\Degree-1]$ such that $p_w \eqdef \delta_w - \hat{Z}_w$ satisfies $p_w(w) = 1$ and $p_w(\vec \alpha) = 0$ for all $\vec \alpha \in S$. It follows that all $w \in A$ are $S$-good, as we can set $q_w = p_w \cdot \delta_w$.

Thus, our efficient test for whether $A$ contains any $S$-bad point is as follows: compute bases for $\LD[Z_A,\Degree-1]|_S$ and $\ReedMuller[\Field,\NumVars,\Degree-1]|_S$, and output ``no'' if they are of the same dimension; otherwise output ``maybe''. By the above discussion, this test does not have any false negatives; however, there may be false positives. Provided there are not too many false positives, this does not cause a problem (we can either include them in $W$ or perform an extra test to filter them out). We bound the number of false positives by noting that if the test says ``maybe'', this in fact means that there exists an ``$S$-bad'' point in $A$ with respect to polynomials of degree $\Degree-1$. By \cite{AaronsonW09}, provided $\Degree \geq 2$, there are at most $|S|$ such points.

\subsection{Local simulation of subcube sums of LDEs}
\label{sec:tech-sigmarm}
So far, we have built a simulator that can answer arbitrary queries to our sumcheck PCP $\Pi$, provided those queries lie within the set $\{0,1\}^{\leq \NumVars} \cup \Field^{\NumVars}$, handling the combinatorial structure on the hypercube and the algebraic, pseudorandom algebraic structure outside of it. In this final part of the overview, we outline how this simulator can be extended to handle all queries; i.e., queries in the set $\Field^{\leq \NumVars}$.

To do this, it suffices to show that the code $\SigmaRM$ admits a locally-simulatable encoding, where $\SigmaRM[\Field,\NumVars,\Degree] = \{ \Sigma[P] : P \in \ReedMuller[\Field,\NumVars,\Degree] \}$ and $\Sigma[P]$ is defined as in \cref{eq:tech-sigmaf} (where $(c_1,\ldots,c_i)$ now ranges over $\Field^i$). Specifically, we show that the following encoding function is locally simulatable:
\begin{itemize}
    \item[] \underline{$\Enc_{\SigmaRM}^\Degree(\Sigma[f]):$}
\begin{enumerate}[noitemsep]
    \item Sample $\hat{f}$ uniformly at random from $\LD[f,\Degree]$.
    \item Output $\Sigma[\hat f]$.
\end{enumerate}
\end{itemize}
Note that the \emph{message} (input to $\Enc_{\SigmaRM}$) is $\Sigma[f]$, even though $\Enc_{\SigmaRM}$ operates only on $f$. This is necessary for local simulation, since individual locations in $\Sigma[\hat{f}]$ depend on partial sums of $f$ which cannot be computed from few queries to $f$ itself.

We would like to follow the strategy from the previous section: given a set $S \subseteq \Field^{\leq \NumVars}$ on which we want to simulate $\Enc_{\SigmaRM}(\Sigma[f])$, compute the set of all ``bad'' points $W \subseteq \Bits^{\leq \NumVars}$, i.e., $w \in W$ if $\Enc_{\SigmaRM}(\Sigma[f])|_S$ conveys information about $\Sigma[f](w)$. Unfortunately, unlike in the ``plain'' Reed--Muller case, it is not clear how to even bound the \emph{size} of $W$, let alone compute it.

The issue here is that a single evaluation of $\Sigma[\hat{f}](\vec \alpha)$ for $\vec \alpha \in \Field^{i}$ depends on $\hat{f}(\vec \beta)$ for every point $\vec \beta \in \vec \alpha \times \Bits^{\NumVars-i}$. Naively applying the lemma of \cite{AaronsonW09} to this set of points yields a set $W \subseteq \Bits^\NumVars$ of size $2^{\NumVars-i}$. To bound $|W|$ by a polynomial, we must therefore crucially make use of the fact that making a few queries to $\Sigma[\hat f]$ can reveal only a few (possibly large) \emph{partial sums} of $\hat{f}$ --- which we can hope to deduce from a small number of queries to $\Sigma[f]$.

To achieve this we will give a \emph{decomposition} of $\SigmaRM$ as a sequence of $\ReedMuller$ codes on different numbers of variables, ``tied together'' with constraints that enforce summation structure.
In more detail, let $T \subseteq \Field^{\leq \NumVars}$ be a set with the following special structure, which we call ``closed'': if $(s_1,\ldots,s_i) \in S$, then its ``parent'' $(s_1,\ldots,s_{i-1})$ and its ``siblings'' $(s_1,\ldots,s_{i-1},0)$ and $(s_1,\ldots,s_{i-1},1)$ are also in $S$. The set $\Bits^{\leq \NumVars}$ is closed; it is also easy to see that any set $T$ can be made closed by adding at most $3\NumVars|T|$ points. We prove the following theorem  about the structure of restrictions of $\SigmaRM$ to closed sets $T$:

\begin{theorem}[Informally stated, see \cref{theorem:sigma-rm-dual-z}]
    \label{theorem:sigma-rm-intro}
    For any closed set $T$, any constraint $z \in \Dual{(\SigmaRM[\Field,\NumVars,\Degree]|_T)}$ lies in the span of the following two types of constraint:
    \begin{itemize}
        \item \emph{summation constraints}, which ensure that $\Sigma[\hat{f}](t_1,\ldots,t_i) = \Sigma[\hat{f}](t_1,\ldots,t_i,0) + \Sigma[\hat{f}](t_1,\ldots,t_i,1)$ for $(t_1,\ldots,t_i) \in T$, $i < \NumVars$; and
        \item \emph{low-degree constraints}, which are given by, for each $i$, the $i$-variate Reed--Muller constraints on the set $T_i \eqdef T \cap \Field^i$; i.e., $\Dual{(\ReedMuller[\Field,i,\Degree]|_{T_i})}$.
    \end{itemize}
\end{theorem}
If we take $T = \Bits^{\leq \NumVars} \cup S$, then $T_i = \Bits^i \cup (S \cap \Field^i)$ (assuming $S$ is closed). Using this decomposition, we show that we can take as the ``bad'' set $W \eqdef \bigcup_{i=1}^\NumVars W_i$, where $W_i$ is the set of $(S \cap \Field^i)$-bad points in $\ReedMuller[\Field,i,\Degree]$. Each $W_i$ can be computed efficiently using the algorithm described in \cref{sec:tech-rm-sim}. 

We conclude by giving some brief intuition on how we prove \cref{theorem:sigma-rm-intro}. We show that any constraint $z$ supported on $T_i \cup T_{i+1} \cup \cdots \cup T_{\NumVars}$ can be ``flattened'' into a constraint supported on $T_i$ only. Such constraints belong to $\Dual{(\ReedMuller[\Field,i,\Degree]|_{T_i})}$. We then use a dimension-counting argument to show that these, augmented with the summation constraints, span $\Dual{(\SigmaRM[\Field,\NumVars,\Degree]|_T)}$.

\section{Preliminaries}
\label{sec:prelims} 

For a vector $\vec a \in \Alphabet^i$, we denote by $\VLen{a}$ the \emph{number of entries} in $\vec a$, i.e. $\VLen{a} = i$.
Throughout, $\Field$ is a finite field. For two vectors $x, y \in \Field^n$, we denote the dot product by $x \cdot y \eqdef \sum_{i = 1}^n x_i y_i$.
For sets $A,B$, we denote by $A \sqcup B$ the \emph{disjoint union} of $A$ and $B$.

\parhead{Algorithms} We write $\mathcal{A}^{\Proof}(\Instance)$ to denote the output of $\mathcal{A}$ when given input $\Instance$ (explicitly) and oracle access to $\Proof$. We use the abbreviation PPT to denote probabilistic polynomial-time. We generally omit the internal randomness of an algorithm from probability statements, that is, we write $\Pr[\Alg(x) = 0]$ to mean $\Pr_{r \gets \Bits^n}[\Alg(x; r) = 0]$.

\parhead{Polynomials} For $\Degree \in \N$, we write $\Polys{\Field}{\Degree}{\NumVars}$ for the ring of polynomials in $\NumVars$ variables over $\Field$ of individual degree $\Degree$. For $\vec{\Degree} = (\Degree_1,\dots,\Degree_m) \in \N^\NumVars$, we use $\Polys{\Field}{\vec{\Degree}}{\NumVars}$ to denote the ring of polynomials in $\NumVars$ variables over $\Field$ of degree $\Degree_i$ in $X_i$ for each $i \in [m]$.

\parhead{Low-degree extensions} Given a function $f\colon \SumSet_1\times\dots\times\SumSet_\NumVars \to \Field$, $\SumSet_1,\ldots,\SumSet_\NumVars \subseteq \Field$, we denote the set of all low-degree extensions of $f$ by $\LD_{\vec{\Degree}}[f] \eqdef \{p \in \Polys{\Field}{\vec{\Degree}}{\NumVars} : p(x) = f(x) ~\forall x \in \SumSet_1\times\dots\times\SumSet_m\}$.

\parhead{Lagrange polynomials}
	Let $\Field$ be a field, and let $\NumVars \in \N$. Given $\SubsetField_1,\dots \SubsetField_\NumVars \subseteq \Field$ and a point $\HPt \in \SubsetField_1\times\dots\times\SubsetField_{\NumVars}$, we define the \defemph{Lagrange interpolating polynomial for $\SubsetTuple$ and $\HPt$} to be 
	\begin{equation*}
		\LagrangePoly{\SubsetTuple, \HPt}(x) \eqdef \prod_{i=1}^{\NumVars}\prod_{\HRoot\in\SubsetField_i \setminus\{\HPt_i\}} \frac{x_i - \HRoot_i}{\HPt_i-\HRoot_i}\enspace,
	\end{equation*}
	where denote $\SubsetTuple \eqdef \SubsetField_1\times\dots\times \SubsetField_\NumVars$. $\LagrangePoly{\SubsetTuple, \HPt}$ is of individual degree $|\SubsetField_i|-1$ in its $i$-th variable.

\begin{fact}
	\label{fact:lagrange-basis}
	The set of Lagrange interpolators $\{\LagrangePoly{\SubsetTuple, \Point}\}_{\Point\in\SubsetTuple}$ form a basis for the space $\Polys{\Field}{\vec{\Degree}}{\NumVars}$ where $\vec{\Degree}\eqdef (|\SearchSet_1|-1, \dots, |\SearchSet_\NumVars|-1)$.
\end{fact}

\parhead{Vanishing functions} Given a subset $\SearchSet \subseteq \Field$ we denote the function which is vanishing on $\SearchSet$ by $\VanishingFunc_{\SearchSet}(x) \eqdef \prod_{\SearchElt \in \SearchSet}(x - \SearchElt)$.

\parhead{Vector reversal} For a vector $\vec{x} = (x_1,x_2,\ldots,x_n)$, we denote by $\revvec{x}$ the \emph{reversed} vector $(x_n,x_{n-1},\ldots,x_1)$.

\subsection{Coding theory}
Let $\Alphabet$ an alphabet and let $\BlockLength \in \N$. A \emph{code} $\Code$ is a subset $\Code \subseteq \Alphabet^\BlockLength$. Given two strings $x, y \in \Alphabet^n$, we denote the Hamming distance between $x$ and $y$ by $\HammingDist(x, y) \eqdef \left|\{i \in [n]:x_i \neq y_i\}\right|$. We say that a vector $x$ is $\varepsilon$-far from a set $S \subseteq \Alphabet^n$ if $\min_{y \in S} \HammingDist(x,y)/n \geq \varepsilon$.

\begin{definition}
    Given a function $\Codeword \colon \Domain \to \Field$ and a subset $\Subdomain \subseteq \Domain$, we denote by $\Codeword|_{\Subdomain}$ the restriction of $\Codeword$ to $\Subdomain$, that is the function $\Codeword|_{\Subdomain} \colon \Subdomain \to \Field$ such that $\Codeword(x) = \Codeword|_{\Subdomain}(x)$ for all $x \in \Subdomain$.
\end{definition}

\begin{definition}
	Given a linear code $\Code \subseteq \Field^\Domain$ and a subset $\Subdomain \subseteq \Domain$, we denote by $\Code|_{\Subdomain}$ 
	the code consisting of codewords from $\Code$ restricted to $\Subdomain$, that is 
	\begin{equation*}
		\Code|_{\Subdomain} = \{ \Codeword|_{\Subdomain} : \Codeword \in \Code \}\enspace. 
	\end{equation*}
\end{definition}

Note that $\Code|_{\Subdomain}$ is itself a linear code, as it is a linear subspace of $\Field^\Subdomain$.

\parhead{Linear codes.}
Let $\Field$ be a finite field. A code $\Code\colon\Field^\MessageLength \to \Field^\BlockLength$ is \emph{linear} if it is an $\Field$-linear map from $\Field^\MessageLength$ to $\Field^\BlockLength$. 

\parhead{Reed-Muller codes and low-degree extensions} 
The Reed-Muller (RM) code is the code consisting of evaluations of multivariate low-degree polynomials over a finite field.

\begin{definition}
    Given a field $\Field$, a positive integer $\NumVars$, and a degree vector $\vec{\Degree} = (\Degree_1, \dots, \Degree_{\NumVars}) \in \N^\NumVars$, we denote by $\ReedMuller[\Field, \NumVars, \vec{\Degree}]$ the code of evaluations over $\Field^\NumVars$ of polynomials in $\Polys{\Field}{\vec \Degree}{\NumVars}$. For $\Degree \in \N$, $\ReedMuller[\Field, \NumVars, \Degree] \eqdef \ReedMuller[\Field, \NumVars, (\Degree, \ldots, \Degree)]$.
\end{definition}

\parhead{Zero and sum codes}
A \emph{zero code} is a subcode of a linear code that is zero on a given subset.
\begin{definition}[Zero codes]
    Given a $\Code \subseteq \Field^\Domain$, and $S \subseteq \Domain$, we define the \defemph{zero code}
    \begin{equation*}
        \ZCode{S}{\Code} = \{\Codeword \in \Code : \Codeword|_S = \vec{0} \}.
    \end{equation*}
\end{definition}

\noindent A \emph{sum code} is the extension of a linear code obtained by appending all sums of entries over a given subdomain.

\begin{definition}[Sum codes]
    \label{def:sumcodes}
    For a product set $\SumDomain \eqdef \SumSet_1 \times \cdots \times \SumSet_\NumVars$, we define $\SumDomain_{i} \eqdef \SumSet_1 \times \SumSet_2 \times \cdots \times \SumSet_{i}$ and $\SumDomain_{> i} \eqdef \SumSet_{i+1} \times \cdots \times \SumSet_{\NumVars}$; $\SumDomain_{0} \eqdef \{ \bot \}$. We denote $\SumCompletion \eqdef \cup_{i=0}^\NumVars \SumDomain_{i}$.

    Given a code $\Code \subseteq \Field^\Domain$ where $\Domain = \Domain_1 \times \cdots \times \Domain_\NumVars$ and $\SumDomain \subseteq \Domain$, we define the \defemph{sum code} $\SumCode{\SumDomain}{\Code} \subseteq \Field^{\bar{\Domain}}$ to be the code consisting of all of the subcube sums over $\SumDomain$ of the codewords in $\Code$. That is,

    \[ \SumCode{\SumDomain}{\Code} \eqdef \{ \SumWord{\SumDomain}{w} : w \in \Code \} \]
    where $\SumWord{\SumDomain}{w} \colon \Field^{\Domain} \to \Field^{\bar{\Domain}}$ is given by
    \begin{equation*}
    	\SumWord{\SumDomain}{w}(\vec x) \eqdef \sum_{\vec y \in \SumDomain_{>\VLen{x}}} w(\vec x, \vec y)~.
    \end{equation*}
    If $\Domain = \SumDomain$, we will typically omit $\SumDomain$ from the notation (i.e., we write $\SumCode{}{C}$, $\SumWord{}{w}$).
\end{definition}

\subsection{Linear-algebraic claims}
We will make use of the following simple linear-algebraic algorithms and facts.
\begin{claim}
	\label{claim:basis-for-dual}
	Let $U \leq \Field^n$ be a vector space. There is a polynomial-time algorithm which, given a basis for $U$, computes a basis for $U^{\perp}$.
\end{claim}

\begin{claim}
	\label{claim:basis-for-image}
	Let $U \leq \Field^n, V \leq \Field^m$ be vector spaces. There is a polynomial-time algorithm which, given a linear map $M \in \Field^{m \times n}$ such that $\{ Mu : u \in U \} = V$ and basis $B \in \Field^{n \times k}$ for $U^{\perp}$, outputs a basis for $V^{\perp}$.
\end{claim}
\begin{proof}
	Use \cref{claim:basis-for-dual} to compute a basis $B' \in \Field^{n \times (n-k)}$ for $U$. Then compute $A = MB' \in \Field^{m \times (n-k)}$; we have that $V = \colspan(A)$ by assumption on $M$. Compute a basis for $V$ from $A$ by Gaussian elimination, and from this a basis for $\Dual{V}$ by \cref{claim:basis-for-dual}.
\end{proof}

\begin{claim}
    \label{claim:linear-trans-uniform}
    Let $T:V\to W$ be a linear map. If $v \sim \mathcal{U}(V)$, then $T(v) \sim \mathcal{U}(T(V))$.
\end{claim}

\begin{proof}
    We have that 
    \begin{align*}
        \Pr_{v\gets V}[T(v) = w] = \Pr_{v\gets V}[v \in \ker(T) + w] = \frac{|\ker(T)|}{|V|} = \frac{1}{|T(V)|}.
    \end{align*}
\end{proof}

\begin{definition}
	Let $(D, <)$ be an ordered domain. For a vector $v \colon D \to \Field$, we define $\lambda(v) \eqdef \min \{ i \in D : v(i) \neq 0 \}$, where the minimum is taken with respect to $<$. A sequence of vectors $v_1,\ldots,v_n \colon D \to \Field$ is in \emph{echelon form} if, for all $i < j$, $\lambda(v_i) < \lambda(v_j)$, and any all-zero vectors are at the end of the sequence.
\end{definition}
It is a straightforward consequence of Gauss--Jordan elimination that every subspace of $(D \to \Field)$ has a basis in echelon form. Conversely, a sequence of vectors in echelon form (not containing the all-zero vector) forms a linearly independent set.

\begin{definition}
    Let $(R, <_R)$ and $(C, <_C)$ be ordered domains and let $A \in \Field^{R\times C}$. The \defemph{leading entry} of a row $w \in \Field^C$ of $A$ is $w(\lambda(w))$, i.e., the first non-zero entry in that row, with respect to the ordering $<_C$.
\end{definition}

\begin{definition}
    Let $(R, <_R)$ and $(C, <_C)$ be ordered domains. A matrix $A \in \Field^{R\times C}$ is said to be in \defemph{reduced row echelon form} if \begin{inparaenum}[(i)]
        \item the columns of $A$, listed in increasing order by their indices with respect to $<_C$ are in row echelon form (with respect to $<_D$);
        \item all leading entries are 1; and
        \item all entries in the same column as a leading entry are 0.
    \end{inparaenum}
\end{definition}
Every matrix can be converted to reduced row echelon form via Gauss-Jordan elimination, and moreover, the resulting matrix is unique. We will often represent systems of linear equations in their matrix form. We refer to variables with a leading entry in their column of the reduced row echelon form matrix as \emph{leading variables}, and the remaining variables as \emph{free variables}.

\subsection{The combinatorial nullstellensatz}
The combinatorial nullstellensatz, due to Alon \cite{Alon99}, is a powerful tool in combinatorial number theory. It is stated as follows.
\begin{lemma}
    \label{lemma:combinatorial-nullstellensatz}
    \newcommand{\Polyn}[2]{#1[X_{1},\ldots,X_{#2}]}
	Let $\Field$ be an arbitrary field, and let $f$ be a polynomial in $\Polyn{\Field}{n}$. Let $S_1, \ldots, S_n$ be nonempty subsets of $\Field$ and define $\VanishingFunc_{S_i}(x_i) = \prod_{s\in S_i} (x_i - s)$. If $f$ vanishes over all the common zeros of $\VanishingFunc_{S_1}, \ldots, \VanishingFunc_{S_n}$ (that is, if $f(s_1, \ldots, s_n) = 0$ for all $s_i \in S_i$), then there are polynomials $h_1, \ldots, h_n \in \Polyn{\Field}{n}$ so that $\deg(h_i) \leq \deg(f) - \deg(\VanishingFunc_{S_i})$ so that
    \begin{equation*}
        f = \sum_{i=1}^n h_i \VanishingFunc_{S_i}
        \enspace.
    \end{equation*}
\end{lemma}

\section{Locally simulatable encodings}
\label{sec:lses}
In this section we introduce the notion of locally simulatable encodings. First, we define randomised encoding functions, which are randomised mappings from a message space $\MsgSpace$ into a code $\CodeSpace$.

\begin{definition}[Randomised encoding function]
    Let $\Alphabet$ be an alphabet, $S,D$ be sets, and $\MsgSpace \subseteq \Alphabet^S,\CodeSpace \subseteq \Alphabet^D$ be the \emph{message space} and \emph{code space}, respectively. A \defemph{randomised encoding function} is a random variable $\Enc$ taking values in $\MsgSpace \to \CodeSpace$.
\end{definition}

Next, we define what it means for a randomised encoding function to be locally simulatable. At a high level, we require that the (conditional) distribution of any local view of the encoding can be efficiently simulated using only a local view of the message.

\begin{definition}[Locally simulatable encoding]
	Let $\Enc\colon \MsgSpace \to \CodeSpace$ be a randomised encoding function, let $\TimeBound\colon\N \to \N$ and let $\NumQuery\colon\N\to\N$. We say that $\Enc$ is \defemph{$(\TimeBound,\NumQuery)$-locally simulatable} if there is a probabilistic algorithm $\CodeSimulator$, called a \defemph{local simulator}, which receives oracle access to a message $m \in \MsgSpace$, a query-answer set $\QueryAnsSet \in (\Domain \times \Alphabet)^n$ and a new query $\Query \in \Domain$, runs in time $\TimeBound(n)$ and makes at most $\NumQuery(n)$ queries to $m$, satisfying:
	
	\begin{equation*}
		\Pr[\CodeSimulator^m(\QueryAnsSet, \Query) = \beta] = \Pr\left[\Enc(m)_\Query = \beta ~|~ \Enc(m)_x = y ~\forall (x,y) \in \QueryAnsSet \right],
	\end{equation*}
	for all $\Message \in \MsgSpace, \Query \in \Domain, \beta \in \Alphabet$, and for all $\QueryAnsSet$ such that there exists $c \in \supp(\Enc_m)$ with $c_x = y$ for all $(x,y) \in \QueryAnsSet$.
\end{definition}

To give a sense of the definition, we consider two examples.

\begin{example}
	The randomised encoding function which maps $m \in \Field^S$, $S \subseteq \Field$, to the evaluation of a random univariate polynomial $p$ over $\Field$ of degree $2(|S|-1)$ such that $p(x) = m(x)$ for all $x \in S$, is $(\Poly(n,d,\log |\Field|),n)$-locally simulatable. Indeed, for $I \subseteq \Field$, $|I| \leq |S|-1$, $p|_I$ is distributed as $(m(x))_{x \in I \cap S} \| (\Uniform(\Field))_{x \in I \setminus S}$. For $|I| \geq |S|$, the simulator $\CodeSimulator$ can simply fully determine $m$ and directly conditionally sample $p$.
\end{example}

\begin{example}
	The encoding function which maps $m \in \Field^S$, $S \subseteq \Field$, to the evaluation of the \emph{unique} univariate polynomial $p$ over $\Field$ of degree $|S|-1$ such that $p(x) = m(x)$ for all $x \in S$, is \emph{not} $(\TimeBound,\NumQuery)$-locally simulatable unless $\NumQuery(n) \geq |S|$ for all $n > 0$. Indeed, determining $p(x)$ for $x \in \Field \setminus S$ requires knowledge of the entire message $m$.
\end{example}

\subsection{Linear randomised encoding functions}
\label{sec:linear-refs}

In this work we will focus on a special family of randomised encoding functions, for which the randomised encoding is given by appending a random vector to the message, and then applying a linear function.

\begin{definition}[Linear randomised encoding function]
    We say that $\Enc\colon \MsgSpace \to \CodeSpace$ is a \defemph{linear randomised encoding function} if $\Alphabet = \Field$ for some finite field $\Field$, and $\Enc(\Message)$ is distributed as $\Enc(\Message; \Randomness)$ where $\Randomness \sim \mathcal{U}(\Field^r)$ and $\Enc\colon \MsgSpace \times \Field^r \to \CodeSpace$ is a (deterministic) linear function (of both message and randomness).
\end{definition}
\noindent That is, we use $\Enc$ to refer to both the fixed encoding function and the resulting random variable.

Linear randomised encodings have certain structural properties that make them easier to work with. One important property of linear randomised encodings is that the distribution of $\Enc(\Message)|_I$ can be characterised as the solution space of an affine system.  An $\ell$-constraint locator is an algorithm which outputs a description of this system, defined as follows.

\begin{definition}
	\label{def:cons-locate}
	Let $\Enc\colon \MsgSpace \to \CodeSpace$ be a linear randomised encoding function, and let $\ell \colon \N \to \N$. We say that an algorithm $\CL$ is an \defemph{$\ell$-constraint locator for $\Enc$} if, for every subdomain $\Subdomain \subseteq \Domain$, it satisfies 
	\begin{equation*}
		\CL(\Subdomain) = (R, Z),
	\end{equation*}
	where $R \subseteq S$ and $Z \in \Field^{k\times (R \sqcup \Subdomain)}$ are such that $|R| = \ell(|\Subdomain|)$, and for every message $m \in \MsgSpace$ and $\vec \beta \in \Field^I$, $(m|_{R}, \vec{\beta})^T \in \ker(Z) \subseteq \Field^{R \sqcup \Subdomain}$ if and only if there exists randomness $\mu \in \Field^r$ such that:
	\begin{equation*}
		\Enc(m; \mu)|_{\Subdomain} = \vec{\beta}.
	\end{equation*}
\end{definition}

Our notion of an $\ell$-constraint locator generalises the notion of a ``constraint detector'' introduced in \cite{BenSassonCFGRS17}. A constraint detector for a linear code $\Code \subseteq \Field^\Domain$ takes as input a subdomain $\Subdomain \subseteq \Domain$ and outputs a matrix representing all linear relations that hold between codewords in $\Code|_{\Subdomain}$. Setting the randomised encoding function for our $\ell$-constraint locator to be the uniform distribution over all codewords in $\Code$, we recover the definition of a constraint detector.

\begin{definition}[Constraint detector]
	\label{def:cons-detect}
    Let $\CodeSpace \subseteq \Field^\Domain$ be a linear code. A \defemph{constraint detector for $\CodeSpace$} is a $0$-constraint locator for $\Enc_{\CodeSpace} \colon \varnothing \to \CodeSpace$, the linear randomised encoding function given by $\Enc_{\CodeSpace} \sim \Uniform(\CodeSpace)$.
\end{definition}

Given an efficient constraint locator, it is straightforward to construct a local simulator.

\begin{claim}
\label{claim:cl-implies-local-sim}
	Let $\Enc\colon \MsgSpace \to \CodeSpace$ be a linear randomised encoding with an $\NumQuery$-constraint locator running in time $\TimeBound$. Then $\Enc$ is $(\TimeBound,\NumQuery)$-locally simulatable. 
\end{claim}

\begin{mdframed}[nobreak=true]
	\begin{construction}
		A local simulator $\CodeSimulator$ for a linear randomised encoding function $\Enc\colon \Field^S \to \Field^\Domain$ given an $\ell$-constraint locator $\CL$ for $\Enc$. \ConstrucSpacing
        $\CodeSimulator_{\Enc}^{\Message}(\OracleTable, \Query)$:
		\begin{enumerate}[nolistsep]
			\item Set $(R, Z) = \CL(\supp(\OracleTable)\Union \{\Query\})$.
			\item For each $\gamma \in R$, query the message and set $\Message_\gamma \eqdef \Message(\gamma)$.
			\item \label{step:output-specific-elt}If there exists a row $z \in \Field^{R\Union \supp(\OracleTable)\Union \{\Query\}}$ of the matrix $Z$ with $z(\Query) \neq 0$, output\\ $\beta \eqdef -\frac{1}{z(\Query)}\sum_{x \in\supp(\OracleTable)}y_x z(x) + \sum_{x \in R}\Message_x z(x)$, where $(x, y_x) \in \OracleTable$.
			\item \label{step:output-random-elt} Otherwise, output $\beta \gets \Field$. 
		\end{enumerate}
	\end{construction}
\end{mdframed}

\begin{proof}
    First we introduce some notation. Denote $p_{\beta} \eqdef \Pr[\Enc(\Message)_\alpha = \beta \mid \Enc(\Message)_x = y ~\forall(x, y) \in \OracleTable]$. Denote $\Subdomain \eqdef \supp(\OracleTable) \Union\{\alpha\}$ and denote $\kappa \eqdef -\frac{1}{z(\Query)}\sum_{x \in \supp(\OracleTable)} y_x z(x) + \sum_{x \in R}\Message_x z(x)$. Define the vector $\vec{y} \in \Field^{\supp(\OracleTable)}$ to be the vector consisting of values $y_x$ for each $(x, y_x) \in \OracleTable$.

    We will break the proof into two parts, corresponding to the two possibilities for the output of the construction in \cref{step:output-specific-elt} and \cref{step:output-random-elt}. First, we show that if there exists a row $z \in \Field^{R\Union \supp(\OracleTable)\Union \{\Query\}}$ of the matrix $Z$ with $z(\Query) \neq 0$ (as in \cref{step:output-specific-elt}), then $p_\beta = 1$ if $\beta = \kappa$ and $p_\beta=0$ otherwise. If there exists such a row, then the vector $(\beta, \vec{y}, \Message|_{R}) \in \ker(Z)$ if and only if $\beta = \kappa$. By the correctness of $\CL$, $(\beta, \vec{y}, \Message|_{R}) \in \ker(Z)$ if and only if there exists randomness $\mu$ such that $\Enc(\Message; \mu)|_{\Subdomain} = (\beta, \vec{y})$. So for all randomness $\mu \in \Field^r$, $\Enc(\Message;\mu)|_{\Subdomain} = (\beta, \vec{y})$ if and only if $\beta = \kappa$; hence $p_\beta = 1$ if $\beta = \kappa$ and $p_\beta=0$ otherwise.

    Second, we show that if there does not exist a row $z$ with $z(\alpha) \neq 0$ (as in \cref{step:output-random-elt}), then $p_\beta = \frac{1}{|\Field|}$, for all $\beta \in \Field$. For each $\beta \in \Field$, consider the set of all possible choices of randomness which agree with $\OracleTable \Union (\alpha, \beta)$, that is $Q_{\beta} \eqdef \{\mu \in \Field^r : \Enc(\Message, \mu)|_{\{\alpha\}\Union\supp(\OracleTable)} = (\beta,\vec{y})\}$.  We will show that $|Q_\beta| = |Q_{\beta'}|$ for each $\beta , \beta' \in \Field$ by giving a bijection, which implies that $p_\beta = \frac{1}{|\Field|}$, since $\mu$ is chosen uniformly at random. 
    
    In this case, $(\beta, \vec{y}, \Message|_{R}) \in \ker(Z)$ for all $\beta \in \Field$. Then, by the correctness of $\CL$, for all $\beta \in \Field$, there exists randomness $\mu \in \Field^r$ such that $\Enc(\Message,\mu)|_{I} = (\beta, \vec{y})$, that is, $Q_\beta$ is non-empty for each $\beta \in \Field$. Let $\mu_\beta \in Q_\beta$ and let $\mu_{\beta'} \in Q_{\beta'}$. It is easy to verify that the map $M\colon Q_\beta \to Q_{\beta'}$ defined by $M(\mu) \eqdef \mu + (\mu_{\beta'} - \mu_{\beta})$ is bijective, completing the proof. 
\end{proof}

We prove that $\ell$-constraint locators satisfy a useful composition property. That is, given constraint locators for encodings $\EncIn$ and $\EncOut$, we can construct a constraint locator for the composed encoding $\EncOut \circ \EncIn$.

\begin{claim}
\label{claim:cl-composes}
    Let $\EncIn \colon \Field^S \to \Field^{\DomainIn}$ and $\EncOut \colon \Field^{\DomainIn} \to \Field^{\DomainOut}$ admit $\QueryIn$-constraint location and $\QueryOut$-constraint location respectively. Then $\EncOut \circ \EncIn (M; \RandIn, \RandOut) \eqdef \EncOut(\EncIn(\Message; \RandIn); \RandOut)$ admits $(\QueryIn\circ\QueryOut)$-constraint location.
\end{claim}

\begin{mdframed}[nobreak=true]
	\begin{construction}
		\label{construc:composed-local-simulator}
        A constraint locator $\CL_{*}$ for $\Enc^*$, which receives oracle access to the message $\Message$ and is given constraint locators $\CLIn$ and $\CLOut$ for $\EncIn$ and $\EncOut$ respectively. It receives as input a subdomain $\Subdomain \subseteq \DomainOut$.\ConstrucSpacing
		$\CL_{*}(\Subdomain)$:
		\begin{enumerate}[nolistsep]
            \item Compute $(\SubmessageOut, \ConstraintsOut) \eqdef \CLOut(I)$.
            \item Compute $(\SubmessageIn, \ConstraintsIn) \eqdef \CLIn(\SubmessageOut)$.

            \item Define $Z_*$ to the $(k_{\text{in}}+k_{\text{out}})\times (|\Subdomain|+|\SubmessageOut|+|\SubmessageIn|)$ matrix defined as
            \begin{equation*}
                Z_{*} \eqdef \left(
                \begin{array}{ccc}
                     \ConstraintsOut|_{\Subdomain} & \ConstraintsOut|_{\SubmessageOut} & \ZeroMatrix \\
                     \ZeroMatrix & \ConstraintsIn|_{\SubmessageOut} & \ConstraintsIn|_{\SubmessageIn} \\
                \end{array}
                \right),
            \end{equation*}
            where $k_{\text{in}}$ and $k_{\text{out}}$ are the number of rows in $\ConstraintsIn$ and $\ConstraintsOut$ respectively. 
            
            \item Using Gaussian elimination, compute a basis $\Basis$ for the space $V \eqdef \ker(Z_*)$.

            \item Compute a basis $\Basis'$ for the space $\Dual{(V|_{\Subdomain \Union \SubmessageIn})}$, using \Cref{claim:basis-for-dual} and \Cref{claim:basis-for-image}. Define $Z$ to be the matrix whose rows are the elements of $\Basis'$.
            
            \item Output $(\SubmessageIn, Z)$.
		\end{enumerate}
	\end{construction}
\end{mdframed}

\begin{proof}
Let $m \in \Field^S$, and denote $w_m \eqdef \EncIn(m;\RandIn)$ and $c_w \eqdef \EncOut(w; \RandOut)$.

By the correctness of $\CLOut$, it holds that $\vec{\alpha} \in \Field^{\Subdomain}$ satisfies $(\vec{\alpha}, w|_{\SubmessageOut}) \in \ker(\ConstraintsOut)$ if and only if there exists randomness $\RandOut$ such that for every message $w \in \Field^{\DomainIn}$, $c_w|_{\Subdomain} = \vec{\alpha}$. Similarly, we have that $\vec{\beta}\in \Field^{\SubmessageOut}$ satisfies $(\vec{\beta}, m|_{\SubmessageIn}) \in \ker(\ConstraintsIn)$ if and only if there exists randomness $\RandIn$ such that for every message $m \in \Field^S$, $w_m|_{\SubmessageOut} = \vec{\beta}$.

Then by definition of $Z_{*}$ we must have that $(\vec{\alpha}, \vec{\beta}, m|_{\SubmessageIn}) \in \ker(Z_*)$ if and only if there exists $\RandOut, \RandIn$ such that for all messages $m \in \Field^S$, it holds that $w_m|_{\SubmessageOut} = \vec{\beta}$ and $c_{w_m}|_{\Subdomain} = \vec{\alpha}$. Finally, as $Z$ is defined to be the matrix whose rows are the elements of a basis for $\Dual{(V|_{\Subdomain\Union\SubmessageIn})} = \Dual{(\ker(Z_{*}|_{\Subdomain\Union\SubmessageIn}))}$, we have that $(\vec{\alpha}, m|_{\SubmessageIn}) \in \ker(Z)$ if and only if there exists randomness $\RandIn, \RandOut$, for which $\Enc^*(m;\RandIn,\RandOut)|_{\Subdomain} = \vec{\alpha}$.
\end{proof}

Note that this claim does not seem to generalise to permit composition of arbitrary locally simulatable encodings. Indeed, here we are implicitly using the fact that a constraint locator allows efficient sampling of a random input to $\EncOut$ consistent with (the restriction to $\Subdomain$ of) a given output. For general locally simulatable encodings, such an algorithm may not exist.


\section{Constraint location for random low-degree extensions}
\label{sec:rm-cl}

In this section we construct an efficient constraint locator for random low-degree extensions of a given function. In \Cref{sec:constraint-detector-zero}, we construct an efficient constraint \emph{detector} (cf. \cref{def:cons-detect}) for polynomials which are zero on a given product set $\SubsetTuple$, i.e., for the code $\ZCode{\SubsetTuple}{\ReedMuller[\Field, \NumVars, \vec{\Degree}]}$. In \Cref{sec:rm-decision-cl}, we use this constraint detector to construct a ``decision'' constraint locator for random low-degree extensions, which accepts on inputs which are constrained, and rejects on inputs which are unconstrained. In \Cref{sec:search-to-decision}, we efficiently reduce the task of (search) constraint location to the decision version.

In order to state the main theorem of this section, we first define a randomised encoding that maps a function $f$ to a random extension of specified degree.

\begin{definition}
    Let $\RMEnc{\vec{\Degree}}{\SubsetTuple}$ to be the randomised encoding function that maps $f \colon \SubsetTuple \to \Field$, where $\SubsetTuple \eqdef \SearchSet_1\times\dots\times\SearchSet_{\NumVars}\subseteq \Field^{\NumVars}$, $\Degree_i \geq 2(|\SearchSet_i|-1)$ for all $i \in [\NumVars]$, to a random degree-$\vec{\Degree}$ extension of $f$, i.e., a uniformly random element of $\LD_{\vec \Degree}[f]$.
\end{definition}

Note the restriction that $|\Degree_i| \geq 2(|\SearchSet_i|-1)$; this is tight in the case that $|\SearchSet_1| = \cdots = |\SearchSet_\NumVars| = 2$. We prove the following theorem.

\begin{theorem}
\label{thm:rm-constraint-loc}
    There is an $n$-constraint locator for $\RMEnc{\vec{\Degree}}{\SubsetTuple}$ running in time $\Poly(\log |\Field|,\NumVars,\max_i \Degree_i,n)$.
\end{theorem}

As a corollary, we obtain a local simulator for $\RMEnc{\vec{\Degree}}{\SubsetTuple}$: that is, we obtain an algorithm that can efficiently simulate a random low-degree extension of $f$ given oracle access to $f$.

\begin{corollary}
    $\RMEnc{\vec{\Degree}}{\SubsetTuple}$ is $(\TimeBound,n)$-locally simulatable for $\TimeBound(n) = \Poly(\log |\Field|,\NumVars,\max_i \Degree_i,n)$.
\end{corollary}

\newcommand{\CC}{\mathsf{CheckConstraints}}

\subsection{A constraint detector for random low-degree extensions of the zero function}
\label{sec:constraint-detector-zero}

We construct an efficient constraint detector for $\ZCode{\SubsetTuple}{\ReedMuller[\Field, \NumVars, \vec{\Degree}]}$, where $\SubsetTuple$ is a product set, and $\ZCode{\SubsetTuple}{\cdot}$ denotes the subcode consisting of codewords which are zero on $\SubsetTuple$. We shall employ the following result, which is proved in \cite{BenSassonCFGRS17}.

\begin{theorem}
\label{thm:rm-constraint-detector}
    There is a constraint detector $\CD_{\vec{\Degree}}$ for $\ReedMuller[\Field, \NumVars, \vec{\Degree}]$, running in time $\Poly(\log|\Field|, \NumVars, \max_i \Degree_i, n)$.
\end{theorem}

We prove the following.

\begin{lemma}
    \cref{construc:rm-zero-constraint-locator} is a constraint detector for $\ZCode{\SubsetTuple}{\ReedMuller[\Field, \NumVars, \vec{\Degree}]}$, running in time $\Poly(\log|\Field|, \NumVars, \max_i \Degree_i ,n)$.
\end{lemma}

\begin{mdframed}[nobreak=true]
    \begin{construction}
    \label{construc:rm-zero-constraint-locator}
        A constraint detector for 
        $\ZCode{\SubsetTuple}{\ReedMuller[\Field, \NumVars, \vec{\Degree}]}$, where $\SubsetTuple \eqdef S_1\times\dots\times S_{\NumVars}$, given a constraint detector $\CD_{\vec{\Degree}}$ for $\ReedMuller[\Field, \NumVars, \vec{\Degree}]$. Receives as input a subdomain $\Subdomain \subseteq \Field^{\NumVars}.$ \ConstrucSpacing
        $\CD_{\ZCode{\SubsetTuple}{\ReedMuller[\Field, \NumVars, \vec{\Degree}]}}(\Subdomain)$:
        \begin{enumerate}[nolistsep]
            \item For each $i \in [\NumVars]$, compute $\SubConstraintMatrix_i \eqdef \CD_{\vec{\Degree_i}}(\Subdomain)$ where $\vec{\Degree_i} \eqdef (\Degree_1, \dots, \Degree_i-|\SearchSet_i|, \dots, \Degree_\NumVars)$.

            \item Define the block diagonal matrix            \begin{equation*}
                \ConstraintMatrix \eqdef \left(\begin{array}{cccc}
                     \SubConstraintMatrix_1 & \ZeroMatrix & \dots & \ZeroMatrix \\
                     \ZeroMatrix & \SubConstraintMatrix_2 & \dots & \ZeroMatrix \\
                     \vdots & \vdots & \ddots & \vdots \\
                     \ZeroMatrix & \ZeroMatrix & \dots & \SubConstraintMatrix_{\NumVars}
                \end{array}\right)\enspace.
            \end{equation*}

            \item \label{step:compute-basis}Using Gaussian elimination, compute a basis $\BasisVecs{k}$ for $\ker(\ConstraintMatrix)$ and define $\Basis$ to be the $k \times |\Subdomain|\cdot\NumVars$ matrix whose $i$-th row is $\BasisElt_i$.

            \item Define the $|\Subdomain|\cdot\NumVars \times |\Subdomain|$ matrix $\CoefficientMatrix$ to be the matrix whose $((\vec{\alpha}, i), \vec{\gamma})$-entry is $\VanishingFunc_{\SearchSet_i}(\alpha_i)$ if $\vec{\alpha} = \vec{\gamma}$ and $0$ otherwise (where $\VanishingFunc_{S_i}(x) \eqdef \prod_{s \in S_i}(x-s)$).
            
            \item Compute $\Generator \eqdef (\Basis \CoefficientMatrix)^T$, an $|\Subdomain|\times k$ matrix, which is a generator matrix for $\ZCode{\SubsetTuple}{\ReedMuller[\Field, \NumVars, \vec{\Degree}]}|_\Subdomain$.

            \item Using $\Generator$, compute a parity check matrix $\ParityCheck$ for $\ZCode{\SubsetTuple}{\ReedMuller[\Field, \NumVars, \vec{\Degree}]}|_\Subdomain$, and output $\ParityCheck$.
            
        \end{enumerate}
    \end{construction}
\end{mdframed}

\begin{proof}
	In order to prove that \cref{construc:rm-zero-constraint-locator} is a constraint detector for $\ZCode{\SubsetTuple}{\ReedMuller[\Field, \NumVars, \vec{\Degree}]}$, it suffices to show that $\Generator \eqdef (\Basis\CoefficientMatrix)^T$ is a valid generator matrix for $\ZCode{\SubsetTuple}{\ReedMuller[\Field, \NumVars, \vec{\Degree}]}|_\Subdomain$; i.e., $\Image(\Generator) =\ZCode{\SubsetTuple}{\ReedMuller[\Field, \NumVars, \vec{\Degree}]}|_\Subdomain$.
 
    We show that $\Image(\Generator) \subseteq \ZCode{\SubsetTuple}{\ReedMuller[\Field, \NumVars, \vec{\Degree}]}|_\Subdomain$ using the fact that $\ker(\ConstraintMatrix) = \prod_{i=1}^{\NumVars}\ReedMuller[\Field, \NumVars, \vec{\Degree}_i]|_\Subdomain$ (\cref{thm:rm-constraint-detector}) and some linear algebra. We prove the reverse inclusion via the combinatorial nullstellensatz (see \cref{lemma:combinatorial-nullstellensatz}).  

    Before we proceed, it will be very useful to write down an expression for an individual entry of $\Generator$. For $\vec{\beta} \in \Subdomain, \ell \in [k]$, the $(\vec{\beta}, \ell)$-entry of $\Generator$, which we will denote by $\Generator_{\ell}(\vec{\beta})$, is given by the expression 
    \begin{equation}
    \label{eqn:generator-entry}
        \Generator_{\ell}(\vec{\beta}) = \sum_{i=1}^{\NumVars}\BasisElt_{\ell}(i, \vec{\beta})\VanishingFunc_{\SearchSet_{i}}(\beta_i),
    \end{equation}
    where $\BasisElt_\ell$ is the $\ell$-th basis element computed in \cref{step:compute-basis} and $\VanishingFunc_{\SearchSet_i}(X) \eqdef \prod_{\SearchElt \in \SearchSet}(x - \SearchElt)$. We now proceed with the proof.

    \parhead{Correctness}
    We first show $\Image(\Generator) \subseteq \ZCode{\SubsetTuple}{\ReedMuller[\Field, \NumVars, \vec{\Degree}]}|_\Subdomain$. Let $y \in \Image(\Generator)$. Then, for some scalars $a_1, \dots, a_k \in \Field$, and for each $\vec{\beta} \in \Subdomain$, we can write 
    \begin{equation*}
        y(\vec{\beta}) = \sum_{j=1}^{k} a_j \Generator_{j}(\vec{\beta}) = \sum_{j=1}^{k} a_j \sum_{i=1}^{\NumVars}\BasisElt_{j}(i, \vec{\beta})\VanishingFunc_{\SearchSet_{i}}(\beta_i) = \sum_{i=1}^{\NumVars} \VanishingFunc_{\SearchSet_i}(\beta_i) \sum_{j=1}^{k}a_j \BasisElt_{j}(i, \vec{\beta}),
    \end{equation*}
    where the second equality follows from \cref{eqn:generator-entry}. We also know that $\BasisElt_{j}(i, \vec{X}) \in \ReedMuller[\Field, \NumVars, \vec{\Degree_i}]|_{\Subdomain}$ for each $j \in [k]$, and $i \in [\NumVars]$ because $\BasisVecs{k}$ is a basis for $\ker(\ConstraintMatrix) = \prod_{i=1}^{\NumVars}\ReedMuller[\Field, \NumVars, \vec{\Degree}_i]|_{\Subdomain}$. Then by the linearity of $\ReedMuller[\Field, \NumVars, \vec{\Degree_i}]$, we have that for each $i \in [\NumVars]$, $\sum_{j=1}^{k}a_j \BasisElt_j(i, \vec{X}) \eqdef R_i(\vec{X}) \in \ReedMuller[\Field, \NumVars, \vec{\Degree}_i]|_{\Subdomain}$. Substituting this into our expression for $y(\vec{\beta})$,
    \begin{equation*}
        y(\vec{\beta}) = \sum_{i=1}^{\NumVars}R_i(\vec{\beta}) \VanishingFunc_{\SearchSet_{i}}(\beta_i).
    \end{equation*}
    Therefore $y \in \ReedMuller[\Field, \NumVars, \vec{\Degree}]|_{\Subdomain}$. As $\VanishingFunc(\beta_i) = 0$ for all $\beta_i \in \SearchSet_i$ and $i \in [\NumVars]$, we have that $y(\vec{\beta}) = 0$ for all $\vec{\beta} \in \SubsetTuple$, so $y \in \ZCode{\SubsetTuple}{\ReedMuller[\Field, \NumVars, \vec{\Degree}]}|_{\Subdomain}$.

    Now we show $\ZCode{\SubsetTuple}{\ReedMuller[\Field, \NumVars, \vec{\Degree}]}|_\Subdomain \subseteq \Image(\Generator)$. Let $\ZeroPoly \in \ZCode{\SubsetTuple}{\ReedMuller[\Field, \NumVars, \vec{\Degree}]}|_{\Subdomain}$. Then by the combinatorial nullstellensatz (\cref{lemma:combinatorial-nullstellensatz}) there exists $\ProdPoly \eqdef (R_1, \dots, R_\NumVars) \in \prod_{i=1}^{\NumVars}\ReedMuller[\Field, \NumVars, \vec{\Degree_i}]|_{\Subdomain}$ such that 
    \begin{equation*}
        \ZeroPoly(\vec{\beta}) = \sum_{i=1}^{\NumVars}\VanishingFunc_{\SearchSet_i}(\beta_i) R_i(\vec{\beta}),
    \end{equation*}
    for all $\beta \in \Subdomain$. As $\{\BasisElt_1,\dots,\BasisElt_k\}$ is a basis for $\prod_{i=1}^{\NumVars}\ReedMuller[\Field, \NumVars, \vec{\Degree_i}]|_{\Subdomain}$, we can write
    \begin{equation*}
        R_i(\vec{\beta}) = \sum_{j=1}^{k}a_j \BasisElt_j(i, \vec{\beta}),
    \end{equation*}
    for some $a_j \in \Field$. Then substituting into the previous expression for $\ZeroPoly$ yields
    \begin{equation*}
        \ZeroPoly(\vec{\beta}) = \sum_{i=1}^{\NumVars}\VanishingFunc_{\SearchSet_i}(\beta_i) \sum_{j=1}^{k}a_j \BasisElt_j(i, \vec{\beta}) = \sum_{j=1}^{k}a_j \left(\sum_{i=1}^{\NumVars} \BasisElt_j(i, \vec{\beta}) \VanishingFunc_{\SearchSet_i}(\beta_i) \right) = \sum_{j=1}^{k}a_j \Generator_{j}(\vec{\beta}).
    \end{equation*}

    \parhead{Efficiency}
    \cref{construc:rm-zero-constraint-locator} makes $\NumVars$ calls to $\CD_{\vec{\Degree_i}}$, which runs in time $\Poly(\log |\Field|,\NumVars,\max_i \Degree_i,n)$. In the remaining steps, the algorithm performs only Gaussian elimination and matrix multiplication of matrices whose sizes are polynomial in $|\Subdomain|$ and $\NumVars$.
\end{proof}

\subsection{The decision problem}
\label{sec:rm-decision-cl}
Before stating the main result of this section we define the notion of a constraint with respect to a linear code on a subset of its domain.

\begin{definition}
	Let $\Code \subseteq \Field^{\Domain}$ be a linear code. A subset $\Subdomain \subseteq \Domain$ is \defemph{constrained with respect to $\Code$} if there exists a nonzero vector $\Constraint \in \Field^{I}$ such that, for every codeword $\Codeword\in\Code$, $\Constraint\cdot \Codeword|_{\Subdomain} = 0$ (equivalently, if there exists $\Constraint \neq 0 \in \DualCode$ with $\supp(\Constraint) \subseteq \Subdomain$); we refer to $\Constraint$ as a \defemph{constraint with respect to $\Code$ on $\Subdomain$}. We say that $\Subdomain$ is \defemph{unconstrained with respect to $\Code$} if it is not constrained with respect to $\Code$. 
\end{definition}

We construct an algorithm $\CC$ which, given as input a (large) product set $\SubsetTuple$ and a polynomial-size set $\Subdomain$, efficiently determines whether $\Subdomain \Union \SubsetTuple$ is constrained with respect to the Reed-Muller code. 

\begin{lemma}
\label{lem:check-constraints-correct}
    Let $\SubsetTuple \eqdef \SearchSet_1 \times \cdots \times \SearchSet_\NumVars$, $\Subdomain \subseteq \Field^\NumVars$ and $\vec \Degree = (\Degree_1,\ldots,\Degree_\NumVars)$ with $\Degree_i \geq |\SearchSet_i|$ for all $i \in [\NumVars]$. $\CC_{\vec \Degree}(\Subdomain, \SubsetTuple)$ outputs $\Accept$ if and only if and $\Subdomain \Union \SubsetTuple$ is constrained with respect to $\ReedMuller[\Field, \NumVars, \vec{\Degree}]$, and runs in time $\Poly(\log |\Field|,\NumVars,\max_i \Degree_i,|I|)$.
\end{lemma}

\begin{mdframed}[nobreak=true]
    \begin{construction}
    \label{construc:constraint-checker}
    An algorithm which efficiently determines whether a set $\Subdomain\Union\SubsetTuple$ is constrained with respect to the Reed--Muller code, given an efficient constraint detector $\CD_{\ZCode{\SubsetTuple}{\ReedMuller[\Field, \NumVars, \vec{\Degree}]}}(\Subdomain)$ for $\ZCode{\SubsetTuple}{\ReedMuller[\Field, \NumVars, \vec{\Degree}]}$. \ConstrucSpacing
        $\CC_{\vec{\Degree}}~(\Subdomain, \SubsetTuple)$:
        \begin{enumerate}[nolistsep]
            \item Compute $H = \CD_{\ZCode{\SubsetTuple}{\ReedMuller[\Field, \NumVars, \vec{\Degree}]}}(\Subdomain\setminus \SubsetTuple)$ (see \cref{construc:rm-zero-constraint-locator}).
            \item If $H = 0$, output $\Reject$, otherwise output $\Accept$.
        \end{enumerate}
    \end{construction}
\end{mdframed}

To prove the correctness of \cref{construc:constraint-checker}, we show an equivalence between constraint detection for $\ZCode{\SubsetTuple}{\ReedMuller}$ and deciding whether $\Subdomain\Union\SubsetTuple$ is constrained with respect to $\ReedMuller$. In order to do so, we shall require the following two properties of linear codes, whose proofs we defer to \cref{sec:unconstrained-equiv} and \cref{sec:general-zcode-constraint-implies-constraint}.

\begin{restatable}{claim}{UnconstrainedEquiv}
\label{claim:unconstrained-equiv}
	Let $\Code\subseteq \Field^\Domain$ be a linear code and let $\Subdomain \subseteq \Domain$. Then $\Subdomain$ is unconstrained with respect to $\Code$ if and only if for all $x \in \Subdomain$
	there exists a codeword $\Codeword_x \in \Code$ satisfying \begin{inparaenum}[(i)]
		\item $\Codeword_x(x) = 1$ and 
		\item $\Codeword_x(y) = 0$ for all $y \in \Subdomain\setminus\{x\}$.
	\end{inparaenum}
\end{restatable}

\begin{restatable}{lemma}{GeneralZCodeConstraint}
\label{lem:general-zcode-constraint-implies-constraint}
    Let $\Code \subseteq \Field^\Domain$ be a linear code, and let $\Subdomain,\SubsetTuple \subseteq \Domain$. Suppose that $\SubsetTuple$ is unconstrained with respect to $\Code$. Then $\Subdomain \Union \SubsetTuple$ is constrained with respect to $\Code$ if and only if  $\Subdomain \setminus \SubsetTuple$ is constrained with respect to $\ZCode{\SubsetTuple}{\Code}$.
\end{restatable}

As a corollary, we show an equivalence between constraints with respect to the general Reed--Muller code and the Reed-Muller code fixed to be zero on a product set.

\begin{corollary}
\label{lem:zcode-constraint-implies-constraint}
    Let $\NumVars \in \N, \vec{\Degree}\in \N^{\NumVars}, \Subdomain \subseteq \Field^\NumVars$, let $\SearchSet_1, \dots, \SearchSet_{\NumVars} \subseteq \Field$, be such that $|\SearchSet_i| \leq \Degree_i +1$ for all $i \in [\NumVars]$, and denote $\SubsetTuple \eqdef \SearchSet_1\times\dots\times\SearchSet_\NumVars$. Then $\Subdomain \Union \SubsetTuple$ is constrained with respect to $\ReedMuller[\Field, \NumVars, \vec{\Degree}]$ if and only if  $\Subdomain \setminus \SubsetTuple$ is constrained with respect to $\ZCode{\SubsetTuple}{\ReedMuller[\Field, \NumVars, \vec{\Degree}]}$.
\end{corollary}

\begin{proof}
    By \cref{lem:general-zcode-constraint-implies-constraint}, it suffices to show that $\SubsetTuple$ is unconstrained with respect to $\ReedMuller[\Field, \NumVars, \vec{\Degree}]$. Consider, for each $w \in \SubsetTuple$, the Lagrange polynomial $\LagrangePoly{\SubsetTuple, w}(\vec{X}) \in \Polys{\Field}{\vec{\Degree}}{\NumVars}$. As $|\SearchSet_i| \leq \Degree_i +1$ for all $i \in [\NumVars]$, the evaluation table of $\LagrangePoly{\SubsetTuple, w}(\vec{X})$ is a codeword in $\ReedMuller[\Field, \NumVars, \vec{\Degree}]$, for each $w \in \SubsetTuple$. These codewords satisfy $\LagrangePoly{\SubsetTuple, w}(w) = 1$ and $\LagrangePoly{\SubsetTuple, w}(x) = 0$ for all $x \neq w \in \SubsetTuple$, so by \cref{claim:unconstrained-equiv}, $\SubsetTuple$ is unconstrained with respect to $\ReedMuller[\Field, \NumVars, \vec{\Degree}]$.
\end{proof}

\begin{proof}[Proof of \cref{lem:check-constraints-correct}]
    The correctness follows from \cref{lem:zcode-constraint-implies-constraint} and the correctness of \cref{construc:rm-zero-constraint-locator}. The runtime follows from the fact that \cref{construc:rm-zero-constraint-locator} runs in time $\Poly(\log |\Field|,\NumVars,\max_i \Degree_i,n)$.
\end{proof}

\subsection{Search-to-decision reduction}
\label{sec:search-to-decision}
\newcommand{\NumConstrPoints}{k}

In this section we show a reduction from constraint location of random low-degree extensions (a search problem) to $\CC$. Before proving \cref{thm:rm-constraint-loc}, we define what it means for a subset $I$ of the domain $D$ of a code to \emph{determine} a point $x \in D$. We then prove some useful technical lemmas.

\begin{definition}
\label{def:Determined}
        Let $\Code \subseteq \Field^{\Domain}$ be a linear code. We say that $\Subdomain\subseteq\Domain$ \defemph{determines} $x\in\Domain$ \defemph{with respect to $\Code$} if $x\in\Subdomain$ or there exists a constraint $\Constraint$ with respect to $\Code$ on $\Subdomain \Union \{x\}$ such that $\Constraint(x) \neq 0$. 
\end{definition}

The following lemma is central to our reduction. It shows that points which are constrained with respect to the Reed--Muller code are determined with respect to another Reed--Muller code of lower degree. Thus, in our algorithm, it suffices to check the latter condition.

\begin{lemma}
\label{lem:constrained-implies-determined}
    Let $\NumVars\in \N, \vec{\Degree}\eqdef (\Degree_1,\dots,\Degree_{\NumVars})\in \N^{\NumVars}, \Subdomain \subseteq \Field^\NumVars$, let $\SearchSet_1,\dots,\SearchSet_\NumVars \subseteq \Field$ and denote $\SubsetTuple \eqdef \SearchSet_1\times\dots\times\SearchSet_{\NumVars}$. If there exists a constraint $\Constraint\colon\Subdomain\Union\SubsetTuple\to \Field$ on $\Subdomain \Union \SubsetTuple$ with respect to $\ReedMuller[\Field,\NumVars,\vec{\Degree}]$ such that for some $\HPt\in\SubsetTuple$, $\Constraint(\HPt) \neq 0$, then $\HPt$ is determined by $\Subdomain$ with respect to $\ReedMuller[\Field,\NumVars,\vec{\Degree'}]$, where $\vec{\Degree'} \eqdef (\Degree_1 - (|\SearchSet_1|-1), \dots ,\Degree_{\NumVars} - (|\SearchSet_{\NumVars}|-1))$.
\end{lemma}

\begin{proof}
    We proceed via contrapositive: if $\HPt\in\SubsetTuple$ is not determined by $\Subdomain$ with respect to $\ReedMuller[\Field, \NumVars, \vec{\Degree'}]$, then $\Subdomain\Union\SubsetTuple$ is unconstrained with respect to $\ReedMuller[\Field,\NumVars, \vec{\Degree}]$.

    Let $\HPt\in\SubsetTuple$ be not determined by $\Subdomain$ with respect to $\ReedMuller[\Field, \NumVars, \vec{\Degree'}]$. Then there exists $\UndetPoly(\vec{X}) \in \Polys{\Field}{\vec{\Degree'}}{\NumVars}$ such that 
    \begin{inparaenum}[(i)]
        \item $\UndetPoly(\HPt) = 1$ and 
        \item $\UndetPoly(\SubdomainElt) = 0$ for all $\SubdomainElt \in \Subdomain$.
    \end{inparaenum}
    Now define $\UnconPoly(\vec{X}) \eqdef \UndetPoly\cdot \LagrangePoly{\SubsetTuple,\HPt}(\vec{X}) \in \Polys{\Field}{\vec{\Degree}}{\NumVars}$, where $\LagrangePoly{\SubsetTuple,\HPt}$ denotes the Lagrange  polynomial. Note $\UnconPoly(\vec{X})$ satisfies:
    \begin{inparaenum}[(i)]
        \item $\UnconPoly(\HPt)=1$,
        \item $\UnconPoly(\HptB) = 0$ for all $\HptB\in \SubsetTuple\setminus\{\HPt\}$ and 
        \item $\UnconPoly(\SubdomainElt) = 0$ for all $\SubdomainElt \in \Subdomain$.
    \end{inparaenum}
    By \cref{claim:unconstrained-equiv} $\Subdomain\Union\SubsetTuple$ is unconstrained with respect to $\ReedMuller[\Field,\NumVars, \vec{\Degree}]$.
\end{proof}

We shall also need the following claim which defines the notion of an  \emph{interpolating set} and gives an algorithm for reducing an arbitrary subdomain of a linear code to an interpolating set. We defer the proof of this claim to \cref{sec:interpolating-sets}.

\begin{restatable}{claim}{InterpolatingSets}
\label{claim:interpolating-sets}
    Let $\Code \subseteq \Field^\Domain$ be a linear code which has an efficient constraint detector. Given a subdomain $\Subdomain \subseteq \Domain$, there exists an efficient algorithm, which computes an \defemph{interpolating set for $\Subdomain$ with respect to $\Code$}: a set $\Intset \subseteq \Subdomain$ such that: \begin{inparaenum}[(i)]
        \item $\Intset$ is unconstrained with respect to $\Code$; and
        \item for every $\SubsetTuple\subseteq\Domain$ which is unconstrained with respect to $\Code$ and satisfies $\Subdomain\Intersect\SubsetTuple = \varnothing$, if there exists a constraint $\Constraint\colon \Subdomain\Union\SubsetTuple \to \Field$ such that $\Constraint(s) \neq 0$ for some $s \in \SubsetTuple$, then there exists a constraint $\Constraint'\colon \Intset\Union\SubsetTuple\to\Field$ such that $\Constraint'(s) \neq 0$.
    \end{inparaenum} 
\end{restatable}

We will also require the following simple monotonicity property of $\CC$ (see \cref{construc:constraint-checker}).

\begin{claim}
\label{claim:cc-monotone}
If $\CC(\Subdomain, \SubsetTuple) = \Accept$ and $\SubsetTuple \subseteq \SubsetTuple'$, then $\CC(\Subdomain, \SubsetTuple') = \Accept$.
\end{claim}
\begin{proof}
    Clear from \cref{lem:check-constraints-correct}.
\end{proof}

Before we give our reduction, we give a slight generalisation of Lemma 4.2 from \cite{AaronsonW09}. Its proof is a relatively straightforward adaptation of the proof of that lemma and is deferred to \cref{sec:new-multilin-generalisation}.

\begin{restatable}{lemma}{MultilinGen}
\label{lem:new-multilin-generalisation}
    Let $\NumVars\in \N$, let $\Subdomain \subseteq \Field^\NumVars$, let $\SearchSet_1,\dots,\SearchSet_\NumVars \subseteq \Field$ and denote $\SubsetTuple \eqdef \SearchSet_1\times\dots\times\SearchSet_{\NumVars}$. Set $\vec{\Degree} \eqdef (|\SearchSet_1|-1, \dots, |\SearchSet_{\NumVars}|-1)$. There exists a set $G_{\Subdomain} \subseteq \SubsetTuple$ with $|G_{\Subdomain}| \geq |\SubsetTuple| - |\Subdomain|$, such that for all $g \in G_{\Subdomain}$ there exists a polynomial $p_g(\vec{X})\in \Polys{\Field}{\vec{\Degree}}{\NumVars}$ satisfying
    \begin{enumerate}[label=(\roman*)]
         \item $p_g(x) = 0$ for all $x\in\Subdomain\Union G_{\Subdomain}\setminus \{g\}$ and
         \item $p_g(g) = 1$.
    \end{enumerate}
\end{restatable}

We are now ready to prove the main result of this section. 

\newcommand{\NumCalls}{N}
\newcommand{\NumAcc}{a}

\begin{proof}[Proof of \cref{thm:rm-constraint-loc}]
We give a reduction from constraint location to $\CC_{\vec{\Degree}}$.

\begin{algorithm}[H]
\caption{Reduction from constraint location for 
$\RMEnc{\vec{\Degree}}{\SubsetTuple}$ to $\CC_{\vec{\Degree}}$ and a constraint detector $\CD_{\vec{\Degree}}$ for $\ReedMuller[\Field, \NumVars, \vec{\Degree}\eqdef(\Degree_1, \dots,\Degree_{\NumVars})]$.}
    Set $\vec{\Degree'}\eqdef (\Degree_1 - (|\SearchSet_1|-1), \dots ,\Degree_{\NumVars} - (|\SearchSet_{\NumVars}|-1))$.
    \vspace{0.2cm}\\
    $\SearchAlg(\Subdomain, \SearchPath, \SearchSpace\eqdef(\SearchSet_i, \dots, \SearchSet_{\NumVars}))$:
    \begin{algorithmic}
    \label{alg:search-reduction}
        \IF{$\SearchSpace = \bot$}
            \STATE Return $\SearchPath$.
        \ELSE 
            \STATE Set $\OutputSet = \varnothing$.
            \FORALL{$\SearchElt \in \SearchSet_i$}
                \IF{$\CC_{\vec{\Degree'}}(\Subdomain, t\times\{\SearchElt\}\times\SearchSet_{i+1}\times\dots\times\SearchSet_\SearchSetDim) = \Accept$}
                    \STATE Add $\SearchAlg(\Subdomain, \SearchPath \times \{\SearchElt\}, (\SearchSet_{i+1},\dots,\SearchSet_{\NumVars}))$ to $\OutputSet$.  
                \ENDIF
            \ENDFOR
            \STATE Return $\OutputSet$.
        \ENDIF
    \end{algorithmic}
    \vspace{0.2cm}

    Compute $\Intset$ to be an interpolating set for $\Subdomain$.
    \\
    Compute $\OutputSet \eqdef \SearchAlg(\Intset, \bot, (\SearchSet_1,\dots,\SearchSet_\NumVars))$.
    \\
    Output $(\OutputSet, \CD_{\vec{\Degree}}~(\Subdomain \Union\OutputSet))$.
\end{algorithm}

\parhead{Correctness}
    First, we show that if $\Subdomain\Union\SubsetTuple$ is unconstrained with respect to $\ReedMuller[\Field, \NumVars, \vec{\Degree}]$ then \Cref{alg:search-reduction} outputs $\bot$. In this case, as $\OutputSet \subseteq \SubsetTuple$, $\Subdomain \Union \OutputSet$ is unconstrained so $\CD_{\vec{\Degree}}~(\Subdomain\Union\OutputSet) = \bot$, as required. 

    Now we consider the case where $\Subdomain\Union\SubsetTuple$ is constrained with respect to $\ReedMuller[\Field, \NumVars, \vec{\Degree}]$. We show that if $\SearchElt \in S_1\times\dots\times S_{\NumVars}$ is constrained, then $\SearchElt\in\OutputSet$, where $\OutputSet$ is the output of \Cref{alg:search-reduction}. By inspection of \Cref{alg:search-reduction} and \cref{claim:cc-monotone}, we see that $\SearchElt \in \OutputSet$ if and only if  $\CC_{\vec{\Degree'}}(\Intset, \{\SearchElt_1\}\times \dots\times\{\SearchElt_{\NumVars}\}) = \Accept$. If there exists a constraint $\Constraint\colon \Subdomain\Union\SearchSet\to\Field$ with respect to $\ReedMuller[\Field, \NumVars, \vec{\Degree}]$ with $\Constraint(\SearchElt) \neq 0$, then by \cref{claim:interpolating-sets}, there exists a constraint $\Constraint'\colon\Intset\Union\SubsetTuple$ with respect to $\ReedMuller[\Field, \NumVars, \vec{\Degree}]$ with $\Constraint'(\SearchElt) \neq 0$. Further, by \cref{lem:constrained-implies-determined}, $\SearchElt$ is determined by $\Intset$ with respect to $\ReedMuller[\Field, \NumVars, \vec{\Degree'}]$.

    Lastly, $\CD_{\vec{\Degree}}~(\Subdomain\Union\OutputSet)$ will output a basis for the space of constraints on $\Subdomain\Union\OutputSet$, so the output is in the desired form.

\parhead{Efficiency}
    We show that the reduction is efficient by bounding the number of calls made to $\CC_{\vec{\Degree}}$. For $i \in [\NumVars]$, we say a point $(\SearchElt_1,\dots,\SearchElt_{i})\in \SearchSet_1\times\dots\times\SearchSet_i$ is \emph{$i$-accepting} if $\CC_{\vec{\Degree'}}(\Subdomain, \{\SearchElt_1\}\times \dots\times\{\SearchElt_i\}\times \SearchSet_{i+1}\times\dots\times \SearchSet_\SearchSetDim) = \Accept$. Note that by correctness of $\CC$ and the fact that $\Intset$ is unconstrained, this means that a point  $(\SearchElt_1,\dots,\SearchElt_{i})\in \SearchSet_1\times\dots\times\SearchSet_i$ is $i$-accepting if and only if there exists $(\SearchElt_{i+1},\dots,\SearchElt_{\NumVars}) \in \SearchSet_{i+1}\times\dots\times\SearchSet_\NumVars$ such that $(\SearchElt_1, \dots, \SearchElt_{\NumVars})$ is involved in a non-zero constraint with respect to $\ReedMuller[\Field, \NumVars, \vec{\Degree'}]$. For all $i \in [\NumVars]$, let $\NumAcc_i$ denote the number of points which are $i$-accepting. 
    
    Now we count the number of calls our reduction makes to $\CC$ in terms of $\NumAcc_1,\dots,\NumAcc_{\NumVars-1}$. Note that regardless of the value of $\Subdomain$, \Cref{alg:search-reduction} makes $|\SearchSet_1|$ calls to $\CC_{\vec{\Degree'}}$ at the beginning. Subsequently, for each $i \in [\NumVars-1]$, \Cref{alg:search-reduction} makes $\NumAcc_i|\SearchSet_{i+1}|$ many calls to $\CC_{\vec{\Degree'}}$. Overall, the total number of calls that \Cref{alg:search-reduction} makes to $\CC_{\vec{\Degree'}}$ is:

    \begin{equation*}
        \NumCalls = |\SearchSet_1| + \sum_{i = 1}^{\NumVars-1}a_i|\SearchSet_{i+1}|.
    \end{equation*}

    Now we bound the size of $\NumAcc_i$ for each $i \in [\NumVars]$. As $\Degree_i' \geq |\SearchSet_i| - 1$ for all $i \in [\NumVars]$, \cref{lem:new-multilin-generalisation} implies that there exists a set $G_{\Subdomain}$ with $|G_{\Subdomain}| \geq |\SubsetTuple| - |\Subdomain|$, such that for all $g \in G_{\Subdomain}$ there exists a polynomial $p_g(\vec{X})\in \Polys{\Field}{\vec{\Degree'}}{\NumVars}$ satisfying
    \begin{inparaenum}[(i)]
         \item $p_g(x) = 0$ for all $x\in\Subdomain\Union G_{\Subdomain}\setminus \{g\}$ and
         \item $p_g(g) = 1$.
    \end{inparaenum}
    Thus, by \cref{claim:unconstrained-equiv}, at most $|\Subdomain|$ many points in $\SubsetTuple$ are constrained by $\Subdomain$ with respect to $\ReedMuller[\Field, \NumVars, \vec{\Degree'}]$. In other words, $a_i \leq |\Subdomain|$ for all $i \in [\NumVars-1]$. Therefore 
    \begin{equation*}
        \NumCalls \leq |\SearchSet_1| + |\Subdomain|\sum_{i = 1}^{\NumVars-1}|\SearchSet_{i+1}|.
    \end{equation*}
    Thus, as each call to $\CC$ takes time $\Poly(\log |\Field|,\NumVars,\max_i \Degree_i,n)$, and $|S_i|$ is bounded by a polynomial in $\Degree_i$ for each $i$, we have that $\TimeBound(n) = \Poly(\log |\Field|,\NumVars,\max_i \Degree_i, n)$.
    
    Finally, the fact that at most $|\Subdomain|$ many points in $\SubsetTuple$ are constrained implies that $\ell(n) = n$.
\end{proof}

\newcommand{\Tower}{\bar{\SubsetTuple}}

\section{Constraint location for subcube sums of random low-degree extensions}

In this section we extend the efficient constraint locator of the previous section to support queries to subcube sums. We begin by defining the corresponding randomised encoding.

\begin{definition}
    Let $\SumDomain \eqdef \SumSet_1\times\dots\times\SumSet_{\NumVars}\subseteq \Field^\NumVars$. We define $\SigRMEnc{\vec{\Degree}}{\SumDomain}\colon \SumCode{\SumDomain}{(\SumDomain \to \Field)} \to (\Field^{\leq\NumVars} \to \Field)$ as the randomised encoding $\SigRMEnc{\vec{\Degree}}{\SumDomain}(\SumWord{\SumDomain}{F}) \sim \Uniform(\SumWord{\SumDomain}{\LD_{\vec{\Degree}}{[F]}})$, where $\Degree_i \geq |\SumSet_i|-1$ for all $i \in [\NumVars]$.
\end{definition}

That is, $\SigRMEnc{\vec{\Degree}}{\SumDomain}$ encodes a message $\SumWord{\SumDomain}{F} \in \SumCode{\SumDomain}{(\SumDomain \to \Field)}$ by first choosing a random degree-$\vec{\Degree}$ extension $\hat{F}$ of $F$, and then outputting the word $\SumWord{\SumDomain}{\hat{F}} \colon \Field^{\leq \NumVars} \to \Field$ obtained by augmenting $\hat{F}$ with partial sums over $\SumDomain$ (see \cref{def:sumcodes}). Note that the input to $\SigRMEnc{\vec{\Degree}}{\SumDomain}$ includes partial sums of $F$; these are not necessary to define the encoding but are essential to permit local simulation.

The main theorem of this section is the following.

\begin{theorem}
\label{thm:sigrm-cl}
    There is a $(t, \ell)$-constraint locator for $\SigRMEnc{\vec{\Degree}}{\SumDomain}$, where $\ell(n) = n\NumVars (\NumVars(a + 1) + 1)^2$, for $a \eqdef \max_i |\SumSet_i|$ and $t(n) = \Poly(\log|\Field|, \NumVars, \max_i \Degree_i, n)$.
\end{theorem}

We prove this theorem by giving, in \cref{sec:characterisation-sigmarm}, a local characterisation of $\SigmaRM$ in terms of the plain Reed--Muller code. We give our construction based on this characterisation in \cref{sec:construction-sigmarm}.

\subsection{Local characterisation of $\Sigma$RM}
\label{sec:characterisation-sigmarm}

In this subsection we give a local characterisation of $\SigmaRM$ in terms of the plain Reed--Muller code. We define the code $\SigmaRM$ as follows. 

\begin{definition}
    Define $\SigmaRM[\Field, \SumDomain, \vec{\Degree}] \eqdef \SumCode{\SumDomain}{\ReedMuller[\Field, \NumVars, \vec{\Degree}]} \subseteq (\Field^{\leq \NumVars} \to \Field)$.
\end{definition}

Our local characterisation will hold for certain ``nice'' subsets of $\Field^{\leq \NumVars}$ which we call ``$\SumDomain$-closed''.

\begin{definition}
    A set $S \subseteq \Field^{\leq m}$ is \defemph{$\SumDomain$-closed} if the following two conditions hold:
	\begin{enumerate}[label=(\roman*)]
        \item if $(s_1, \dots, s_{t}) \in S$, then $(s_1, \dots, s_{t-1}) \in S$; and 
		\item if $(s_1,\ldots,s_\ell) \in S$, then for all $a_\ell \in \SumSet_\ell$, $(s_1,\ldots,s_{\ell-1},a_{\ell}) \in S$.
	\end{enumerate}
    Given a set $X \subseteq \Field^{\leq \NumVars}$, we refer to the smallest $\SumDomain$-closed set containing $X$ as the \defemph{$\SumDomain$-closure} of $X$, i.e., $\hat{X}$ is the $\SumDomain$-closure of $X$ if $\hat{X}$ is $\SumDomain$-closed, $X \subseteq \hat{X}$, and for all $\SumDomain$-closed $Y \supseteq X$, it holds that $\hat{X} \subseteq Y$.
\end{definition}

Globally the $\SigmaRM$ code can be viewed as a collection of plain Reed--Muller codes on $1,2,\ldots,\NumVars$ variables, related by summation constraints. The following theorem asserts that this characterisation also holds locally: over any $\SumDomain$-closed subdomain $S$, $\Dual{(\SigmaRM|_S)}$ is spanned by summation constraints $\{z_{\vec s}\}_{\vec s \in S^*}$ and the local constraints on a collection of Reed-Muller codes. For technical reasons, we show this with respect to both the ``plain'' $\SigmaRM$ code and the subcode $\ZCode{\SumDomain}{\SigmaRM}$ of encodings of the zero word.

\newcommand{\ZeroSets}{X}

\begin{theorem}
	\label{theorem:sigma-rm-dual-z}
	Let $\SumSet_1,\ldots,\SumSet_\NumVars \subseteq \Field$, $\SumDomain \eqdef \SumSet_1 \times \cdots \times \SumSet_\NumVars$ and let $S$ be an $\SumDomain$-closed set. Then for $\ZeroSets \in \{\varnothing, \SumDomain\}$, we have 
	\begin{equation*}
		\Dual{(\ZCode{\ZeroSets}{\SigmaRM[\Field,\SumDomain,\vec{\Degree}]}|_{S})} = \Span\big(\{z_{\vec s}\}_{\vec s \in S^*} \cup \bigcup_{i=1}^{\NumVars} \Dual{(\ZCode{\ZeroSets_{i}}{\ReedMuller[\Field,i,(\Degree_1,\ldots,\Degree_i)]}|_{S_i})}\big)
	\end{equation*}
	where $\ZeroSets_i \eqdef \{\vec{x} \in \ZeroSets : \VLen{x} = i\}$, $S_i \eqdef \{ \vec s \in S : \VLen{s} = i \}$, $S^* \eqdef \{ (s_1,\ldots,s_{\ell-1}) \in S : \exists s_{\ell} \text{ s.t. } (s_1,\ldots,s_{\ell-1},s_{\ell}) \in S \}$, and
	\begin{equation*}
		z_{\vec s}(\vec t) \eqdef \begin{cases}
			1 & \text{if $\vec t = \vec s$;} \\
			-1 & \text{if $\vec t = (\vec{s}, a)$ for some $a \in \SumSet_{\VLen{s}+1}$;} \\
			0 & \text{ otherwise.}
		\end{cases}	
	\end{equation*}
\end{theorem}
\noindent We will use the above notation for $\ZeroSets_i$, $S_i$, $S^*$ and $z_{\vec s}(\vec t)$ throughout the remainder of this section.

We prove \cref{theorem:sigma-rm-dual-z} by a dimensionality argument. In particular, we show that any basis for $\Dual{(\ZCode{\ZeroSets}{\SigmaRM[\Field,\SumDomain,\vec{\Degree}]}|_{S})}$ which is in echelon form can be mapped to a set of linearly independent vectors in the span of summation constraints and Reed--Muller constraints.

The key technical tool used here is a ``flattening lemma'' (\cref{lem:sigrm-shaving-z}). This lemma provides a means to map any constraint $\Constraint$ on subcube sums of $m$-variate polynomials to a new constraint $\Constraint'$ on subcube sums of $(m-1)$-variate polynomials, while preserving the value of $\lambda(\Constraint)$. Crucially, in the case where $|\lambda(\Constraint)| = m - 1$, all dependence of $\Constraint'$ on subcube sums vanishes, and we obtain a constraint with respect to the plain Reed--Muller code.

By repeated applications of the flattening lemma we map any constraint over subcube sums to a Reed--Muller constraint; we then show by a counting argument that these ``flattened'' constraints, along with the summation constraints, span the dual code.

\begin{proof}[Proof of \cref{theorem:sigma-rm-dual-z}]
	Define $Z^{*} \eqdef \{z_{\vec s}\}_{\vec s \in S^*}$, and $Z_i \eqdef \Dual{(\ZCode{\ZeroSets_i}{\ReedMuller[\Field,i,(\Degree_1,\ldots,\Degree_i)]}|_{S_i})}$ for $i \in [\NumVars]$. It follows easily from the definition of $\SigmaRM$ that $Z_1,\ldots,Z_\NumVars \subseteq \Dual{(\ZCode{\ZeroSets}{\SigmaRM[\Field,\SumDomain,\vec{\Degree}]}|_S)}$. As $S$ is $\SumDomain$-closed we also have that $Z^{*}, \subseteq \Dual{(\ZCode{\ZeroSets}{\SigmaRM[\Field,\SumDomain,\vec{\Degree}]}|_S)}$. Hence it suffices to show that $\dim(\Span(Z^{*} \cup \bigcup_{i=1}^\NumVars Z_i)) \geq \dim(\Dual{(\ZCode{\ZeroSets}{\SigmaRM[\Field,\SumDomain,\vec{\Degree}]}|_{S})})$.
	
	Let $\Delta \eqdef \dim(\Dual{(\ZCode{\ZeroSets}{\SigmaRM[\Field,\SumDomain,\vec{\Degree}]}|_{S})})$. Fix some ordering $<$ on $S$ so that
	\begin{inparaenum}[(i)]
		\item $\VLen{s} < \VLen{t}$ then $\vec s < \vec t$, and
		\item if $\VLen{s} = \VLen{t}$, $\vec s \notin S^{*}$ and $\vec t \in S^{*}$ then $\vec s < \vec t$.
	\end{inparaenum}
	Let $b_1,\ldots,b_{\Delta}$ be a basis for $\Dual{(\ZCode{\ZeroSets}{\SigmaRM[\Field,\SumDomain,\vec{\Degree}]}|_{S})}$ which is in echelon form with respect to $<$.
	
	We show that $\dim(\Span(Z^{*} \cup \bigcup_{i=1}^\NumVars Z_i)) \geq \Delta$ by exhibiting a sequence of linearly independent vectors $b'_1,\ldots,b'_\Delta \in Z^{*} \cup \bigcup_{i=1}^\NumVars Z_i$. These vectors are defined as follows:
	\begin{equation*}
		b'_i \eqdef \begin{cases}
			z_{\vec \lambda(b_i)} & \text{if $\lambda(b_i) \in S^{*}$; and} \\
			R(b_i) & \text{otherwise,}
		\end{cases}
	\end{equation*}
	where $R$ is as guaranteed by \cref{claim:dual-reduction} below. To show that the $b'_i$ are linearly independent, it suffices to show that they are in echelon form with respect to $<$. This follows since $\lambda(b'_i) = \lambda(b_i)$ for all $i \in [\Delta]$: if $\lambda(b_i) \in S^{*}$, then $\lambda(b'_i) = \lambda(z_{\lambda(b_i)}) = \lambda(b_i)$; and if $\lambda(b_i) \notin S^{*}$ then $\lambda(b'_i) = \lambda(R(b_i)) = \lambda(b_i)$ by \cref{claim:dual-reduction}.
\end{proof}

It remains to prove the following technical claim.

\begin{claim}
	\label{claim:dual-reduction}
    Let $S \subseteq \Field^{\leq\NumVars}$, $\SumDomain\eqdef \SumSet_{1}\times\dots \times \SumSet_{\NumVars}$, $\vec{\Degree} \eqdef (\Degree_1, \dots, \Degree_{\NumVars})$ and $\ZeroSets \in \{\varnothing, \SumDomain\}$. If $d_i \geq |\SumDomain_{i}| -1$ for all $i \in [\NumVars]$, then there is a function $T \colon (S \to \Field) \to \bigcup_{i=1}^\NumVars(S_i \to \Field)$ such that for all $z \in \Dual{(\ZCode{\ZeroSets}{\SigmaRM[\Field,\SumDomain,\vec{\Degree}]}|_{S})}$, $T(z) \in \Dual{(\ZCode{\ZeroSets_i}{\ReedMuller[\Field,i,(\Degree_1,\ldots,\Degree_i)]}|_{S_i})}$ for $i = |\lambda(z)|$, and if $\lambda(z) \notin S^{*}$, $\lambda(T(z)) = \lambda(z)$.
\end{claim}

This claim will be a consequence of \cref{lem:sigrm-shaving-z} (the ``flattening lemma''), which provides a way to transform a constraint on an $\NumVars$-variate $\SigmaRM$ code into a related constraint on an $(\NumVars-1)$-variate $\SigmaRM$ code.

\newcommand{\SetS}{S}
\newcommand{\SUp}{S^*}

\newcommand{\IndPoint}[1]{a_{#1}^*}
\newcommand{\DimRedFunc}[1]{Q_{#1, \IndPoint{#1}}}

\begin{definition}
	\label{def:shaving-function}
    Let $S \subseteq\Field^{\leq \NumVars}$ and $\SumSet_1 \times \dots \times \SumSet_{\NumVars} \subseteq \Field^{\NumVars}$. For $i \in [\NumVars]$ and an element $a \in \SumSet_i$, we define the ``flattening map'' $Q_{i,a}\colon (S \to \Field) \to (S\setminus S_i \to \Field)$ as follows, for $\vec{s} \in S \setminus S_i$: 
	
	\begin{equation*}
		Q_{i, a}(\Constraint)(\vec{s}) \eqdef \begin{cases}
			\Constraint(\vec{s}) + \sum_{s_i\in\SetS_{i}}\Constraint(\vec{s},s_i)\LagrangePoly{\SumSet_i, a}(s_i) & \text{if $\vec{s} \in \SetS_{i-1} \cap \SetS^{*}$, and}
			\\
			\Constraint(\vec{s}) & \text{otherwise.}
		\end{cases}
	\end{equation*}
\end{definition}

\begin{lemma}[Flattening lemma]
	\label{lem:sigrm-shaving-z}
    Let $S \subseteq \Field^{\leq\NumVars}$, $\SumDomain\eqdef \SumSet_{1}\times\dots \times \SumSet_{\NumVars}$, $\vec{\Degree} \eqdef (\Degree_1, \dots, \Degree_{\NumVars})$, $\ZeroSets \in \{\varnothing, \SumDomain\}$, and let $\Constraint \in \Dual{(\ZCode{\ZeroSets}{\SigmaRM[\Field,\SumDomain,\vec{\Degree}]}|_{S})}$, where $d_m \geq |\SumSet_m|-1$.
	For any point $a_m^*\in \SumSet_{m}$, the function $\DimRedFunc{m}\colon (S \to \Field) \to (S\setminus S_m \to \Field)$ (see \cref{def:shaving-function}) satisfies the following properties: 
	\begin{enumerate}[]
		\item \label{item:sigrm-shaving1} If $\Constraint$ satisfies $|\lambda(\Constraint)| < \NumVars - 1$, then $\DimRedFunc{m}(\Constraint) \in \Dual{(\ZCode{\ZeroSets_{\NumVars-1}}{\SigmaRM[\Field,\SumDomain_{\NumVars-1},(\Degree_1,\dots,\Degree_{\NumVars-1})]}|_{S\setminus S_i})}$ and $\lambda(\DimRedFunc{m}(\Constraint)) = \lambda(\Constraint)$.
		
		\item If $\Constraint$ satisfies $|\lambda(\Constraint)| = \NumVars - 1$, then $\DimRedFunc{m}(\Constraint) \in \Dual{(\ZCode{\ZeroSets_{|\lambda(\Constraint)|}}{\ReedMuller[\Field, |\lambda(\Constraint)|,(\Degree_1,\ldots,\Degree_{|\lambda(\Constraint)|})]}|_{S_{|\lambda(\Constraint)|}})}$, and furthermore if $\lambda(\Constraint) \notin \SUp$, then $\lambda(\DimRedFunc{m}(\Constraint)) = \lambda(\Constraint)$.
	\end{enumerate}
\end{lemma}

\begin{proof}
    We will proceed by showing how a constraint on subcube sums of an $m$-variate polynomials induces a constraint on subcube sums of $(m-1)$-variate polynomials. 
    
	Let $\Constraint \in \left(\ZCode{\ZeroSets}{\SigmaRM[\Field, \SumDomain, \vec{\Degree}]}|_\SetS\right)^{\bot}$. 
    Define the set $P_0 \eqdef  \{p \in \Polys{\Field}{\vec{\Degree}}{\NumVars} : p(\vec{x}) = 0 ~\forall \vec{x} \in \ZeroSets\}$.
    Then for all polynomials $p \in P_0$ it holds that
	
	\begin{equation}
		\label{eqn:constraint-eq}
		\sum_{\vec{s}\in\SetS}\Constraint(\vec{s})\sum_{\vec{a}\in \SumDomain_{>\VLen{s}}} p(\vec{s},\vec{a}) = 0.
	\end{equation}
	Let $a_m^* \in \SumSet_m$ and define the set of ($\NumVars-1$)-variate polynomials $Q_0 \eqdef \{q \in\Polys{\Field}{\vec{\Degree'}}{\NumVars-1} : q(\vec{x}) = 0 ~\forall \vec{x} \in \ZeroSets_{\NumVars-1}\}$, where $\vec{\Degree'} \eqdef (d_1,\dots,d_{\NumVars-1})$. Then for any $q \in Q_{0}$, the polynomial defined by $q_{a_{m}^*}(x_1,\dots,x_{\NumVars}) \eqdef q(x_1,\dots,x_{\NumVars-1})\cdot \LagrangePoly{\SumSet_m, a_m^*}(x_m)$ satisfies $q_{a_{m}^*}(\vec{x}) = 0$ for all $\vec{x} \in \ZeroSets$. So $q_{a_{m}^*}\in P_0$. Thus, substituting $q_{a_m^*}$ into \Cref{eqn:constraint-eq} implies that for all polynomials $q \in Q_0$,
	
	\begin{equation*}
		\sum_{\vec{s} \in \SetS\setminus (\SetS_{\NumVars-1}\Union\SetS_{\NumVars})}\Constraint(\vec{s})\sum_{\vec{a}\in\SumDomain'_{>\VLen{s}}} q(\vec{s},\vec{a}) + \sum_{\vec{s}\in\SetS_{\NumVars-1}} \Constraint(\vec{s})q(\vec{s}) + \sum_{\vec{s}\in\SetS_{\NumVars-1}}q(\vec{s}) \sum_{\substack{s_m \\ (\vec{s},s_m) \in S_m}}\Constraint(\vec{s},s_m)\LagrangePoly{\SumDomain_m, a_m^*}(s_m) = 0,
	\end{equation*}
    where $\SumDomain' \eqdef \SumDomain_{\NumVars-1}$. Recall that for points $(s_1, \dots, s_{\NumVars-1}) \notin \SUp$, there does not exist $s_{\NumVars} \in \SetS_{\NumVars}$ such that $(s_1, \dots, s_{\NumVars-1}, s_{\NumVars}) \in \SetS$. Thus we can rearrange this expression as: 
	\begin{equation}
		\label{eqn:q-constraint}
		\sum_{\vec{s} \in \SetS\setminus (\SetS_{\NumVars-1}\Union\SetS_{\NumVars})}\Constraint(\vec{s})\sum_{\vec{a}\in\SumDomain'_{>\VLen{s}}} q(\vec{s},\vec{a}) + \sum_{\vec{s}\in\SetS_{\NumVars-1}\setminus S^{*}} \Constraint(\vec{s})q(\vec{s}) + \sum_{\vec{s}\in\SetS_{\NumVars-1}\cap S^{*}}q(\vec{s})\left( \Constraint(\vec{s}) + \sum_{\substack{s_m \\ (\vec{s},s_m) \in S_m}}\Constraint(\vec{s},s_m)\LagrangePoly{\SumDomain_m, a_m^*}(s_m) \right) = 0.
	\end{equation}
	Notice that the coefficients of \cref{eqn:q-constraint} are precisely the values of $\DimRedFunc{m}$ (\cref{def:shaving-function}) at each point in $S \setminus S_m$. This demonstrates that for all $\Constraint \in \left(\ZCode{\ZeroSets}{\SigmaRM[\Field, \SumDomain, \vec{\Degree}]}|_\SetS\right)^{\bot}$, $\DimRedFunc{m}(\Constraint) \in \left(\ZCode{\ZeroSets_{\NumVars-1}}{\SigmaRM[\Field, \SumDomain', \vec{\Degree'}]}|_{\SetS\setminus\SetS_{\NumVars}}\right)^{\bot}$.
	
	Now we will deal with the two cases of the lemma statement separately. 
	
	\parhead{Case (i) $|\lambda(\Constraint)| + 1 < \NumVars$}
	In this case we show that  $\lambda(\DimRedFunc{m}(\Constraint))=\lambda(\Constraint)$.
	
	As $|\lambda(\Constraint)| + 1 < \NumVars$, $\lambda(\Constraint) \in \SetS\setminus(\SetS_m \Union \SetS_{m-1})$. But we have that $\DimRedFunc{m}(\Constraint)(\vec{s}) = \Constraint(\vec{s})$ for all $\vec{s} \in \SetS\setminus(\SetS_m \Union \SetS_{m-1})$, so in particular $\DimRedFunc{m}(z)(\vec{s}) = \Constraint(\vec{s}) = 0$ for all $\vec{s} < \lambda(\Constraint)$, and $\DimRedFunc{m}(z)(\lambda(z)) = z(\lambda(z)) \neq 0$. Thus $\lambda(\DimRedFunc{m}(\Constraint))=\lambda(\Constraint)$.
	
	\parhead{Case (ii) $|\lambda(\Constraint)| + 1 = \NumVars$}
	First, we show that $\DimRedFunc{m}(\Constraint) \in \Dual{(\ZCode{\ZeroSets_{i}}{\ReedMuller[\Field, i,(\Degree_1,\ldots,\Degree_i)]}|_{S_i})}$, where $i\eqdef|\lambda(\Constraint)|$. Second, we show that if $\lambda(\Constraint) \notin \SUp$, then $\lambda(\DimRedFunc{m}(\Constraint)) = \lambda(\Constraint)$.
	
	For the first part, as $|\lambda(\Constraint)| = \NumVars-1$, $\Constraint(\vec{s}) = 0$ for all $\vec{s} \in S \setminus(S_{\NumVars-1} \Union S_{\NumVars})$. As a result, \Cref{eqn:q-constraint} simplifies greatly and becomes:
	\begin{equation}
		\label{eqn:special-case-q-constraint}
		\sum_{\vec{s}\in\SetS_{|\lambda(\Constraint)|}\setminus S^{*}} \Constraint(\vec{s})q(\vec{s}) + \sum_{\vec{s}\in\SetS_{|\lambda(\Constraint)|}\cap \SetS^{*}}q(\vec{s})\left( \Constraint(\vec{s}) + \sum_{\substack{s_m \\ \vec{s}||s_m \in S_m}}\Constraint(\vec{s},s_m)\LagrangePoly{\SumDomain_m, a_m^*}(s_m) \right) = 0.
	\end{equation}
	Note that \Cref{eqn:special-case-q-constraint} is a constraint with respect to the Reed-Muller code $\Dual{(\ZCode{\ZeroSets_i}{\ReedMuller[\Field, i,(\Degree_1,\ldots,\Degree_i)]}|_{S_i})}$. So $\DimRedFunc{m}(\Constraint) \in \Dual{(\ZCode{R_{i}}{\ReedMuller[\Field, i,(\Degree_1,\ldots,\Degree_i)]}|_{S_i})}$. 
	
	For the second part, assume $\lambda(\Constraint) \notin \SUp$. By the definition of $<$, if $\vec{s} < \lambda(\Constraint)$, then it must be that $\vec{s} \notin \SUp$. However, if $\vec{s} \notin \SUp$, then $\DimRedFunc{m}(\Constraint)(\vec{s}) = \Constraint(\vec{s})$. In particular, for $\vec{s} < \lambda(\Constraint)$,  $\DimRedFunc{m}(\Constraint)(\vec{s}) = \Constraint(\vec{s}) = 0$, and $\DimRedFunc{m}(\Constraint)(\lambda(\Constraint)) = \Constraint(\lambda(\Constraint)) \neq 0$, so $\lambda(\DimRedFunc{m}(\Constraint)) = \lambda(\Constraint)$.
\end{proof}

We now apply \cref{lem:sigrm-shaving-z} to prove \cref{claim:dual-reduction}, which will complete the proof of \cref{theorem:sigma-rm-dual-z}.

\begin{proof}[Proof of \cref{claim:dual-reduction}]
	For each $i \in \{|\lambda(\Constraint)|+1,\dots,\NumVars\}$ arbitrarily chose an element $a_i \in \SumSet_{i}$, and for notational convenience, denote by $Q_i \eqdef Q_{i, a_{i}}$ the flattening map (\cref{def:shaving-function}). For each $\Constraint \in \Dual{(\ZCode{\ZeroSets}{\SigmaRM[\Field,\SumDomain,\vec{\Degree}]}|_{S})}$, define the value of $T(\Constraint) \eqdef (Q_{|\lambda(\Constraint)|+1}\circ\dots\circ Q_{\NumVars-1}\circ Q_{\NumVars})(\Constraint)$.
	
	By part 1 of \cref{lem:sigrm-shaving-z}, we see that $(Q_{|\lambda(z)|+2}\circ\dots\circ Q_{m})(z) \in \Dual{(\ZCode{\ZeroSets_{i+1}}{\SigmaRM[\Field, \SumDomain_{i+1}, (\Degree_1,\dots,\Degree_{i+1})]}|_{\SetS\setminus\SetS_{\NumVars}})}$, with $\lambda((Q_{|\lambda(z)|+2}\circ\dots\circ Q_{m})(z)) = \lambda(z)$. Combining this with part 2 of \cref{lem:sigrm-shaving-z} we see that $T(z) \eqdef (Q_{|\lambda(z)|+1}\circ\dots\circ Q_{m})(z) \in \Dual{(\ZCode{\ZeroSets_i}{\ReedMuller[\Field, i,(\Degree_1,\ldots,\Degree_i)]}|_{S_i})}$, and provided that $\lambda(z) \notin \SUp$, then $\lambda(T(z)) = \lambda(z).$
\end{proof}

\subsection{An $\ell$-constraint locator for $\Sigma$RM}
\label{sec:construction-sigmarm}

In this section we construct an constraint locator for $\SigmaRM$, by combining the structure implied by \cref{theorem:sigma-rm-dual-z} with the constraint locator we constructed in \cref{sec:rm-cl}. We will often denote $\SigmaRM \eqdef \SigmaRM[\Field, \NumVars,\vec{\Degree}]$ when $\Field, \NumVars, \vec{\Degree}$ are clear from context.

\newcommand{\ClosedSubdomain}{\hat{I}}

\newcommand{\RClosed}{\hat{R}}

\begin{mdframed}[nobreak=true]
	\begin{construction}
		\label{construc:sigrm-cl}
		An $\ell$-constraint locator for $\SigRMEnc{\vec{\Degree}}{\SumDomain}$. For each $i \in [\NumVars]$, let $\CL_i$ be an $\ell_i$-constraint locator for $\RMEnc{(\Degree_1,\dots,\Degree_i)}{\SumDomain_i}$ (cf. \cref{thm:rm-constraint-loc}). Input: a subdomain $\Subdomain \subseteq \Field^{\NumVars}.$ \ConstrucSpacing
        $\CL_\Sigma(\Subdomain)$:
        \begin{enumerate}[nolistsep]
            \item Compute the $\SumDomain$-closure $\ClosedSubdomain$ of $\Subdomain$.

            \item \label{step:compute-rm-cl}For each $i\in [\NumVars]$, compute $(R_i, \SubConstraintMatrix_i) \eqdef \CL_i(\ClosedSubdomain_{i})$, where $\ClosedSubdomain_{i} \eqdef \{\vec{x} \in \ClosedSubdomain \colon \VLen{x} = i\}$.

            \item \label{step:define-R}Define the set $R \eqdef \bigcup_{i=1}^{\NumVars} R_i$ and compute the $\SumDomain$-closure $\RClosed$ of $R$.

            \item For each $i \in [\NumVars]$, define $\SubConstraintMatrix_i'$ to be the natural embedding of $\SubConstraintMatrix_i$ into the space of matrices $\Field^{k_i\times\RClosed_i\sqcup \ClosedSubdomain}$ (by appropriate padding with $|\RClosed_i|-|R_i|$ columns of zeros), where $k_i$ is the number of rows of $\SubConstraintMatrix_i$. 

            \item Construct the matrix
			\begin{equation*}
				Z' \eqdef \left(\begin{array}{cccc}
					\SubConstraintMatrix_1' & \ZeroMatrix & \dots & \ZeroMatrix \\
					\ZeroMatrix & \SubConstraintMatrix_2' & \dots & \ZeroMatrix \\
					\vdots & \vdots & \ddots & \vdots \\
					\ZeroMatrix & \ZeroMatrix & \dots & \SubConstraintMatrix_{\NumVars}' \\
				\end{array}\right)\enspace.
			\end{equation*}

            \item Construct the matrix $Z^*$ whose $(\vec{s}, \vec{\gamma})$-th entry, for $\vec{s} \in (\RClosed \sqcup \ClosedSubdomain)^*$, and $\vec{\gamma} \in \RClosed\sqcup\ClosedSubdomain$, is defined by
            \begin{equation*}
				\Constraint_{\vec{s}}(\vec{\gamma}) = \begin{cases}
					1 & \text{if}~\vec{\gamma}=\vec{s} ,\\
					-1 &\text{if}~\vec{\gamma} = (\vec{s},a) ~\text{for some}~a \in \SumSet_{\VLen{\gamma}+1}, \\
					0 & \text{otherwise.}
				\end{cases}
			\end{equation*}

            \item Define $Z$ to be the matrix obtained by vertically stacking $Z'$ on top of $Z^*$.

            \item Compute a basis $B$ for the space $\{(\vec{r},\vec{\beta}) \in \Field^{\RClosed} \times \Field^{\Subdomain} : \exists \vec{\gamma} \in \Field^{\ClosedSubdomain \setminus\Subdomain}, Z(\vec{r}, \vec{\beta},\vec{\alpha})^T = 0\}$.
            
            \item Output $(\RClosed, B)$.
        \end{enumerate}
    \end{construction}
\end{mdframed}

\newcommand{\RowSpace}{\mathsf{rowspace}}

Our proof of correctness for \cref{construc:sigrm-cl} consists of two main steps. The first is showing that $\ker(B) = \SigmaRM|_{\RClosed \sqcup \Subdomain}$, which we accomplish by appealing to \cref{theorem:sigma-rm-dual-z}. We then show that $(\vec{r}, \vec{\beta}) \in \SigmaRM|_{\RClosed \sqcup \Subdomain}$ if and only if for all messages $\Message \in \MsgSpace$ such that $\Message|_{\RClosed} = \vec{r}$, we have $(\Message, \vec{\beta}) \in \SigmaRM|_{\SumCompletion \sqcup I}$.

\begin{claim}
\label{claim:cl-output}
    Let $\Subdomain \subseteq \Field^{\NumVars}$ and let $(\RClosed, B) \eqdef \CL_\Sigma(\Subdomain)$  (cf. \cref{construc:sigrm-cl}), then $\ker(B) = \SigmaRM[\Field, \NumVars, \vec{\Degree}]|_{\RClosed\sqcup \Subdomain}$.
\end{claim}

\begin{proof}
    It suffices to show that $\ker(Z) = \SigmaRM|_{\RClosed\sqcup \ClosedSubdomain}$ (for $Z$ as in \cref{construc:sigrm-cl}). This is because $(\vec{r}, \vec{\beta}) \in \ker(B)$ if and only if there exists $\vec{\gamma} \in \Field^{\ClosedSubdomain\setminus\Subdomain}$ such that $(\vec{r}, \vec{\beta}, \vec{\gamma}) \in \ker(Z)$. So $\ker(B) = \ker(Z)|_{\RClosed \sqcup \Subdomain}.$ If $\ker(Z) = \SigmaRM|_{\RClosed\sqcup \ClosedSubdomain}$, then upon restricting both sides of the equivalence to $\RClosed \sqcup \Subdomain$, we have that  $\ker(B) = \SigmaRM|_{\RClosed \sqcup \Subdomain}$, as required.
    
    To see that $\ker(Z) = \SigmaRM|_{\RClosed\sqcup \ClosedSubdomain}$, consider the following. By construction, the rows of $Z$ consist of trivial constraints over $(\ClosedSubdomain\sqcup\RClosed)^*$ and Reed-Muller constraints on the set $R_i\sqcup\ClosedSubdomain$ padded with zeros so as to be supported on $\RClosed\sqcup \ClosedSubdomain$. That is,

    \begin{equation*}
        \RowSpace(Z) = \Span\big(\{z_{\vec s}\}_{\vec s \in (\RClosed \sqcup \ClosedSubdomain)^*} \cup \bigcup_{i=1}^{\NumVars} \Dual{({\ReedMuller[\Field,i,(\Degree_1,\ldots,\Degree_i)]}|_{{R_i\sqcup \ClosedSubdomain}_i})} \times \{0\}^{\RClosed_i\setminus R_i}\big).
    \end{equation*}
    By the correctness of each $\CL_i$, we know that for all $i \in [\NumVars]$, for all constraints $\Constraint \in \Dual{(\ReedMuller[\Field, i, (\Degree_1,\dots, \Degree_i)]|_{\SumDomain_i \sqcup \ClosedSubdomain})}$, $\Constraint(\vec{\alpha}) = 0$ for all $\vec{\alpha} \in \SumDomain_i \setminus R_i$. Therefore as $R_i \subseteq \RClosed_i \subseteq \SumDomain_i$, we have that $\Dual{(\ReedMuller[\Field, i, (\Degree_1,\dots, \Degree_i)]|_{\RClosed_i \sqcup \ClosedSubdomain})} = \Dual{(\ReedMuller[\Field, i, (\Degree_1,\dots, \Degree_i)]|_{R_i \sqcup \ClosedSubdomain})} \times \{0\}^{\RClosed_i \setminus R_i}$. Hence,
    \begin{equation*}
        \RowSpace(Z) = \Span\big(\{z_{\vec s}\}_{\vec s \in (\RClosed \sqcup \ClosedSubdomain)^*} \cup \bigcup_{i=1}^{\NumVars} \Dual{({\ReedMuller[\Field,i,(\Degree_1,\ldots,\Degree_i)]}|_{{\RClosed_i\sqcup \ClosedSubdomain}_i})}\big)
         = \Dual{(\SigmaRM|_{\RClosed\sqcup \ClosedSubdomain})},
    \end{equation*}
    where the final equality follows from \cref{theorem:sigma-rm-dual-z} and the fact that $\RClosed\sqcup\ClosedSubdomain$ is $\SumDomain$-closed. In particular, we have shown that $Z$ is a parity check matrix for $\SigmaRM|_{\RClosed\sqcup \ClosedSubdomain}$. In other words, $\ker(Z) = \SigmaRM|_{\RClosed\sqcup \ClosedSubdomain}$.
\end{proof}

Now we focus on the second part of our proof of correctness: showing that $(\vec{r}, \vec{\beta}) \in \SigmaRM|_{\RClosed \sqcup \Subdomain}$ if and only if for all messages $\Message \in \MsgSpace$ such that $\Message|_{\RClosed} = \vec{r}$, we have $(\Message, \vec{\beta}) \in \SigmaRM|_{\SumCompletion \sqcup I}$. We accomplish this via the following property of linear codes, whose proof is deferred to \cref{sec:zcode-implies-padded-message}.

\begin{restatable}{claim}{ZCodePadded}
    \label{claim:zcode-implies-padded-message}
    Let $\Code \subseteq \Field^\Domain$ be a linear code, and let $\SmallSubdomain \subseteq \BigSubdomain \subseteq \Domain$, $\Subdomain \subseteq \Domain$ be such that $\ZCode{\SmallSubdomain}{\Code}|_{\Subdomain} = \ZCode{\BigSubdomain}{\Code}|_{\Subdomain}$. Then $(\SmallElt, \ExtElt) \in \Code|_{\SmallSubdomain\sqcup \Subdomain}$ if and only if for all $\BigElt \in \Code|_\BigSubdomain$ such that $\BigElt|_\SmallSubdomain = \SmallElt$, $(\BigElt, \ExtElt) \in \Code|_{\BigSubdomain\sqcup \Subdomain}$.
\end{restatable}

We prove that $\ZCode{\SumDomain}{\SigmaRM}|_{\ClosedSubdomain} = \ZCode{\RClosed}{\SigmaRM}|_{\ClosedSubdomain}$ where $\ClosedSubdomain,\RClosed$ are as above. To do this, we first show an analogous statement for ``plain'' $\ReedMuller$, and then lift that result to $\SigmaRM$ via \cref{theorem:sigma-rm-dual-z}.

\begin{claim}
\label{claim:rm-constraint-bound}
    Let $\SumDomain = \SumSet_1\times\dots\times\SumSet_{\NumVars} \subseteq \Field^{\NumVars}$ and let $S \subseteq \Field^{\NumVars}$. It holds that $\ZCode{\SumDomain}{\ReedMuller[\Field, \NumVars, \vec{\Degree}]}|_{S} = \ZCode{R}{\ReedMuller[\Field, \NumVars, \vec{\Degree}]}|_{S}$, where $(R,Z) \eqdef \CL(S)$, for an $\ell$-constraint locator $\CL$ for $\RMEnc{\vec{\Degree}}{\SumDomain}$.
\end{claim}

\begin{proof}
    As $R\subseteq \SumDomain$, it is clear that $\ZCode{\SumDomain}{\ReedMuller[\Field, \NumVars, \vec{\Degree}]}|_{S} \subseteq \ZCode{R}{\ReedMuller[\Field, \NumVars, \vec{\Degree}]}|_{S}$. We now argue the reverse containment. For brevity we write $\ReedMuller$ for $\ReedMuller[\Field, \NumVars, \vec{\Degree}]$.

    Let $\Codeword \in \ReedMuller$ such that $\Codeword(r) = 0$ for all $r \in R$. Then $\Codeword|_{S}$ is a general element of $\ZCode{R}{\ReedMuller}|_{S}$. We will show that there exists $\Codeword' \in \ReedMuller$ which agrees with $\Codeword$ on $S$, but is zero on $\SumDomain$, i.e., $\Codeword'|_{S} = \Codeword|_{S}$ and $\Codeword'|_{\SumDomain} = \vec{0}$.

    By the correctness of $\CL$, all constraints on $\SumDomain\sqcup S$ with respect to $\ReedMuller$ are supported only on $R \sqcup S$; in particular, they have no support on the set $\SumDomain \setminus R$. For $\alpha \in \SumDomain\setminus R$, define $\Codeword_\alpha \colon \SumDomain \sqcup S \to \Field$ by $\Codeword_\alpha(\alpha) \eqdef 1$ and $\Codeword_\alpha (\beta) \eqdef 0$ for all $\beta \in \SumDomain \sqcup S \setminus \{\alpha\}$. Then since $\Codeword_\alpha$ is only nonzero in $\SumDomain \setminus R$, $\Codeword_\alpha \in \ReedMuller|_{\SumDomain\sqcup S}$, and so $\Codeword_\alpha = \Codeword_\alpha'|_{\SumDomain\sqcup S}$ for $\Codeword_\alpha' \in \ReedMuller$. Consider the codeword defined by 
    \begin{equation*}
        \Codeword' \eqdef \Codeword - \sum_{\alpha \in \SumDomain \setminus R} \Codeword(\alpha) \Codeword_{\alpha}'.
    \end{equation*}
    This choice of $\Codeword'$ satisfies $\Codeword' \in \ZCode{\SumDomain}{\ReedMuller}$, and $\Codeword'|_{S} = \Codeword|_{S}$. Thus $\Codeword|_{S} \in \ZCode{\SumDomain}{\ReedMuller}|_{S}$. 
\end{proof}

\begin{corollary}
\label{cor:z-sigrm-not-many-constraints}
   For $\ClosedSubdomain, R$ as in \cref{construc:sigrm-cl}, $\Dual{(\ZCode{\SumDomain}{\SigmaRM[\Field, \SumDomain, \vec{\Degree}]}|_{\ClosedSubdomain})} = \Dual{(\ZCode{\RClosed}{\SigmaRM[\Field, \SumDomain, \vec{\Degree}]}|_{\ClosedSubdomain})}$.
\end{corollary}

\begin{proof}
    It is clear that $\Dual{(\ZCode{\RClosed}{\SigmaRM[\Field,\SumDomain,\vec{\Degree}]}|_{\ClosedSubdomain})} \subseteq  \Dual{(\ZCode{\SumDomain}{\SigmaRM[\Field,\SumDomain,\vec{\Degree}]}|_{\ClosedSubdomain})}$, so we show the reverse inclusion. 

    For $R_i, Z_i, \ClosedSubdomain_i$ as in \cref{construc:sigrm-cl}, we have that 
    \begin{align*}
        \Dual{(\ZCode{\SumDomain}{\SigmaRM[\Field,\SumDomain,\vec{\Degree}}|_{\ClosedSubdomain})} &= \Span\big(\{z_{\vec s}\}_{\vec s \in {\ClosedSubdomain}^*} \cup \bigcup_{i=1}^{\NumVars} \Dual{(\ZCode{\SumDomain_i}{\ReedMuller[\Field,i,(\Degree_1,\ldots,\Degree_i)]}|_{\ClosedSubdomain_i})}\big)
        \\
        &= \Span\big(\{z_{\vec s}\}_{\vec s \in \ClosedSubdomain^*} \cup \bigcup_{i=1}^{\NumVars} \Dual{(\ZCode{R_{i}}{\ReedMuller[\Field,i,(\Degree_1,\ldots,\Degree_i)]}|_{\ClosedSubdomain_i})}\big)
        \\
        &\subseteq \Span\big(\{z_{\vec s}\}_{\vec s \in \ClosedSubdomain^*} \cup \bigcup_{i=1}^{\NumVars} \Dual{(\ZCode{{\RClosed}_{i}}{\ReedMuller[\Field,i,(\Degree_1,\ldots,\Degree_i)]}|_{\ClosedSubdomain_i})}\big) 
        \\
        &\subseteq \Dual{(\ZCode{\RClosed}{\SigmaRM[\Field,\SumDomain,\vec{\Degree}]}|_{\ClosedSubdomain})}.
    \end{align*}
    The first equality is a consequence of \Cref{theorem:sigma-rm-dual-z}. The second equality is a consequence of \cref{claim:rm-constraint-bound}. The penultimate inclusion is true as $R_i \subseteq \RClosed_i$ for each $i \in [\NumVars]$. The final inclusion follows from the fact that trivially $\Dual{(\ZCode{\RClosed_i}{\ReedMuller[\Field,i,(\Degree_1,\ldots,\Degree_i)]}|_{\ClosedSubdomain_i})} \subseteq \Dual{(\ZCode{\RClosed}{\SigmaRM[\Field,\SumDomain,\vec{\Degree}]}|_{\ClosedSubdomain})}$ for each $i \in [\NumVars]$ and as $\RClosed$ is $\SumDomain$-closed, we also have that $\{z_{\vec s}\}_{\vec s \in \ClosedSubdomain^*} \subseteq \Dual{(\ZCode{\RClosed}{\SigmaRM[\Field,\SumDomain,\vec{\Degree}]}|_{\ClosedSubdomain})}$.
\end{proof}

\begin{proof}[Proof of \Cref{thm:sigrm-cl}]
\parhead{Correctness}
    By \cref{claim:cl-output}, we have that $\ker(B) = \SigmaRM|_{\RClosed \sqcup \Subdomain}$. Then by \cref{claim:zcode-implies-padded-message}, setting $U \eqdef \RClosed$ and $V \eqdef \SumCompletion$, and \cref{cor:z-sigrm-not-many-constraints}, noting that $\SigRMEnc{\vec{\Degree}}{\SumDomain}$ is a systematic encoding (that is $\MsgSpace = \SigmaRM|_{\SumCompletion}$), it follows that $(\vec{r}, \vec{\beta}) \in \SigmaRM|_{\RClosed\sqcup \ClosedSubdomain}$ if and only if for all messages $\Message \in \MsgSpace$ such that $\Message|_{\RClosed} = \vec{r}$, we have $(\Message, \vec{\beta}) \in \SigmaRM|_{\SumCompletion \sqcup \ClosedSubdomain}$. Thus, taking the restriction of both sides, $(\vec{r}, \vec{\beta}|_{\Subdomain}) \in \SigmaRM|_{\RClosed\sqcup \Subdomain}$ if and only if for all messages $\Message \in \MsgSpace$ such that $\Message|_{\RClosed} = \vec{r}$, we have $(\Message, \vec{\beta}|_\Subdomain) \in \SigmaRM|_{\SumCompletion \sqcup \Subdomain}$. This is sufficient since, by definition of the encoding, $\vec{\beta} \in \supp(\SigRMEnc{\vec{\Degree}}{\SumDomain}(\Message)|_{\Subdomain})$ if and only if $(\Message, \vec{\beta}) \in \SigmaRM|_{\SumCompletion \sqcup I}$.

    \parhead{Efficiency} By \cref{claim:closure-alg} (below), $|\RClosed| \leq (\sum_{i=1}^{\NumVars}|\SumSet_i| + \NumVars + 1)|R|$ and $|\ClosedSubdomain| \leq (\sum_{i=1}^{\NumVars}|\SumSet_i| + \NumVars + 1)|\Subdomain|$.
    By \cref{thm:rm-constraint-detector}, $|R_i| \leq |\ClosedSubdomain| \leq (\sum_{i=1}^{\NumVars}|\SumSet_i| + \NumVars + 1)|\Subdomain|$. As $|R| \leq \sum_{i=1}^{\NumVars}|R_i| \leq \NumVars (\sum_{i=1}^{\NumVars}|\SumSet_i| + \NumVars + 1)|\Subdomain|$, we have $|\RClosed| \leq \NumVars (\sum_{i=1}^{\NumVars}|\SumSet_i| + \NumVars + 1)^2 |\Subdomain|$.

    The runtime follows from \cref{claim:closure-alg} and \Cref{thm:rm-constraint-loc}, as $\CL_i$ runs in time $\Poly(\log|\Field|, \NumVars, \max_i \Degree_i, n)$.
\end{proof}

\begin{claim}
\label{claim:closure-alg}
    Let $\Subdomain \subseteq \Field^{\leq \NumVars}$ and $\SumDomain = \SumSet_1\times\dots\times\SumSet_{\NumVars}$. Then there exists a set $\ClosedSubdomain\subseteq \Field^{\leq \NumVars}$ such that $\Subdomain \subseteq \ClosedSubdomain$ and $\ClosedSubdomain$ is $\SumDomain$-closed, satisfying $|\ClosedSubdomain| \leq (\sum_{i=1}^{\NumVars}|\SumSet_i| + \NumVars + 1)|\Subdomain|$. Moreover, $\ClosedSubdomain$ can be computed in time $\Poly(\NumVars, |\Subdomain|, |\SumSet_1|,\dots,|\SumSet_{\NumVars}|)$.
\end{claim}

\begin{proof}
    We construct $\ClosedSubdomain$ as follows. Let $\Subdomain_i \eqdef \{x \in \Subdomain : |x| = i\}$. Initialise $\ClosedSubdomain \eqdef \Subdomain$. For each $(x_1,\dots, x_m) \in \Subdomain_m$, we add the point $(x_1,\dots, x_{m-1})$ and for each $a_m \in \SumSet_m$ we add the point $(x_1, \dots, x_{m-1}, a_m)$ to $\ClosedSubdomain$. Then repeat this process for $\Subdomain_{m-1}, \Subdomain_{m-2}, \dots, \Subdomain_1$. By construction, $\ClosedSubdomain$ is $\SumDomain$-closed. 

    Now we argue about the size of $\ClosedSubdomain$. Note that for each point in $\Subdomain$ at most $\sum_{i=1}^{\NumVars}|\SumSet_i| + \NumVars$ points are added to $\ClosedSubdomain$. Thus it holds that $|\ClosedSubdomain| \leq (\sum_{i=1}^{\NumVars}|\SumSet_i| + \NumVars + 1)|\Subdomain|$. The efficiency follows from this bound. 
\end{proof}


\section{Constraint location for $\Sigma$AntiSym}
\label{sec:cl-antisym}

In this section we define the code $\AntiSym$ of antisymmetric functions, and provide a constraint locator for an encoding function related to its sum code $\SigmaAntiSym$.

\newcommand{\SigmaAntiSymEnc}{\Enc_{\SigmaAntiSym}}

\begin{definition}
    \label{def:antisym}
    Let $\SumDomain \eqdef \SumSet_1 \times \cdots \times \SumSet_{\NumVars}$ be such that $\SumSet_i = \SumSet_{\NumVars-i+1}$ for all $1 \leq i \leq \NumVars$. We define
	\[ \AntiSym[\SumDomain] \eqdef \{ w \in \Field^{\SumDomain} : \forall \vec{x} \in \SumDomain, w(\vec{x}) + w(\revvec{x}) = 0 \}~. \]
 The code $\SigmaAntiSym$ is the sum code of $\AntiSym$ (see \cref{def:sumcodes}).
\end{definition}

The main theorem of this section is the following, which gives an efficient constraint locator for $\SigmaAntiSym$ with polynomial locality.

\begin{theorem}
   \label{theorem:antisym-cl}
   Let $\SumDomain$ be as above. Let $\MsgSpace \eqdef \{ F \colon \SumDomain \cup \{ \bot \} \to \Field \mid F(\bot) = \sum_{\vec a \in \SumDomain} F(\vec a) \}$. Let $\SigmaAntiSymEnc \colon \MsgSpace \to \Field^{\SumCompletion}$ be the following randomised encoding:
	\begin{equation*}
		\SigmaAntiSymEnc(F) \eqdef \SumWord{}{F|_\SumDomain+G},
	\end{equation*}
	where $G$ is a codeword sampled uniformly at random from $\AntiSym[\SumDomain]$.

    $\SigmaAntiSymEnc$ has an $\ell$-constraint locator for $\ell(n) = n^2\NumVars^2 + 1 + O(n^4\NumVars^4/|\SumDomain|)$, running in time $\Poly(a,m,n)$, where $a \eqdef \sum_i |\SumSet_i|$.
\end{theorem}

In \cref{sec:antisym-props} we analyse the combinatorial structure of $\Dual{\SigmaAntiSym|_I}$, obtaining a bound on the ``complexity'' of constraints in terms of the size of $I$. In \cref{sec:antisym-algs}, we show how this bound leads to an efficient constraint locator for $\SigmaAntiSymEnc$. 

\subsection{Properties of \texorpdfstring{$\SigmaAntiSym$}{Sigma-AntiSym}}
\label{sec:antisym-props}

We start by introducing some useful notation.

\begin{definition}
    We say $\vec a \in \SumCompletion$ is a \emph{prefix} of $\vec b \in \SumCompletion$ if $\VLen{a} \leq \VLen{b}$ and $(a_1,\ldots,a_{\VLen{a}}) = (b_1,\ldots,b_{\VLen{a}})$. A subset $J \subseteq \SumCompletion$ is \defemph{prefix-free} if for all $\vec a, \vec b \in J$, if $\vec a$ is a prefix of $\vec b$ then $\vec a = \vec b$.
	
	We associate with an element $\vec a \in \SumCompletion$ the set $\vec a \times \SumDomain_{> \VLen{a}}$; we denote this set also by $\vec a$. In particular, $|\vec a| = |\SumDomain_{> \VLen{a}}| = \prod_{i=\VLen{a}+1}^\NumVars |\SumSet_i|$. (See \cref{def:sumcodes}.)
\end{definition}

It is straightforward to see that $\vec b \subseteq \vec a$ if and only if $\vec a$ is a prefix of $\vec b$. The following simple proposition gives a useful alternative characterisation of $\AntiSym$.

\begin{proposition}
	\label{prop:antisym}
	$w \in \AntiSym[\SumDomain]$ if and only if there exists $f \in \Field^{\SumDomain}$ such that $w(\vec{x}) = f(\vec{x}) - f(\revvec{x})$.
\end{proposition}
\begin{proof}
	For the ``if'' direction, note that $w(\vec{x}) + w(\revvec{x}) = f(\vec{x}) - f(\revvec{x}) + f(\revvec{x}) - f(\vec{x}) = 0$ for all $\vec{x} \in \SumDomain$. For the ``only if'' direction, let $w \in \AntiSym[\vec{\SumSet}]$, and consider some strict total ordering $<$ on $\SumDomain$. Define
	\begin{equation*}
		f(\vec{x}) \eqdef \begin{cases}
 			w(\vec{x}) & \text{if $\vec{x} < \revvec{x}$, and} \\
 			0 & \text{otherwise.}
 		\end{cases}
	\end{equation*}
	Then for all $\vec{x}$,
	\begin{equation*}
		f(\vec{x}) - f(\revvec{x}) = \begin{cases}
				w(\vec{x}) & \text{ if $\vec{x} < \revvec{x}$, and} \\
				-w(\revvec{x}) & \text{ if $\revvec{x} < \vec{x}$}
			\end{cases} = w(\vec{x})~. \qedhere
	\end{equation*}
\end{proof}

\begin{definition}
    For $\vec a = (a_1,\ldots,a_{\ell}) \in \SumCompletion$, we define the \defemph{reverse set} $\revvec a \eqdef \SumCompletion_{\NumVars-\ell} \times (a_{\ell}, \ldots, a_1)$. For a set $H \subseteq \SumCompletion$, we define $H_\rev \eqdef \{ \revvec a : \vec a \in H \}$.
    
    We say that $H$ is \defemph{symmetric} if $\cup H = \cup H_\rev$, where $\cup H \eqdef \cup_{\vec h \in H} \vec h$. We say that $H$ is \defemph{minimal symmetric} if $H$ is symmetric and for all nonempty $H' \subsetneq H$, $H'$ is not symmetric.
\end{definition}

\begin{proposition}
    \label{prop:minimal-symmetric-decomp}
    If $H \subseteq \SumCompletion$ is symmetric and prefix-free, then $H = \cup_{i=1}^k H_i$ for $H_1,\ldots,H_k \subseteq H$ where each $H_i$ is minimal symmetric and $H_i \cap H_j = \varnothing$ for $i \neq j$.
\end{proposition}
\begin{proof}
    We proceed by induction on the size of $H$. Indeed, suppose that the statement is true for all sets $H' \subseteq \SumCompletion$ of size less than $n$, and let $H \subseteq \SumCompletion$ be a set of size $n$. If $H$ is minimal symmetric, then the statement holds trivially, so suppose not. Then there is a nonempty symmetric $H' \subsetneq H$. Since $H$ is prefix-free, $\cup (H \setminus H') = (\cup H) \setminus (\cup H') = (\cup H_\rev) \setminus (\cup H'_\rev) = \cup (H \setminus H')_\rev$; hence $H \setminus H'$ is also symmetric. Since $H'$ and $H \setminus H'$ are strictly smaller than $H$, we can apply the inductive hypothesis to complete the proof.
\end{proof}

The following lemma shows that, for $G$ prefix-free, $\SigmaAntiSym|_G^\perp$ has a basis of constraints of the form ``$\sum_{\vec h \in H} w(\vec h) = 0$'', taken over all minimal symmetric subsets $H \subseteq G$.

\newcommand{\MinSymSets}{\mathcal{H}}
\begin{lemma}
	\label{lemma:sigmaas-basis}
	Let $G \subseteq \SumCompletion$ be prefix-free. For $H \subseteq G$, let $1_H \colon G \to \Field$ be the indicator vector for $H$, i.e., $1_H(\vec h) = 1$ if $\vec h \in H$ and $0$ otherwise. A basis for $\SigmaAntiSym[\SumDomain]|_G^{\perp}$ is $\{1_H\}_{H \in \MinSymSets}$ where $\MinSymSets$ is the set of all minimal symmetric subsets of $G$.
\end{lemma}
\begin{proof}
    Let $B \eqdef \{1_H\}_{H \in \MinSymSets}$. First, note that \cref{prop:minimal-symmetric-decomp} implies that for all $H \neq H' \in \MinSymSets$, $H \cap H' = \varnothing$. Hence $B$ is a linearly independent set. It is also straightforward to see that $B \subseteq \SigmaAntiSym[\SumDomain]|_G^{\perp}$. It remains to show that $\Span(B) = \SigmaAntiSym[\SumDomain]|_G^{\perp}$.

    Suppose that $z \in \SigmaAntiSym[\SumDomain]|_G^{\perp}$; equivalently, by \cref{prop:antisym}, for all $f \colon \SumDomain \to \Field$,
	\begin{equation*}
		\sum_{\vec a \in G} z(\vec a) \sum_{\vec x \in \vec a} f(\vec x) = \sum_{\vec a \in G} z(\vec a) \sum_{\vec x \in \vec a} f(\revvec x)~.
	\end{equation*}
	Let $S_\gamma \eqdef \{ \vec a \in G : z(\vec a) = \gamma \}$. Then we can write
	\begin{equation*}
		\sum_{\gamma \in \Field} \gamma \sum_{\vec a \in S_\gamma} \sum_{\vec x \in \vec a} f(\vec x) = \sum_{\gamma \in \Field} \gamma \sum_{\vec a \in S_\gamma} \sum_{\vec x \in \vec a} f(\revvec x)~.
	\end{equation*}
	Fix $\gamma^* \in \Field$, $\vec a^* \in S_{\gamma^*}$, and let $\vec x^* \in \vec a^*$. Let $f = 1_{\vec x^*}$ be the indicator function for $\vec x^*$. It follows that
	\begin{equation*}
		\gamma^* = \sum_{\gamma \in \Field} \gamma \sum_{\vec a \in S_\gamma} \sum_{\vec x \in \vec a} 1_{\vec x^*}(\revvec x) = \gamma'
	\end{equation*}
	where $\gamma' \in \Field$ is unique such that $\revvec{x}^* \in \cup S_{\gamma'}$. Thus $\revvec{x}^* \in \cup S_{\gamma^*}$; equivalently, $\vec{x}^* \in \cup (S_{\gamma^*})_{\rev}$.  It follows that for all $\gamma$, $\cup S_{\gamma} \subseteq \cup (S_{\gamma})_\rev$; a similar argument establishes that in fact $S_{\gamma}$ is symmetric. By \cref{prop:minimal-symmetric-decomp}, $S_{\gamma}$ is a union of minimal symmetric sets. It follows that $z \in \Span(B)$.
\end{proof}

Next, we prove the main technical lemma of this section.

\begin{lemma}
	\label{lemma:reverse-set-bounds}
    Let $H,G \subseteq \SumCompletion$ be prefix-free with $|H| \cdot |G| \leq |\SumDomain|/4$. Suppose that $\cup H = \cup G_\rev$. Then
	\begin{equation*}
		|\cup H| \geq \frac{K}{2}(1 + \sqrt{1 - 4|H|\cdot|G|/K}) \qquad \text{or} \qquad |\cup H| \leq \frac{K}{2}(1 - \sqrt{1 - 4|H|\cdot|G|/K})~.
	\end{equation*}
	In particular,
	\begin{equation*}
		\min(|\cup H|, |\SumDomain| - |\cup H|) \leq |H| \cdot |G| + O(|H|^2|G|^2/|\SumDomain|)~.
	\end{equation*}
\end{lemma}

Taking $G = H$, we obtain strong bounds on the support size of symmetric sets: if $H$ is symmetric then $\cup H$ must be of size either $\Poly(|H|)$ or $|\SumDomain| - \Poly(|H|)$.

\begin{corollary}
    \label{cor:symmetric-set-size}
    If $H \subseteq \SumCompletion$ is symmetric and prefix-free, then
    \[
		\min(|\cup H|, |\SumDomain| - |\cup H|) \leq |H|^2 + O(|H|^4/|\SumDomain|)~.
    \]
\end{corollary}

Before proving \cref{lemma:reverse-set-bounds}, we state and prove a simple but crucial claim about the size of the intersection of sets $\vec a$ and $\revvec b$ for general $\vec a, \vec b \in \SumCompletion$.

\begin{claim}
    For any $\vec a, \vec b \in \SumCompletion$,
	\begin{equation*}
		|\vec a \cap \revvec b| \leq \begin{cases}
			\prod_{j=\VLen{a}+1}^{\NumVars - \VLen{b}} |\SumSet_j|
			& \text{if $\VLen{a} + \VLen{b} < \NumVars$} \\
			1 & \text{otherwise.}
		\end{cases}
	\end{equation*}
\end{claim}
\begin{proof}
	Recall that, if $\VLen{a} + \VLen{b} < \NumVars$,
	\begin{align*}
		\vec a &= \{ a_1 \} \times \cdots \times \{ a_{\VLen{a}} \} \times \SumSet_{\VLen{a}+1} \times \cdots \times \SumSet_{\NumVars-\VLen{b}} \times \cdots \times \SumSet_{\NumVars} \\
		\revvec b &= \SumSet_{1} \times \cdots \times \SumSet_{\VLen{a}} \times \cdots \times \SumSet_{\NumVars-\VLen{b}} \times \{ b_{\VLen{b}} \} \times \cdots \times \{ b_{1} \}
	\end{align*}
	and so
	\begin{equation*}
		\vec a \cap \revvec b = \{ a_1 \} \times \cdots \times \{ a_{\VLen{a}} \} \times \SumSet_{\VLen{a}+1} \times \cdots \times \SumSet_{\NumVars-\VLen{b}} \times \{ b_{\VLen{b}} \} \times \cdots \times \{ b_{1} \}~.
	\end{equation*}
	
	On the other hand, if $\VLen{a} + \VLen{b} \geq \NumVars$ then
	\begin{equation*}
		\vec a \cap \revvec b = \begin{cases}
 			\{ (a_1, \ldots, a_{\VLen{a}}, b_{\NumVars - \VLen{a}}, \ldots, b_{1}) \} & \text{if $a_i = b_{\NumVars - i + 1}$ for all $\NumVars + 1 - \VLen{b}  \leq i \leq \VLen{a}$;} \\
 			\varnothing & \text{otherwise.}
 	\end{cases} \qedhere
 	\end{equation*}
\end{proof}

\begin{proof}[Proof of \cref{lemma:reverse-set-bounds}]
	Since $\cup H = \cup G_\rev$, $(\cup H) \cap (\cup G_\rev) = \cup H = \cup G_\rev$. Since $H, G$ are prefix-free,
	\begin{equation*}
		|\cup H| = |(\cup H) \cap (\cup G_\rev)| = \sum_{\vec a \in H} \sum_{\vec b \in G} |\vec a \cap \revvec b| \leq |H|\cdot|G| + \sum_{\substack{\vec a \in H,\vec b \in G \\ \VLen{a} + \VLen{b} < \NumVars}} \prod_{j=\VLen{a}+1}^{\NumVars - \VLen{b}} |\SumSet_j| =: |H| \cdot |G| + N ~.
	\end{equation*}
	
	Now observe that
	\begin{align*}
		|\cup H|^2 = |\cup H| \cdot |\cup G_\rev| &= (\sum_{\vec a \in H} \prod_{j = \VLen{a}+1}^{n} |\SumSet_j|) \cdot (\sum_{\vec b \in G} \prod_{j = 1}^{n-\VLen{b}} |\SumSet_j|) \\
		&= \sum_{\vec a \in H} \sum_{\vec b \in G} \prod_{j = \VLen{a}+1}^{n} |\SumSet_j| \cdot \prod_{j = 1}^{\NumVars-\VLen{b}} |\SumSet_j| \\
		&\geq |\SumDomain| \cdot \sum_{\substack{\vec a \in H,\vec b \in G \\ \VLen{a} + \VLen{b} < \NumVars}}  \prod_{j = \VLen{a}+1}^{\NumVars - \VLen{b}} |\SumSet_j| = |\SumDomain| \cdot N.
	\end{align*}
	Hence, writing $t \eqdef |H| \cdot |G|$, $|\SumDomain| \cdot N \leq (t + N)^2$. Completing the square to solve for $N$, we obtain $(N + t - |\SumDomain|/2)^2 \geq \frac{|\SumDomain|^2}{4}(1-4t/|\SumDomain|)$. Hence either
	\begin{equation*}
		N \geq \frac{|\SumDomain|}{2}(1 + \sqrt{1 - 4t/|\SumDomain|}) - t \qquad \text{or} \qquad N \leq \frac{|\SumDomain|}{2}(1 - \sqrt{1 - 4t/|\SumDomain|}) - t
	\end{equation*}
	The first part of the lemma follows by recalling that $|\cup H| = N + t$. The second part then follows since $\sqrt{1 - 4t/|\SumDomain|} = 1 - 2t/|\SumDomain| - O(t^2/|\SumDomain|^2)$.
\end{proof}

\subsection{Algorithms for \texorpdfstring{$\SigmaAntiSym$}{Sigma-AntiSym}}
\label{sec:antisym-algs}

In this section, we present two useful algorithms for working with $\SigmaAntiSym$. The first (\cref{lem:disjoint-alg}) takes as input an arbitrary set $I \subseteq \SumCompletion$ and outputs a prefix-free set $G \subseteq \SumCompletion$ of a similar size with the property that any $\vec a \in I$ can be obtained as the union of sets in $G$. The second (\cref{lem:rev-alg}) takes as input a prefix-free set $G$ and outputs the set of minimal symmetric subsets of $G$ (i.e., a basis for $\Dual{\SigmaAntiSym[\SumDomain]|_G}$).

\begin{lemma}
	\label{lem:disjoint-alg}
	There is a polynomial-time algorithm $\PrefixFree$ which, given as input a set $I \subseteq \SumCompletion$, outputs a \emph{prefix-free} set $G$ of size at most $|I| \cdot \NumVars$ and a list $(\Lambda_{\vec a} \subseteq G)_{\vec a \in I}$ such that for each $\vec a \in I$, $\cup \Lambda_{\vec a} = \vec a$.
\end{lemma}
\begin{proof}
    For $\vec a \in \SumCompletion$ and $G \subseteq \SumCompletion$, define $N_G(\vec a) \eqdef \{ \vec b \in G : \vec b \subsetneq \vec a \}$. The algorithm $\PrefixFree$ operates as follows.
    \begin{mdframed}[nobreak=true]
	\begin{enumerate}
		\item For each $\vec a \in I$, set $\Lambda_{\vec a} \eqdef \{ \vec a \}$. Set $G \eqdef I$.
		\item \label{step:disj-alg-find-term} Let $\vec a^* = (a_1,\ldots,a_{\Length^*})$ be a shortest element in $G$ (i.e., $\Length^*$ is minimal) for which $N_G(\vec a^*)$ is nonempty. If there is no such term, output $G$.
		\item \label{step:disj-alg-replace} Let $\vec a' = (a_1,\ldots,a_{\Length'})$ be a shortest element in $N_G(\vec a^*)$. (Note $\Length' > \Length^*$.)
		\item Remove $\vec a^*$ from $G$.
		\item \label{step:disj-alg-add-terms} For each $j \in \{ \Length^*+1, \ldots, \Length' \}$, and for each $b_j \in \SumSet_j \setminus \{a_j\}$, add $\vec \alpha_j \eqdef (a_1,\ldots,a_{j-1},b_j)$ to $G$.
		\item For each $\vec a \in I$ such that $\vec a^* \in \Lambda_{\vec a}$, remove $\vec a^*$ from $\Lambda_{\vec a}$ and add $\{ \vec \alpha_{\ell^*+1},\ldots,\vec \alpha_{\ell'}, \vec a' \}$.
		\item Go to \cref{step:disj-alg-find-term}.
	\end{enumerate}
    \end{mdframed}
    Clearly if the algorithm terminates then the correctness condition is satisfied, and so it remains to bound the number of iterations. Denote by $\Delta_i$ the value of $|\cup_{\vec a \in G} N_G(\vec a)|$ at the beginning of the $i$-th iteration. Clearly $\Delta_1$ cannot be larger than $|I|$, and if $\Delta_t = 0$ then the algorithm terminates at the beginning of the $t$-th iteration. We show that $\Delta$ is a progress measure for the above algorithm.
	
	\begin{claim}
		If $\Delta_i > 0$ then $\Delta_{i+1} < \Delta_i$.
	\end{claim}
	\begin{proof}
		Let $\vec a'$ be the term chosen in \cref{step:disj-alg-replace} of the $i$-th iteration. We show that at the termination of the $i$-th iteration, $\vec a'$ has been removed from $\cup_{\vec a \in G} N_G(\vec a)$, and no element has been added. By choice of $\vec a'$, $\vec a' \in N_G(\vec a^*)$, and $\vec a' \notin N_G(\vec a'')$ for any $\vec a'' \neq \vec a^*$, since that would mean $\vec a' \subset \vec a'' \subset \vec a^*$. Moreover, for every element $\vec \alpha_j$ added to $G$ in \cref{step:disj-alg-add-terms}, $\vec a' \notin N_G(\vec \alpha_j)$ since $\vec \alpha_j$ is disjoint from $\vec a'$ by construction, $N_G(\vec \alpha_j) \subset N_G(\vec a)$ since $\vec \alpha_j \subset \vec a$, and $\vec \alpha_j \notin \cup_{\vec a \in G} N_G(\vec a)$ by choice of $\vec a^*,\vec a'$.
	\end{proof}
	
	It follows that the number of iterations is at most $|I|$; since each iteration clearly runs in polynomial time, the algorithm runs in polynomial time. The bound on $|G|$ is obtained by noting that, in each iteration, the size of $G$ increases by at most $\NumVars - 1$.
\end{proof}

\begin{lemma}
    \label{lem:rev-alg}
    There is a polynomial-time algorithm $\SymSets$ which, given as input a prefix-free set $G \subseteq \SumCompletion$, outputs the set $\MinSymSets$ of all minimal symmetric subsets $H$ of $G$.
\end{lemma}
\begin{proof}
    The algorithm $\SymSets$ operates as follows:
	\begin{enumerate}
		\item Construct a graph $\Gamma = (G,E)$ where $(\vec a, \vec b) \in E$ if and only if $\vec a \cap \revvec b \neq \varnothing$.
		\item Compute $C(\Gamma)$, the set of connected components of $\Gamma$.
		\item Output the set $\{H \in C(\Gamma) : \cup H = \cup H_\rev \}$.
	\end{enumerate}
	It is straightforward to construct $\Gamma$ and compute $C(\Gamma)$ in polynomial time. To determine whether $\cup H = \cup H_\rev$, it suffices to compute $K \eqdef |(\cup H) \cap (\cup H_\rev)|$ and then check whether $|\cup H| = |\cup H_\rev| = K$; this can also be achieved in polynomial time since $H$ is prefix-free. The correctness of the algorithm is a consequence of the following claim.
	
	\begin{claim}
		A set $H$ is minimal symmetric if and only if $H \in C(\Gamma)$ and $\cup H = \cup H_\rev$.
	\end{claim}
	\begin{proof}
		First, suppose that $H$ is minimal symmetric. Then it remains to show that $H$ is a connected component in $\Gamma$. First, suppose that there is some edge $(\vec a, \vec b) \in E$ with $a \in H$ and $b \notin H$. Since $J$ is prefix-free, $\vec b \cap (\cup H) = \varnothing$. However, since $\revvec a \in H_\rev$ and $\revvec a \cap \vec b \neq \varnothing$, we have that $\vec b \cap (\cup H_\rev) \neq \varnothing$, which is a contradiction since $\cup H = \cup H_\rev$.
		
		Hence $H$ must be a disjoint union of $k$ connected components $H_1,\ldots,H_k$. Now consider two cases: either $\cup H_1 = \cup (H_1)_\rev$, or not. If $\cup H_1 \neq \cup (H_1)_\rev$, take $\vec x \in \cup H_1 \setminus \cup (H_1)_\rev$ (so $\vec x \in \SumDomain$). Then $\revvec x \notin \cup H_i$ for all $i \in \{2,\ldots,k\}$, otherwise there would be an edge from $H_1$ to $H_i$. Hence $\vec x \notin \cup H_\rev$; but then $\cup H \neq \cup H_\rev$, which contradicts that $H$ is symmetric. Hence $\cup H_1 = \cup (H_1)_\rev$; so since $H$ is minimal, $H_1 = H$, and  $k = 1$.
		
		Now suppose conversely that $H$ is a connected component in $\Gamma$ with $\cup H = \cup H_\rev$, and let $H' \subseteq H$ be minimal symmetric. By the above, $H'$ is a connected component in $\Gamma$ contained in $H$, whence $H' = H$.
	\end{proof}
	\noindent This completes the proof.
\end{proof}

\subsection{Proof of \cref{theorem:antisym-cl}}

    We construct a constraint locator $\CL_{\SigmaAntiSym}$ for $\Enc_{\SigmaAntiSym}$.
	\begin{mdframed}[nobreak=true]
		\begin{construction}
			$\CL_{\SigmaAntiSym}(I)$:
			\begin{enumerate}
				\item Compute $(G,(\Lambda_{\vec a})_{\vec a \in I}) \eqdef \PrefixFree(I)$.
				\item Compute $\MinSymSets \eqdef \SymSets(G)$. Let $\MinSymSets_0 \eqdef \{ H \in \MinSymSets : |\cup H| \leq |\SumDomain|/2 \}$, and $\MinSymSets_1 \eqdef \{ H \in \MinSymSets : |\cup H| > |\SumDomain|/2 \}$. Let $R \eqdef \{ \bot \} \cup \bigcup_{H \in \MinSymSets_0} (\cup H) \cup \bigcup_{H \in \MinSymSets_1} (\SumDomain \setminus \cup H)$.
                    \item Let $B \eqdef (1_H)_{H \in \mathcal{H}} \in \Field^{|\MinSymSets| \times G}$, $C \eqdef ((1_{\cup H})_{H \in \MinSymSets_0}, (1_{\bot} - 1_{\SumDomain \setminus \cup H})_{H \in \MinSymSets_1}) \in \Field^{|\MinSymSets| \times R}$, $M \eqdef (1_{\Lambda_{\vec a}})_{\vec a \in I} \in \Field^{I \times G}$, and $\mathbf{I} \in \Field^{I \times I}$ be the identity matrix. Construct the block matrix
                    \[ Y = \begin{pmatrix}
                        C & -B & 0 \\
                        0 & M & -\mathbf{I}
                    \end{pmatrix} \in \Field^{(|\MinSymSets| + |I|) \times (R \sqcup G \sqcup I)} ~. \]
				\item Output a basis $Z$ for the space $\{ (\vec \alpha, \vec \gamma) \in \Field^{R} \times \Field^{I} : \exists \vec \beta,\, Y (\vec \alpha, \vec \beta, \vec \gamma)^T = 0 \}$.
			\end{enumerate}
		\end{construction}
	\end{mdframed}
        By \cref{cor:symmetric-set-size}, $|R| \leq 1 + \sum_{H \in \MinSymSets} |H|^2 + O(|H|^4/|\SumDomain|) \leq 1 + |G|^2 + O(|G|^4/|\SumDomain|)$. The bound on $\ell(n)$ follows since $|G| \leq \NumVars \cdot |I|$. Efficiency follows from the efficiency of $\PrefixFree$ and $\SymSets$.
 
        It remains to show that for all $f \in \MsgSpace, \vec \gamma \in \Field^I$, it holds that $\vec \gamma \in \supp(\Enc_{\SigmaAntiSym}(f)|_I)$ if and only if $(f|_R, \vec \gamma)^T \in \ker(Z)$. Note that, by \cref{lemma:sigmaas-basis}, $\vec \beta \in \supp(\Enc_{\SigmaAntiSym}(f)|_G)$ if and only if $B \vec \beta = B(\SumWord{}{f}|_G)$. For all $f$ and all $H \in \MinSymSets_0$, $1_H \cdot \Sigma[f]|_G = 1_{\cup H} \cdot f|_R$. Similarly, for all $H \in \MinSymSets_1$, $1_H \cdot \Sigma[f]|_G = (1_{\bot} - 1_{\SumDomain \setminus \cup H}) \cdot f|_R$. Hence $B(\SumWord{}{f}|_G) = C (f|_R)$.
        
        By the guarantee of \cref{lem:disjoint-alg}, we have that $\cup \Lambda_{\vec a} = \vec a$ for all $\vec a \in I$, and so for all functions $\Phi \colon \SumDomain \to \Field$, $M(\SumWord{}{\Phi}|_G) = \SumWord{}{\Phi}|_I$. Hence $\vec \gamma \in \supp(\Enc_{\SigmaAntiSym}(f)|_I)$ if and only if there exists $\vec \beta \in \supp(\Enc_{\SigmaAntiSym}(f)|_G)$ such that $\vec \gamma = M \vec \beta$.
        
        It follows that $\vec \gamma \in \supp(\Enc_{\SigmaAntiSym}(f)|_I)$ if and only if there exists $\vec \beta \in \supp(\Enc_{\SigmaAntiSym}(f)|_G)$ such that $B \vec \beta = C(f|_R)$ and $\vec \gamma = M \vec \beta$. By construction, this is equivalent to the existence of $\vec \beta \in \Field^G$ such that $(f|_R,\vec \beta,\vec \gamma)^T \in \ker(Y)$, and such a $\vec \beta$ exists if and only if $(f|_R,\vec \gamma)^T \in \ker(Z)$.


\newcommand{\ZKPCPEnc}{\Enc_{\PCP}}
\newcommand{\ZKPCPCodeSim}{\CodeSimulator}

\section{ZKPCP for \texorpdfstring{$\SharpP$}{\#P}}

In this section we prove our main theorem, that there exists a PZK-PCP for (decision) $\SharpSAT$. First, we recall the definition:

\[ \SharpSAT = \{ (\Phi,N) : \text{$\Phi$ is a CNF, } \sum_{x \in \Bits^n} \Phi(x) = N \} . \]

\begin{theorem}
\label{thm:main}
    There exists a PZK-PCP for $\SharpSAT$. The honest verifier is non-adaptive, and the zero-knowledge property holds against arbitrary (adaptive) polynomial-time malicious verifiers.
\end{theorem}

To formally obtain \cref{thm:main_informal}, we must define what we mean by $\SharpP$ as a class of languages (typically, $\SharpP$ is treated as a class of function problems). One natural way is by reference to the number of accepting paths of a Turing machine, i.e.: a pair language $\Language \subseteq (\Bits^* \times \N)$ is in $\SharpP$ if and only if there is a polynomial-time nondeterministic Turing machine $M$ such that $x \in \Language$ if and only if $M(x)$ has $N$ accepting paths. $\SharpSAT$ is complete for $\SharpP$ under this definition.

One can also adopt a more permissive definition of $\SharpP$ as the set of all languages that have a deterministic polynomial-time reduction to $\SharpSAT$. This permits, for example, the inclusion $\mathsf{coNP} \subseteq \SharpP$, via reduction from $\mathsf{UNSAT}$: $\Phi \mapsto (\Phi,0)$. \cref{thm:main_informal} follows from \cref{thm:main}  also with respect to this broader definition.

\begin{remark}[Sampling field elements]
    Our PZK-PCP construction relies on sampling uniformly random field elements. This can be an issue if we demand perfect simulation, as typically PPT algorithms are given a tape of uniformly random \emph{bits}, from which it is not possible to perfectly sample random field elements in strict polynomial time if the field has characteristic different from $2$. There are three ways to resolve this: \begin{inparaenum}[(1)]
        \item provide the simulator with uniformly random field elements on its random tape;
        \item use rejection sampling, so that the simulator runs in expected polynomial time, or;
        \item restrict attention to $\oplus P$, which can be verified by sumcheck over fields of characteristic $2$.
    \end{inparaenum}
\end{remark}

\subsection{Zero-Knowledge Probabilistically Checkable Proofs (ZK-PCPs)}

In this section we formally define the notion of a (perfect) zero-knowledge probabilistically-checkable proof.

\begin{definition}[PCP]
    A \defemph{probabilistically checkable proof (PCP)} for a language $\Language$ consists of a prover $\Prover$ and a PPT verifier $\Verifier$ such that the following holds.
    \begin{enumerate}
        \item \parhead{Completeness} For every $\Instance \in \Language$,
        \begin{equation*}
            \Pr_{\Proof \gets \Prover(\Instance)}[\Verifier^{\Proof}(\Instance) = 1] = 1.
        \end{equation*}

        \item \parhead{Soundness} For every $\Instance \notin \Language$ and every oracle $\MalProof$
        \begin{equation*}
            \Pr[\Verifier^{\MalProof}(\Instance) = 1] \leq \frac{1}{2}.
        \end{equation*}
    \end{enumerate}
\end{definition}

\begin{definition}[View]
	For a PCP $(\Prover, \Verifier)$ and a (possibly malicious) verifier $\MalVerifier$, $\View_{\MalVerifier,\Prover}(\Instance)$ denotes the view of $\MalVerifier$ with input $\Instance$ and oracle access to $\Proof \gets \Prover(\Instance)$. That is, $\View_{\MalVerifier,\Prover}(\Instance)$ comprises $\MalVerifier$'s random coins and all answers to $\MalVerifier$'s queries to $\Proof$. 
\end{definition}

\begin{definition}[Perfect-Zero-Knowledge PCP]
    We say that a PCP system $(\Prover, \Verifier)$ for a language $\Language$ is \defemph{perfect zero-knowledge} (a PZK-PCP) if there exists an (expected) PPT algorithm $\Simulator$ (the simulator), such that for every (possibly malicious) adaptive polynomial-time verifier $\MalVerifier$, and for every $\Instance \in \Language$, $\Simulator^{\MalVerifier}(\Instance)$ is distributed identically to $\View_{\MalVerifier,\Prover}(\Instance)$.
\end{definition}

\subsection{Our Construction}

We define the following randomised encoding function. 
\begin{definition}
    The linear randomised encoding $\ZKPCPEnc \colon \Field^\SumDomain\to \Field^{\Field^{\leq \NumVars}}$ is given by $\ZKPCPEnc(\Message) \eqdef \SigRMEnc{\vec{\Degree}}{\SumDomain}\circ \SigmaAntiSymEnc(\Message)$, for $\vec{\Degree} \in \Field^{\NumVars}$ and $\SumDomain \subseteq \Field^{\NumVars}$ such that $\Degree_i \geq 2(|\SumSet_i| - 1)$ for all $i \in [\NumVars]$.
\end{definition}

\begin{corollary}
	\label{theorem:zk-pcp-local-sim}
	The encoding function $\ZKPCPEnc$ is $(t,\ell)$-locally simulatable, for $\ell(n) = \ell_{\SigmaRM}(\ell_{\AntiSym}(n))$ and $t(n) = \Poly(\log|\Field|, \NumVars, \max_i \Degree_i, n)$, where $\ell_{\SigmaRM}(n) \eqdef n\NumVars (\NumVars(a + 1) + 1)^2$, for $a \eqdef \max_i \SumSet_i$ and $\ell_\AntiSym(n) \eqdef  n^2\NumVars^2 + 1 + O(n^4\NumVars^4/|\SumDomain|)$.
\end{corollary}

\begin{proof}
By \Cref{thm:sigrm-cl}, there is an $\ell_{\SigmaRM}$-constraint locator for $\SigRMEnc{\vec{\Degree}}{\SumDomain}$. By \Cref{theorem:antisym-cl} there is an $\ell_{\AntiSym}$-constraint locator for $\SigmaAntiSymEnc$. Thus \Cref{claim:cl-composes} and \Cref{claim:cl-implies-local-sim} imply the result.
\end{proof}

\begin{theorem}[Low individual degree test \cite{GoldreichS06,GurR15}]
    \label{thm:ldt}
    Let $\NumVars \in \N, \vec \Degree \in \N^{\NumVars}$ be such that $\sum_i \Degree_i < |\Field|/10$, $\varepsilon \in (0,1/10)$ and $\delta \in [0,1]$. There exists an efficient test that, given oracle access to a function $P \colon \Field^\NumVars \to \Field$, makes $O(\NumVars \Degree \cdot \Poly(1/\varepsilon) \cdot \log(1/\delta))$ queries to $P$, and:
    \begin{itemize}
        \item if $P \in \ReedMuller[\Field,\NumVars,\vec \Degree]$ then the test accepts with probability $1$;
        \item if $P$ is $\varepsilon$-far from $\ReedMuller[\Field,\NumVars,\vec \Degree]$ then the test accepts with probability at most $\delta$.
    \end{itemize}
\end{theorem}
	
\newcommand{\MaskPoly}{R}
\newcommand{\ZeroP}{Z}

\begin{mdframed}[nobreak=true]
\begin{construction}
	\label{construction:zkpcp}
    A PZK-PCPP for sumcheck. Both parties receive the common input $(\Field, \NumVars, \Degree, \Subcube, \SumVal)$, and oracle access to the evaluation table of $\SCPoly \colon \Field^\NumVars \to \Field$. The proof proceeds as follows. 
    
    \noindent\textbf{Proof:} 
    \begin{enumerate}[nolistsep]
        \item  Sample a polynomial $\RandPoly \gets \Polys{\Field}{\vec{\Degree}}{\NumVars}$ uniformly at random. Compute the full evaluation table $\Proof_\RandPoly$ of $\RandPoly$.

        \item For each $i \in [\NumVars]$, sample $\RandLDPoly_i \gets \Polys{\Field}{\vec{\Degree}_i}{\NumVars}$, where $\vec{\Degree}_i = (\Degree, \dots, \Degree-|\Subcube|, \dots, \Degree)$ is the vector which takes the value $\Degree$ in every coordinate except for the $i$-th location which takes the value $\Degree-|\Subcube|$. Compute the full evaluation table $\Proof_{\RandLDPoly_i}$ of $\RandLDPoly_i$.

        \item Define $\MaskPoly(\vec{X}) \eqdef \RandPoly(\vec{X}) - \RandPoly_\rev(\vec{X}) + \sum_{i=1}^{\NumVars} \ZeroP_\Subcube(X_i) \RandLDPoly_i(\vec{X})$, where $\ZeroP_\Subcube \eqdef \prod_{a \in \Subcube} (X - a)$, and $\RandPoly_\rev(\vec{X}) \eqdef \RandPoly(\revvec{X})$.

        \item Compute $\Proof_\Sigma \eqdef \SumWord{\Subcube^\NumVars}{\SCPoly+\MaskPoly} \in \Field^{\leq \NumVars}$.

        \item Output $\Proof \eqdef (\Proof_\Sigma,\Proof_{\RandPoly}, \Proof_{\RandLDPoly_1}, \dots, \Proof_{\RandLDPoly_\NumVars})$.
    \end{enumerate}
    \vspace{0.5cm}
    \noindent\textbf{Verifier:} 
    \begin{enumerate}[nolistsep]
        \item Run the sumcheck verifier on $\Proof$, yielding a claim ``$(F+R)(\vec \alpha) = \beta$'' for a uniformly random $\vec \alpha \in \Field^\NumVars$. Query $F(\vec \alpha)$. Evaluate $R(\vec \alpha)$ by querying $\Proof_{\RandLDPoly_1}, \dots, \Proof_{\RandLDPoly_\NumVars}$ at $\vec \alpha$, and $\Proof_{\RandPoly}$ at $\vec \alpha, \revvec \alpha$. If $F(\vec \alpha) + R(\vec \alpha) \neq \beta$, reject.
        \item Perform a low individual degree test (\cref{thm:ldt}) on each of $\Proof_{\RandLDPoly_1}, \dots, \Proof_{\RandLDPoly_\NumVars}, \Proof_{\RandPoly}$, with $\varepsilon = \frac{1}{4(m+2)}$, $\delta = 1/2$, and reject if any test rejects.
        \item If none of the above tests reject, then accept.
    \end{enumerate}
\end{construction}
\end{mdframed}

\newcommand{\RMCD}{\mathsf{CD}}

Let $\ZKPCPCodeSim$ be the local simulator guaranteed by \cref{theorem:zk-pcp-local-sim}, and let $\RMCD$ be the $\TimeBound$-constraint detector for $\ReedMuller$ guaranteed by \cref{thm:rm-constraint-detector}. 
\begin{mdframed}[nobreak=true]
	\begin{construction}[Simulator]
        A simulator for \Cref{construction:zkpcp}. It receives as input $(\Field, \NumVars, \Degree, \Subcube, \SumVal, \RandSumVal)$, and oracle access to $\SCPoly \in \Polys{\Field}{\vec{\Degree}}{\NumVars}$.

        \noindent \underline{$\Simulator^{F}(\OracleTable_{\Sigma},\OracleTable_{\RandPoly},\OracleTable_{T_1},\ldots,\OracleTable_{T_\NumVars},(\Oracle, \vec \alpha))$:}
		\begin{enumerate}[nolistsep]
			\item Sample $\beta \gets \ZKPCPCodeSim^{F}(\OracleTable_\Sigma, \vec{\alpha})$, and store $(\vec{\alpha},\beta)$ in $\OracleTable_\Sigma$.
			\item If $\vec{\alpha} \in \Field^{\leq \NumVars}$ is a query to $\Proof_{\Sigma}$, return $\beta$.

            \item Otherwise, for each $i \in [\NumVars]$, compute $Z_{T_{i}} \eqdef \RMCD_{\vec{\Degree_i}}(\supp(\OracleTable_{T_i})\Union \{\vec{\alpha}\})$ and $ Z_{\RandPoly} \eqdef \RMCD_{\vec{\Degree}}(\supp(\OracleTable_{\RandPoly}) \Union \{\vec{\alpha}, \vec{\alpha}_{\rev}\})$.

            \item Let $z_{\Sigma,\vec \alpha} \eqdef (1,-1,\ZeroP_\Subcube(\alpha_1),\ldots,\ZeroP_\Subcube(\alpha_\NumVars))$. For each $i \in [\NumVars]$, let $\gamma_i, \gamma_{\RandPoly}, \gamma_{\RandPoly_{\rev}}\in \Field$ be uniformly random solutions to the linear system:
            \begin{align*}
			z_{\Sigma,\vec \alpha} \cdot (\gamma_{\RandPoly}, \gamma_{\RandPoly_{\rev}}, \gamma_1,\ldots,\gamma_\NumVars) &= \beta - F(\vec{\alpha}) \\
                Z_{\RandPoly} (\OracleTable_{\RandPoly}, \gamma_{\RandPoly}, \gamma_{\RandPoly_{\rev}})^T &= 0 \\
                \forall i \in [\NumVars],\, Z_{\RandLDPoly_{i}}(\OracleTable_{\RandLDPoly_{i}}, \gamma_i)^T &= 0
		\end{align*}
            where we view each $\OracleTable$ as a vector in $\Field^{\supp(\OracleTable)}$.
            
			\item For each $i$, add $(\vec \alpha, \gamma_i)$ to $\OracleTable_{T_i}$. Add $(\vec\alpha, \gamma_\RandPoly)$ and $(\revvec\alpha, \gamma_{\RandPoly_\rev})$ to $\OracleTable_\RandPoly$.
			\item Return $\OracleTable_{\Oracle}(\vec \alpha)$.
		\end{enumerate}
	\end{construction}
\end{mdframed}

\begin{lemma}
	\label{lem:zkpcp-sumcheck}
    \cref{construction:zkpcp} is a perfect zero-knowledge PCP of proximity for sumcheck. That is, if $\NumVars\Degree < |\Field|/10$, and $\Degree \geq |\Subcube|+1$, then
    \begin{itemize}
        \item \textbf{Completeness:} if $\SCPoly \in \ReedMuller[\Field,\NumVars,\Degree]$ and $\sum_{\vec a \in \Subcube^\NumVars} \SCPoly(\vec a) = \gamma$, then the verifier accepts $\Proof$ with probability $1$;
        \item \textbf{Soundness:} if $\SCPoly \in \ReedMuller[\Field,\NumVars,\Degree]$ and $\sum_{\vec a \in \Subcube^\NumVars} \SCPoly(\vec a) \neq \gamma$, then for any proof $\MalProof$, the probability that the verifier accepts $\MalProof$ is at most $1/2$;
        \item \textbf{Perfect zero knowledge:} there exists an algorithm $\Simulator$ running in time $\Poly(\log |\Field|, \NumVars, \Degree, |\Subcube|)$ such that, for any (malicious) PPT verifier $\MalVerifier$, $\Simulator^{\SCPoly,\MalVerifier}(\Field,\NumVars,\Degree,\Subcube,\SumVal)$ outputs the distribution $\View_{\MalVerifier,\Prover}(\Instance)$.
    \end{itemize}
    The proof is of length $O(\log |\Field| \cdot \NumVars \cdot |\Field|^\NumVars)$.
\end{lemma}
\begin{proof}
Completeness is straightforward; we prove soundness and zero knowledge.

\parhead{Soundness}
Suppose that $\SCPoly$ is such that $\sum_{\vec a \in \Subcube^\NumVars} \SCPoly(\vec a) = \tilde \gamma \neq \gamma$, and let $\Proof^* =(\Proof_\Sigma,\Proof_{\RandPoly},\Proof_{\RandLDPoly_1},\ldots,\Proof_{\RandLDPoly_\NumVars})$.  If any of $\Proof_{\RandPoly},\Proof_{\RandLDPoly_1},\ldots,\Proof_{\RandLDPoly_\NumVars}$ are more than $\varepsilon$-far from their corresponding $\ReedMuller$ codes, then the verifier rejects with probability $1/2$ as required. Otherwise, there exist polynomials $\RandPoly \in \Polys{\Field}{\Degree}{\NumVars}$ and $\RandLDPoly_i \in \Polys{\Field}{\vec \Degree_i}{\NumVars}$ for each $i \in [\NumVars]$ whose evaluations over $\Field^\NumVars$ are each $\varepsilon$-close to their corresponding codewords in $\Proof^*$.

Let $\hat{\MaskPoly} \eqdef \RandPoly - \RandPoly_{\rev} + \sum_{i=1}^\NumVars X_i (1-X_i) \RandLDPoly_i$. Then $\hat{\MaskPoly} \in \Polys{\Field}{\Degree}{\NumVars}$ and, by construction, $\sum_{\vec a \in \Subcube^\NumVars} \SCPoly(\vec a) + \hat{\MaskPoly}(\vec a) = \tilde \gamma$. By the soundness of the (standard) sumcheck PCPP, with probability $1 - \frac{\NumVars\Degree}{|\Field|} \geq 3/4$ the verifier's claim ``$(\SCPoly + \hat{\MaskPoly})(\vec \alpha) = \beta$'' is false. Since $\vec \alpha$ is uniformly random in $\Field^\NumVars$, and the evaluation of $\MaskPoly$ queries each $\RandLDPoly_i$ at $\vec \alpha$ and $\RandPoly$ at $\vec \alpha,\revvec \alpha$, the probability that $\MaskPoly(\vec \alpha)$ as computed by the verifier is equal to $\hat{\MaskPoly}(\vec \alpha)$ is at least $3/4$. The stated soundness follows by a union bound.

\parhead{Zero knowledge}
    Fix $\Proof_\Sigma \in \supp(\ZKPCPEnc(\SCPoly|_{\Subcube^\NumVars}))$, and consider the following distribution $D(\Proof_\Sigma)$ on $(\Proof_{\RandPoly}, \Proof_{\RandLDPoly_1},\ldots,\Proof_{\RandLDPoly_\NumVars})$:
    \begin{center}
    \begin{minipage}{0.9\textwidth}
        \item[] Sample $\RandPoly \gets \Polys{\Field}{\vec{\Degree}}{\NumVars}$, $\RandLDPoly_i \gets \Polys{\Field}{\vec{\Degree}_i}{\NumVars}$ for $i \in [\NumVars]$, such that $\Proof_\Sigma(\vec \alpha) = \SCPoly(\vec \alpha) + \RandPoly(\vec \alpha) - \RandPoly(\revvec \alpha) + \sum_{i=1}^{\NumVars} \alpha_i (1 - \alpha_i) \RandLDPoly_i(\vec \alpha)$ for all $\vec \alpha \in \Field^{\NumVars}$. Let $\Proof_{\RandPoly}, \Proof_{\RandLDPoly_1},\ldots,\Proof_{\RandLDPoly_\NumVars}$ be the corresponding evaluation tables.
    \end{minipage}
    \end{center}
    
    \noindent By construction, $D(\Proof_{\Sigma})$ is the distribution of $(\Proof_{\RandPoly}, \Proof_{\RandLDPoly_1},\ldots,\Proof_{\RandLDPoly_\NumVars})$ produced by the sumcheck prover, conditioned on a fixed $\Proof_{\Sigma}$. For any $S \subseteq \Field^\NumVars$,
    \begin{align*}
        \Pr_{D(\Proof_{\Sigma})}\left[\begin{array}{c}
        \Proof_{\RandPoly}|_{S \cup S_{\rev}} = \OracleTable_{\RandPoly} \\ \forall i,\, \Proof_{\RandLDPoly_i}|_S = \OracleTable_{\RandLDPoly_i}
        \end{array}
        \right]
        &= \Pr\left[\begin{array}{c}
        Q|_{S \cup S_{\rev}} = \OracleTable_{\RandPoly} \\ \forall i,\, \RandLDPoly_i|_S = \OracleTable_{\RandLDPoly_i}
        \end{array}
        \middle\vert
        \begin{array}{r}
            \RandPoly \gets \Polys{\Field}{\vec{\Degree}}{\NumVars} \\
            \RandLDPoly_i \gets \Polys{\Field}{\vec{\Degree}_i}{\NumVars} \\[5pt]
            z_{\Sigma,\vec \alpha} \cdot (\RandPoly(\vec \alpha),\RandPoly(\revvec \alpha),\RandLDPoly_1(\vec \alpha),\ldots,\RandLDPoly_\NumVars(\vec \alpha)) = \\
            \Proof_\Sigma(\vec \alpha) - \SCPoly(\vec \alpha)\,\forall 
            \vec \alpha \in S
        \end{array}
        \right]~.
    \end{align*}
    This implies, by the definition of a $\TimeBound$-constraint detector, that under $D(\Proof_\Sigma)$, $\Proof_{\RandPoly}|_{S \cup S_{\rev}}$ and $\Proof_{\RandLDPoly_i}|_S$, $i \in [\NumVars]$, are distributed as a uniformly random solution to the linear system $\mathcal{L}_S$ given by
    \begin{align*}
        Z_{\RandPoly} \cdot \Proof_{\RandPoly}|_{S \cup S_{\rev}} &= 0 \\
        \forall i \in \NumVars,\, Z_{\RandLDPoly_i} \cdot \Proof_{\RandLDPoly_i}|_S &= 0 \\
        Z_{\Sigma} (\Proof_{\RandPoly}|_{S \cup S_{\rev}}, \Proof_{\RandLDPoly_i}|_S) &= \Proof_\Sigma|_S - \SCPoly|_S
    \end{align*}
    where $Z_{\RandPoly} \eqdef \CD_{\Degree}(S \cup S_{\rev})$, $Z_{\RandLDPoly} \eqdef \CD_{\vec \Degree_i}(S)$, and $Z_{\Sigma} \eqdef (z_{\Sigma,\vec \alpha})_{\vec \alpha \in S}$.

    Suppose, by induction, that $\OracleTable_{\RandPoly} \colon S \cup S_\rev \to \Field$, $\OracleTable_{\RandLDPoly_i} \colon S \to \Field$, $i \in [\NumVars]$, are distributed as a uniformly random solution to $\mathcal{L}_S$. Let $U \eqdef S \cup \{ \vec \alpha \}$. By construction, $\OracleTable_{\RandPoly} \cup \{ (\vec \alpha, \gamma_\RandPoly), (\revvec \alpha, \gamma_{\RandPoly_\rev}) \}$ and $\OracleTable_{\RandLDPoly_i} \cup \{ (\vec \alpha, \gamma_{\RandLDPoly_i}) \}$, $i \in [\NumVars]$, are distributed as a uniformly random solution to $\mathcal{L}_U$.

    It remains to show that, in the real world, $\Proof_\Sigma$ is distributed as $\ZKPCPEnc(F|_{\Subcube^\NumVars})$; i.e., a uniformly random element $w$ of $\SigmaRM$ such that $w(\vec a) = F(\vec a) + A(\vec a)$ for all $\vec a \in \Subcube^\NumVars$, for a uniformly random element $A$ of $\SigmaAntiSym$. By \cref{prop:antisym}, for random $\RandPoly$, $(\RandPoly - \RandPoly_\rev)|_{\Subcube^\NumVars}$ is distributed as $A$ since $\Subcube^\NumVars$ is a (subset of an) interpolating set.

    By the combinatorial nullstellensatz \cite{Alon99,ChenCGOS23}, the polynomial $Z = \sum_{i=1}^\NumVars \ZeroP_\Subcube \RandLDPoly_i$ is distributed as a uniformly random element of $\ZCode{\Subcube^{\NumVars}}{\ReedMuller[\Field,\NumVars,\Degree]}$. Observe that, for any degree-$\Degree$ extension $\hat{f}$ of a function $f \colon \Subcube^\NumVars \to \Field$, $\hat{f} + Z$ is a uniformly random degree-$\Degree$ extension of $f$. Hence, $\SCPoly + \MaskPoly = \SCPoly + \RandPoly - \RandPoly_\rev + Z$ is distributed as a uniformly random degree-$\Degree$ extension of $(\SCPoly + A)|_{\Subcube^\NumVars}$. Therefore $\Proof_{\Sigma} = \SumWord{\Subcube^\NumVars}{\SCPoly + \MaskPoly}$ is distributed as $w$.
\end{proof}

\appendix
\section{Deferred proofs}
\label{app:linear-codes}

\subsection{Proof of \cref{claim:unconstrained-equiv}}
\label{sec:unconstrained-equiv}

\UnconstrainedEquiv*

\begin{proof}
	First we show the forward implication. Assume $\Subdomain$ is unconstrained with respect to $\Code$. By definition, this means that any vector $\Constraint \in \Field^\Subdomain$ satisfying $\Constraint\cdot \Codeword|_{\Subdomain} = 0$ for all $\Codeword \in \Code$ must be identically zero. In other words, $(\Code|_{\Subdomain})^\perp = \{\vec{0}\}$. However, we know that $\dim(\Code|_{\Subdomain}) + \dim((\Code|_{\Subdomain})^\perp) = |\Subdomain|$. As $\dim((\Code|_{\Subdomain})^\perp) = 0$, it must be that $\dim(\Code|_{\Subdomain}) = |\Subdomain|$. Therefore $\Code|_{\Subdomain}$ is an $|\Subdomain|$-dimensional subspace of $\Field^{\Subdomain}$, meaning $\Code|_{\Subdomain}$ is isomorphic to $\Field^{\Subdomain}$. Thus, for each $x \in \Subdomain$ there exists a codeword $\Codeword_{x}|_{\Subdomain} \in \Code|_{\Subdomain}$ such that $\Codeword_{x}|_{\Subdomain}(x) = 1$ and $\Codeword_{x}|_{\Subdomain}(y) = 0$ for all $y \in \Subdomain\setminus\{x\}$.
	
	For the reverse implication, let $\Constraint\in \Field^\Subdomain$ be a constraint satisfying $\Constraint\cdot\Codeword|_{\Subdomain} = 0$ for all $\Codeword \in \Code$. We will show that $\Constraint$ is identically zero. By assumption, for each $x \in \Subdomain$, there exists a codeword $\Codeword_x \in \Code$ satisfying $\Codeword_x(x) = 1$ and $\Codeword_x(y) = 0$ for all $y \in \Subdomain\setminus\{x\}$. Substituting $\Codeword_x$ into our constraint equation for each $x \in \Subdomain$ give $\Constraint(x) = \Constraint\cdot\Codeword_x|_{\Subdomain} = 0$ for each $x \in \Subdomain$, that is, $\Constraint$ is identically zero. In other words, $\Subdomain$ is unconstrained with respect to $\Code$.
\end{proof}

\subsection{Proof of \cref{lem:general-zcode-constraint-implies-constraint}}
\label{sec:general-zcode-constraint-implies-constraint}
\GeneralZCodeConstraint*
\begin{proof}  
    First we show the forward implication. If $\Subdomain \Union \SubsetTuple$ is constrained with respect to $\Code$ then there exists a non-zero constraint $\Constraint \in \Field^{\Subdomain\Union\SubsetTuple}$ such that for every $\Codeword \in \Code$ we have 
    \begin{equation*}
        \sum_{x \in \Subdomain\Union\SubsetTuple} \Constraint(x)\Codeword(x) = 0\enspace.
    \end{equation*}
    As $\ZCode{\SubsetTuple}{\Code}$ is a subcode of $\Code$, we have that all $w' \in \ZCode{\SubsetTuple}{\Code}$ satisfy
    \begin{equation*}
        \sum_{x \in \Subdomain\Union\SubsetTuple} \Constraint(x)w'(x) =  \sum_{x \in \Subdomain} \Constraint(x)w'(x) = 0\enspace,
    \end{equation*} 
    where we have used the fact that $\Codeword'(x) = 0$ for all $x\in \SubsetTuple$. Lastly, there must exist $x \in \Subdomain\setminus \SubsetTuple$ for which $\Constraint(x) \neq 0$, for if not, then as $\Constraint$ is non-zero, $\Constraint|_{\SubsetTuple}$ is a non-zero constraint on $\SubsetTuple$, contradicting the hypothesis. Hence, $\Subdomain$ is constrained with respect to $\ZCode{\SubsetTuple}{\Code}$.

    For the reverse implication, we will employ the contrapositive: if $\Subdomain \Union \SearchSet$ is unconstrained with respect to $\Code$, then $\Subdomain\setminus\SubsetTuple$ is unconstrained with respect to $\ZCode{\SubsetTuple}{\Code}$. If $\Subdomain \Union \SearchSet$ is unconstrained with respect to $\Code$, then, by the forward implication of \cref{claim:unconstrained-equiv}, for each $x \in \Subdomain\Union\SubsetTuple$ there exists $\Codeword_x \in \Code$ satisfying \begin{inparaenum}[(i)]
        \item $\Codeword_x(x) = 1$; and
        \item $\Codeword_x(y) = 0$ for all $y \in \Subdomain\Union\SearchSet\setminus\{x\}$.
    \end{inparaenum}
    In particular, for every $x \in \Subdomain\setminus \SubsetTuple$, $\Codeword_x(y) = 0$ for all $y \in S$, so we have that for all $x \in \Subdomain\setminus\SubsetTuple$, $\Codeword_x \in \ZCode{\SubsetTuple}{\Code}$. Then the reverse implication of \cref{claim:unconstrained-equiv} implies that $\Subdomain\setminus \SubsetTuple$ is unconstrained with respect to $\ZCode{\SubsetTuple}{\Code}$.
\end{proof}

\newcommand{\InterpolatorAlg}{\text{Interpolate}}
\subsection{Proof of \cref{claim:interpolating-sets}}
\label{sec:interpolating-sets}
\InterpolatingSets*
First we give the construction, as follows.
\begin{mdframed}[nobreak=true]    
    \begin{construction}
        An algorithm to compute the interpolating set for a linear code $\Code$ given a constraint detector $\CD$ for $\Code$. Receives as input a subdomain $\Subdomain \subseteq \Field^\NumVars$.
        \ConstrucSpacing
        $\InterpolatorAlg(\Subdomain)$:
        \begin{enumerate}[nolistsep]
            \item Compute $Z \eqdef \CD(\Subdomain)$.
            \item Perform Gaussian elimination on $Z$ to obtain a matrix $Z'$ which is in reduced row echelon form.
            \item Output the set $\Intset \eqdef \{x \in \Subdomain : x~\text{is a free variable in}~Z'\}.$
        \end{enumerate}
    \end{construction}
\end{mdframed}

The efficiency is clear by construction. To see that this construction is correct, we will show that $\Intset$ is unconstrained and each $x \in \Subdomain\setminus\Intset$ is determined by $\Intset$. We then show that the second property in \cref{claim:interpolating-sets} is a consequence of these two properties.

To see that $\Intset$ is unconstrained, note that as each element $x \in \Intset$ is a free variable, for each choice of values $\vec{u} \in\Field^{\Intset}$ for the variables in $\Intset$ there will be some setting $\vec{v} \in \Field^{\Subdomain\setminus\Intset}$ of the remaining leading variables in $\Subdomain\setminus\Intset$ such that $(\vec{u}, \vec{v}) \in \ker(Z')$.  By the correctness of $\CD$, $(\vec{u}, \vec{v}) \in \ker(Z')$ if and only if $(\vec{u}, \vec{v}) \in \Code|_{\Subdomain}$. We have just demonstrated that $\Code|_{\Intset} = \Field^{\Intset}$, in other words, $\Intset$ is unconstrained. 

For the second part, as each $x \in \Subdomain\setminus\Intset$ is a leading variable, the row of $Z'$ in which it is leading forms a constraint 
$\Constraint\colon\Field^{\Subdomain} \to \Field$ with $\Constraint(x) = 1$, and $\Constraint(x') = 0$ for all $x' \in \Subdomain\setminus\Intset$ (as $Z'$ is in reduced row echelon form). Thus the restriction of this constraint to $\Intset \Union \{x\}$ is a constraint
$\Constraint\colon \Field^{\Intset \Union\{x\}} \to \Field$ with $\Constraint(x) = 1$, so $x$ is determined by $\Intset$.

For the third part, suppose there exists a constraint $\Constraint\colon\Subdomain\Union\SubsetTuple$ with respect to $\Code$ such that $\Constraint(s) \neq 0$ for some $s \in \SubsetTuple$. Denote $\Subdomain^* \eqdef \{x \in \Subdomain\setminus\Intset : z(x) \neq 0\}$. We assume that $\Subdomain^*$ is non-empty, for if it were not, then the restriction of $\Constraint$ to $\Intset\Union\SubsetTuple$ would suffice. Then we can write 
\begin{equation}
\label{eqn:constraint-int}
    \sum_{x \in \Intset}\Constraint(x)\Codeword(x) + \sum_{x^* \in \Subdomain^*}\Constraint(x^*)\Codeword(x^*) + \sum_{x \in \SubsetTuple}\Constraint(x)\Codeword(x) = 0
\end{equation}
for all $\Codeword \in \Code$. As we showed in the second part, each $x^* \in \Subdomain^*$ is determined by $\Intset$. In particular, for each $x^* \in \Subdomain^*$, we can write 
\begin{equation*}
    \Codeword(x^*) = \sum_{x \in \Intset} a_{x, x^*} \Codeword(x),
\end{equation*}
for some scalars $a_{x,x^*} \in \Field$. Substituting this expression into \cref{eqn:constraint-int} we obtain
\begin{equation*}
    \sum_{x \in \Intset}\Constraint(x)\Codeword(x) + \sum_{x^* \in \Subdomain^*}\Constraint(x^*) \left(\sum_{x \in \Subdomain} a_{x, x^*} \Codeword(x)\right) + \sum_{x \in \SubsetTuple}\Constraint(x)\Codeword(x) = 0,
\end{equation*}
which simplifies to
\begin{equation*}
    \sum_{x \in \Intset} \left(
    \Constraint(x) + \sum_{x^* \in \Subdomain^*}a_{x, x^*}\Constraint(x^*)
    \right)\Codeword(x) + \sum_{x \in \SubsetTuple}\Constraint(x)\Codeword(x) = 0.
\end{equation*}
Thus $\Constraint$ induces a constraint $\Constraint'\colon\Intset\Union\SubsetTuple\to\Field$ defined as follows:
\begin{equation*}
    \Constraint(x) \eqdef \begin{cases}
        \Constraint(x) + \sum_{x^* \in \Subdomain^*}a_{x, x^*}\Constraint(x^*) &\text{if}~x \in \Intset,\\
        \Constraint(x) &\text{otherwise.}
    \end{cases}
\end{equation*}
Observe that $\Constraint(s) = \Constraint'(s)$ for all $s \in \SubsetTuple$.

\subsection{Proof of \cref{lem:new-multilin-generalisation}}
\label{sec:new-multilin-generalisation}
\MultilinGen*
\begin{proof}
     By \cref{fact:lagrange-basis}, we can write any polynomial $p(\vec{X}) \in \Polys{\Field}{\vec{\Degree}}{\NumVars}$ as
    \begin{equation*}
        p(\vec{X}) = \sum_{\Point\in\SubsetTuple} a_{\Point}\LagrangePoly{\SubsetTuple, \Point}(\vec{X}),
    \end{equation*}
    for some scalars $a_{\Point} \in \Field$. Then the constraints $p(x) = 0$ for all $x \in \Subdomain$ can be written as an $|\Subdomain|\times|\SubsetTuple|$ matrix $A$ of linear equations over $\Field$, relating the $|\SubsetTuple|$ many coefficients $a_{\Point}$. By applying Gaussian elimination we can uniquely convert $A$ to a matrix $B$ which is in reduced row echelon form. As $\mathrm{rank}(B) = \mathrm{rank}(A) \leq |\Subdomain|$, $B$ will have at most $|\Subdomain|$ many ``leading ones''. Set $G_{\Subdomain}$ to be the set of all $\Point\in\SubsetTuple$ such that there is no leading one in the $\Point$-th column of $B$, i.e., for each $g \in G_\Subdomain$, the coefficient $a_g$ is a free variable with respect to the linear system. Thus, in particular, for any $g \in G_\Subdomain$, we can set $a_{g} = 1$ and $a_{g'} = 0$ for all $g' \in G_\Subdomain \setminus \{g\}$ and obtain a solution to the linear system, yielding the result.
\end{proof}

\subsection{Proof of \cref{claim:zcode-implies-padded-message}}
\label{sec:zcode-implies-padded-message}
\ZCodePadded*
\begin{proof}
    The reverse implication is straightforward: if $(\BigElt, \ExtElt) \in \Code|_{\BigSubdomain \sqcup \Subdomain}$, then $(\BigElt|_{\SmallSubdomain}, \ExtElt) \in \Code|_{\SmallSubdomain\sqcup \Subdomain}$ by taking the restriction to $\SmallSubdomain\sqcup\Subdomain$.

    For the forward implication, let $(\SmallElt, \ExtElt) \in \Code|_{\SmallSubdomain\sqcup \Subdomain}$. Then, as $\SmallSubdomain \subseteq \BigSubdomain$, trivially there exists $\BigElt \in \Code|_{\BigSubdomain}$ such that $\BigElt|_{\SmallSubdomain} = \SmallElt$ and $(\BigElt, \ExtElt) \in \Code|_{\BigSubdomain \sqcup I}$. Let $\BigElt' \in \Code|_{\BigSubdomain}$ be an arbitrary message which agrees with $\BigElt$ when restricted to $\SmallSubdomain$, that is, $\BigElt'|_{\SmallSubdomain} = \SmallElt$. Then for some $\ExtElt' \in \Code|_{\Subdomain}$, $(\BigElt', \ExtElt') \in \Code|_{\BigSubdomain \sqcup I}$. Restricting this vector to $\SmallSubdomain\sqcup\Subdomain$, we have $(\SmallElt, \ExtElt') \in \Code|_{\SmallSubdomain \sqcup\Subdomain}$. By linearity of $\Code|_{\SmallSubdomain \sqcup I}$, we have that $(\SmallElt, \ExtElt) - (\SmallElt, \ExtElt') = (\ZVec_{\SmallSubdomain}, \ExtElt-\ExtElt') \in \Code|_{\SmallSubdomain \sqcup I}$, where $\ZVec_{\SmallSubdomain} \in \Code|_\SmallSubdomain$ is the zero codeword. If we can show that $(\vec{0}_{\BigSubdomain}, \ExtElt - \ExtElt') \in \Code|_{\BigSubdomain\sqcup\Subdomain}$, where $\ZVec_{\BigSubdomain} \in \Code|_{\BigSubdomain}$ is the zero codeword. we will be done. Clearly $(\SmallElt, \ExtElt') \in \ZCode{\SmallSubdomain}{\Code}|_{\Subdomain}$, and by \cref{cor:z-sigrm-not-many-constraints}, we have that $\ZCode{\SmallSubdomain}{\Code} = \ZCode{\BigSubdomain}{\Code}$, so $(\SmallElt, \ExtElt') \in \ZCode{\BigSubdomain}{\Code}|_{\Subdomain}$. In other words, $(\ZVec_{\BigSubdomain}, \ExtElt-\ExtElt') \in \Code|_{\BigSubdomain\sqcup\Subdomain}$. Hence, by linearity again, $(\BigElt', \ExtElt') + (\ZVec_{\BigSubdomain}, \ExtElt-\ExtElt') = (\BigElt', \ExtElt) \in \Code|_{\BigSubdomain\sqcup\Subdomain}$, as required.
\end{proof}

\clearpage

\printbibliography

\end{document}